\documentclass[11pt]{article}
\pdfoutput=1 %for arxiv

\usepackage[utf8]{inputenc}
\usepackage[margin=1in]{geometry}
\usepackage{amsfonts,amsmath,amssymb,amsthm,mathtools,stmaryrd}
\usepackage{dsfont}

\usepackage{tocloft}
\usepackage[nottoc]{tocbibind} %Bibliography in ToC

\usepackage{enumitem}
\setlist{listparindent=\parindent,parsep=0pt,itemsep=1em}
\setlist[itemize]{label=$-$,noitemsep}
\setlist[enumerate]{itemsep=1mm}
\setlist[description]{leftmargin=\parindent}

\usepackage{bold-extra}

\usepackage{float}
\allowdisplaybreaks
\usepackage{microtype,xspace,bm}
\usepackage{doi}
\usepackage{hyperref}
\usepackage[dvipsnames]{xcolor}
\hypersetup{
	colorlinks=true,
	urlcolor=MidnightBlue,
	linkcolor=MidnightBlue,
	citecolor=MidnightBlue,
	unicode
}
\usepackage[capitalise,nameinlink,noabbrev]{cleveref}
\crefname{equation}{}{}

\usepackage{thmtools}
\usepackage{thm-restate}

\usepackage{tikz}
\usepackage{graphicx}
\usetikzlibrary{positioning,arrows.meta,math,shapes.geometric,decorations.pathmorphing,decorations.pathreplacing,calligraphy}
\newcommand{\spacex}{1.1}
\newcommand{\spacey}{1.1}
\newcommand{\labelspace}{1.5}

\newlength{\nodeline}
\newlength{\arrowline}
\definecolor{color_gadget_PHP}{RGB}{219, 48, 122}
\colorlet{color_gadget_PHP_inner}{color_gadget_PHP!10!white}
\colorlet{color_gadget_PHP_label}{color_gadget_PHP!70!black}

\definecolor{color_gadget_SOPL}{RGB}{255, 153, 0}
\colorlet{color_gadget_SOPL_inner}{color_gadget_SOPL!10!white}
\colorlet{color_gadget_SOPL_label}{color_gadget_SOPL!70!black}

\definecolor{color_gadget_PATHPHP}{RGB}{0, 153, 255}
\colorlet{color_gadget_PATHPHP_inner}{color_gadget_PATHPHP!10!white}
\colorlet{color_gadget_PATHPHP_label}{color_gadget_PATHPHP!70!black}

\definecolor{color_gadget_ITER}{RGB}{153, 204, 0}
\colorlet{color_gadget_ITER_inner}{color_gadget_ITER!10!white}
\colorlet{color_gadget_ITER_label}{color_gadget_ITER!70!black}

\definecolor{color_gadget_EOPL}{RGB}{0, 153, 255}
\colorlet{color_gadget_EOPL_inner}{color_gadget_EOPL!10!white}
\colorlet{color_gadget_EOPL_label}{color_gadget_EOPL!70!black}

\tikzstyle{node} = [circle, line width=\nodeline, draw = black, fill = white, inner sep = 0mm, minimum size = 3.5mm]
\tikzstyle{node_small} = [node, circle, line width = 0.25mm, minimum size = 2.5mm]
\tikzstyle{solution} = [fill=red!90!,draw=black!50!red]
\tikzstyle{side} = [fill=Goldenrod,draw=Brown]
\tikzstyle{node_text} = [] % no_style node, to debug

\tikzstyle{node_regular} = [node]
\tikzstyle{node_regular_small} = [node_small]

\tikzstyle{node_solution} = [node, solution]
\tikzstyle{node_green} = [node, fill=Green!20!LimeGreen,draw=black!80!Green]
\tikzstyle{node_notice} = [node, draw = Green!20!LimeGreen, line width=2pt, fill=none, minimum size = 5.93mm,dotted]
\tikzstyle{node_solution_small} = [node_small,solution]

\tikzstyle{node_a} = [node, side, rectangle, minimum size = 4.3mm]
\tikzstyle{node_a_solution} = [node_a, solution]
\tikzstyle{node_a_small} = [node_small, side, rectangle, minimum size = 2.15mm]
\tikzstyle{node_a_solution_small} = [node_a_small,solution]

\tikzstyle{node_b} = [node, side, diamond, minimum size = 6.2mm]
\tikzstyle{node_b_small} = [node_a_small, diamond, minimum size = 3.1mm]

\tikzstyle{node_notice_small} = [node_notice, line width=2*\nodeline, minimum size = 8mm]

\tikzstyle{naive} = [minimum size = 4.5mm,rectangle]
\tikzstyle{node_regular_intro} = [node,fill=Gray!10!white]

\tikzstyle{edge} = [-{Latex[round]}, line width=\arrowline]

\tikzstyle{edge_regular} = [edge]
\tikzstyle{edge_regular_small} = [-{Latex[round]}, line width = 0.25mm]

\tikzstyle{edge_php} = [edge, color=color_gadget_PHP!70!black]
\tikzstyle{edge_php_small} = [edge_php, edge_regular_small]

\tikzstyle{edge_eopl} = [edge, color=color_gadget_EOPL!70!black]
\tikzstyle{edge_eopl_small} = [edge_eopl, edge_regular_small]

\tikzstyle{edge_iter} = [edge, color=color_gadget_ITER!70!black]
\tikzstyle{edge_iter_small} = [edge_iter, edge_regular_small]

\tikzstyle{edge_pathphp_small} = [line width = 0.25mm, -{Latex[round]}, decorate, decoration={snake, segment length=2.5mm, amplitude=1mm, pre length=7pt,post length=8pt}, shorten < = 3pt, shorten >=3pt, color=color_gadget_PATHPHP_label]

\tikzstyle{gadget} = [rounded corners, line width = 0.4mm, dashed]

\tikzstyle{gadget_PHP} = [gadget, draw = color_gadget_PHP, fill=color_gadget_PHP_inner]
\tikzstyle{gadget_PHP_small} = [gadget_PHP, line width = 0.2mm]

\tikzstyle{gadget_SOPL} = [gadget, draw = color_gadget_SOPL, fill=color_gadget_SOPL_inner]
\tikzstyle{gadget_SOPL_small} = [gadget_SOPL, line width = 0.2mm]

\tikzstyle{gadget_PATHPHP} = [gadget, draw = color_gadget_PATHPHP, fill=color_gadget_PATHPHP_inner]
\tikzstyle{gadget_PATHPHP_small} = [gadget_PATHPHP, line width = 0.2mm]

\tikzstyle{gadget_ITER} = [gadget, draw = color_gadget_ITER, fill=color_gadget_ITER_inner]
\tikzstyle{gadget_ITER_small} = [gadget_ITER, line width = 0.2mm]

\tikzstyle{gadget_EOPL} = [gadget, draw = color_gadget_EOPL, fill=color_gadget_EOPL_inner]
\tikzstyle{gadget_EOPL_small} = [gadget_EOPL, line width = 0.2mm] % some global settings such as color etc.

\usepackage[font=small,labelfont=bf]{caption} %customise text style on fig captions
\usepackage{subcaption}

\usepackage{stackengine}

\usepackage{centernot}

\theoremstyle{definition}
\newtheorem{definition}{Definition}

\theoremstyle{plain}
\newtheorem{theorem}{Theorem}
\newtheorem{lemma}{Lemma}
\newtheorem{corollary}{Corollary}

\newtheorem{claim}{Claim}

\theoremstyle{remark}

\newtheorem{fact}{Fact}

\Crefname{claim}{Claim}{Claims}
\Crefname{definition}{Definition}{Definitions}

\newcommand{\newclass}[2]{\newcommand{#1}{{\text{\upshape\sffamily #2}}\xspace}}
\renewcommand{\P}{{\text{\upshape\sffamily P}}\xspace}
\newclass{\NP}{NP}
\newclass{\coNP}{coNP}
\newclass{\FP}{FP}
\newclass{\TFNP}{TFNP}
\newclass{\PLS}{PLS}
\newclass{\PPA}{PPA}
\newclass{\PPAD}{PPAD}
\newclass{\PPADS}{PPADS}
\newclass{\PPP}{PPP}
\newclass{\PWPP}{PWPP}
\newclass{\CLS}{CLS}
\newclass{\EOPL}{EOPL}
\newclass{\SOPL}{SOPL}
\newclass{\UEOPL}{UEOPL}
\newclass{\cA}{A}
\newclass{\cB}{B}
\newclass{\BPP}{BPP}
\newclass{\queryD}{D}
\newclass{\queryNP}{C}

\newcommand{\newprob}[2]{\newcommand{#1}{{\text{\upshape\scshape #2}}\xspace}}
\newprob{\eol}{EoL}
\newprob{\eolLong}{End-of-Line}
\newprob{\sol}{SoL}
\newprob{\solLong}{Sink-of-Line}
\newprob{\iter}{Iter}
\newprob{\sod}{SoD}
\newprob{\sodLong}{Sink-of-Dag}
\newprob{\kkt}{KKT}
\newprob{\pA}{A}
\newprob{\pB}{B}
\newprob{\grid}{Grid}
\newprob{\eopl}{EoPL}
\newprob{\eoplLong}{End-of-Potential-Line}
\newprob{\ueoplLong}{Unique-EoPL}
\newprob{\ueopl}{UEoPL}
\newprob{\eoml}{EoML}
\newprob{\eomlLong}{End-of-Metered-Line}
\newprob{\sopl}{SoPL}
\newprob{\soplLong}{Sink-of-Potential-Line}
\newprob{\pigeon}{Pigeon}
\newprob{\lonely}{Lonely}
\newprob{\reversiblepigeon}{RPigeon}
\newprob{\reversiblepigeonlong}{Reversible-Pigeon}
\newprob{\factoring}{Factoring}
\newprob{\nash}{Nash}
\newprob{\Or}{Or}

\newcommand{\SA}{\ensuremath{\text{SA}}}

\newcommand{\B}{\ensuremath{\{0,1\}}}
\newcommand{\T}{\ensuremath{\{0,1,*\}}}

\newcommand{\size}{\textup{size}}

\newcommand{\poly}{\textup{poly}}

\newcommand{\leaves}{\textup{leaves}}
\newcommand{\sols}{\textup{Sol}}
\newcommand{\nul}{\textup{\textsf{null}}}

\newcommand{\set}[1]{\ensuremath{\{#1\}}}
\newcommand{\st}{\medspace | \medspace}

\newcommand{\E}{{\mathbb{E}}}

\renewcommand{\gets}{\shortleftarrow}

\newcommand{\Z}{\mathbb{Z}}
\newcommand{\R}{\mathbb{R}}

\newcommand{\N}{\mathbb{N}}
\newcommand{\F}{\mathbb{F}}
\newcommand{\calR}{\mathcal{R}}
\newcommand{\calI}{\mathcal{I}}
\newcommand{\calF}{\mathcal{F}}

\begin{document}

\newgeometry{margin=1.3in,top=1.7in,bottom=1in}

\begin{center}
{\huge Separations in Proof Complexity and $\TFNP$}
\\[9mm] \large

\setlength\tabcolsep{1em}
\begin{tabular}{cccc}
Mika G\"o\"os&
Alexandros Hollender&
Siddhartha Jain&
Gilbert Maystre\\[-1mm]
\small\slshape EPFL &
\small\slshape University of Oxford &
\small\slshape UT Austin &
\small\slshape EPFL
\end{tabular}

\vspace{1mm}
\begin{tabular}{ccc}
William Pires&
Robert Robere&
Ran Tao\\[-1mm]
\small\slshape Columbia University &
\small\slshape McGill University &
\small\slshape Carnegie Mellon University
\end{tabular}

\vspace{9mm}

\vspace{5mm}

\normalsize
\end{center}

\noindent
{\bf Abstract.}~
It is well-known that Resolution proofs can be efficiently simulated by Sherali--Adams~(SA) proofs. We show, however, that any such simulation needs to exploit huge coefficients: Resolution cannot be efficiently simulated by SA when the coefficients are written in unary. We also show that \emph{Reversible Resolution} (a variant of MaxSAT Resolution) cannot be efficiently simulated by Nullstellensatz (NS).

These results have consequences for total $\NP$ search problems. First, we characterise the classes $\PPADS$, $\PPAD$, $\SOPL$ by unary-SA, unary-NS, and Reversible Resolution, respectively. Second, we show that, relative to an oracle, $\PLS\not\subseteq\PPP$, $\SOPL\not\subseteq\PPA$, and $\EOPL\not\subseteq\UEOPL$. In particular, together with prior work, this gives a complete picture of the black-box relationships between all classical $\TFNP$ classes introduced in the 1990s.

\vspace{15mm}

\setlength{\cftbeforesecskip}{0pt}
\renewcommand\cftsecfont{\mdseries}
\renewcommand{\cftsecpagefont}{\normalfont}
\renewcommand{\cftsecleader}{\cftdotfill{\cftdotsep}}
\setcounter{tocdepth}{1}
\tableofcontents

\thispagestyle{empty}
\setcounter{page}{0}

\newpage
\restoregeometry

\section{Separations in Proof Complexity} \label{sec:intro-proof}

The main results of this work are two separations between standard propositional proof systems, as summarised in~\cref{fig:systems}. Moreover, these results can be further interpreted as black-box separations in the theory of total~$\NP$ search problems ($\TFNP$), as we explain later in \cref{sec:intro-tfnp}. This connection between $\TFNP$ and proof complexity, which has proved fruitful in past works and which we further explore here, also yields a new type of result in proof complexity, which we call \emph{intersection theorems}; see \cref{sec:intro-intersection-thms}.

\subsection{Resolution vs.\ Sherali--Adams}
Our first separation is between the most basic and well-studied proof system \emph{Resolution} (see the textbooks~\cite{Jukna2012,Krajicek2019} for an introduction) and the semi-algebraic proof system \emph{Sherali--Adams}~\cite{Sherali1994,Dantchev2012} (see the monograph~\cite{Fleming2019} for an introduction). Let us briefly recall these systems. Each system aims to refute a given CNF contradiction (unsatisfiable CNF formula) $F\coloneqq C_1\land\ldots\land C_m$ over the~$n$ boolean variables $x=(x_1,\ldots,x_n)$.

\begin{description}
\item[Resolution (Res).]
A Resolution refutation of $F$ starts with the set of clauses of $F$ and repeatedly applies the \emph{resolution rule} $C\lor x_i, D\lor \bar{x}_i \vdash C\lor D$. That is, if we have already deduced premise clauses $C\lor x_i$ and $D\lor \bar{x}_i$ for some $i$, then we can further deduce the clause $C\lor D$. Once this rule has been applied enough times to produce the empty clause $\bot$, the refutation is complete. The \emph{size} of the refutation is the number of deduction steps, and its \emph{width} is the maximum width~$|C|$ (number of literals) of any clause $C$ appearing in the refutation.

\item[Sherali--Adams (SA).]
Sherali--Adams refutes unsatisfiable sets of polynomial equations $\{a_i(x)=0: i\in[m]\}$ with real coefficients, $a_i\in\R[x]$. A CNF contradiction~$F$ can be translated into this language by encoding each clause, say, $C\coloneqq (x_1\lor \overline{x}_2 \lor x_3)$, as the equation $(1-x_1)x_2(1-x_3)=0$, and by enforcing each variable $x_i$ to take boolean values by the equation $x_i^2-x_i=0$. An SA refutation of $\{a_i(x)=0\}$ is a polynomial identity of the form\footnote{This particular form is valid for refuting sets of polynomial equations, and can be easily obtained from more general forms used for refuting sets of polynomial inequalities.}
\begin{equation}\label{eq:sa}
\sum_{i\in[m]} p_i(x)\cdot a_i(x) ~=~ 1 + J(x),
\end{equation}
where $p_i\in\R[x]$ are polynomials and $J$ is a \emph{conical junta}: a nonnegative linear combination of terms, that is, $J(x)=\sum_j \alpha_j \cdot t_j(x)$ where $\alpha_j\in\R_{\geq0}$ are nonnegative coefficients and each~$t_j$ is a conjunction of literals; for example, $t_j(x)=x_1\overline{x}_2x_3 = x_1(1-x_2)x_3$.
The \emph{size} of the refutation is the combined total number of monomials in $p_i$, $a_i$, and $t_j$ (viewed as a polynomial) and its \emph{degree} is the maximum of $\deg(p_i)+\deg(a_i)$ and of $\deg(t_j)$ over all $i,j$.
\end{description}
It is a basic fact that SA is strictly more powerful than Resolution. First, Resolution is \emph{$p$-simulated} by SA, that is, with only polynomial overhead in proof width/degree and size. Indeed, if~$F$ can be refuted by width-$w$ Resolution, then~$F$ can be refuted by SA in degree $w+1$~\cite{Dantchev2009}. Moreover, if one allows \emph{twin variables} in SA, the simulation can also be made efficient relative to size~\cite{Atserias2016}. Second, SA is not $p$-simulated by Resolution: there are $n$-variate CNF contradictions~$F$ (e.g., graph pigeonhole principles) that can be refuted by SA in constant degree but such that any Resolution refutation of $F$ requires width~$\Omega(n)$ and size $\exp(\Omega(n))$~\cite{Atserias2019}.

Our first result highlights a previously overlooked inefficiency in the way that SA simulates Resolution. We show that any low-degree simulation needs to exploit huge coefficients.
\begin{restatable}{theorem}{ResVuSA}
\label{thm:res-sa}
There are $n$-variate CNF formulas $F$ that can be refuted by constant-width Resolution, but such that any  SA refutation of $F$ in degree $n^{o(1)}$ requires coefficients of magnitude $\exp(n^{\Omega(1)})$.
\end{restatable}

\cref{thm:res-sa} is qualitatively tight in that the singly-exponential lower bound~$\exp(n^{\Omega(1)})$ cannot be improved much. Namely, if a CNF formula can be refuted by a degree-$d$ SA proof, then there also exist a degree-$d$ SA proof with integer coefficients of magnitude $\exp(n^{O(d)})$ (see~\cref{sec:coeff-ub}). To our knowledge, \cref{thm:res-sa} is the first exponential coefficient lower bound for a constant-width CNF formula in any semi-algebraic proof system. (Examples of systems of polynomial equations---not coming from CNFs---requiring even doubly-exponential coefficients were known previously~\cite{ODonnell2017,Raghavendra2017,Hakoniemi2021}.)

We also note that the conclusion of \cref{thm:res-sa} can be slightly strenthened using standard \emph{lifting/xorification} techniques~\cite[\S4]{BenSasson2009} to show that any SA refutation of $F$ must either use exponentially many monomials or exponentially large coefficients. This trade-off has consequences for the \emph{unary Sherali--Adams}~(uSA) system where we restrict the coefficients to be integers written in \emph{unary} and where their magnitude counts towards proof size (more precisely, the size of a uSA proof is the sum of the magnitudes of all coefficients appearing in the proof). Thus we conclude that Resolution is not $p$-simulated by uSA. In particular, this answers a question raised in a concurrent work by Bonacina and Bonet~\cite{Bonacina2022a}. For comparison, proving a similar lower bound for the Cutting Planes system (separating CP from unary-CP) is a long-standing open problem.

\begin{figure}[t]
\centering
\begin{tikzpicture}[scale=2]
\tikzset{inner sep=0,outer sep=0.5}

% Set bounding box
%\node (a) at (-7.1,0) {\mbox{}};
%\node (b) at (7.1,0) {\mbox{}};

\tikzstyle{a}=[inner sep=4pt,fill=white!93!Cerulean,rounded corners=3pt]

\begin{scope}[yscale=1]
\large
\node[a] (TreeRes) at (0,0) {TreeRes};
\node[a] (RevRes) at (-1,1) {RevRes};
\node[a] (Res) at (-2,2) {Res};
\node[a] (uNS) at (1,1) {uNS};
\node[a] (uSA) at (0,2) {uSA};
\node[a] (NS) at (2,2) {$\Z$-NS};
\node[a] (FNS) at (2,3) {$\F$-NS};
\node[a] (SA) at (0,3) {SA};
\end{scope}

\path[-{Stealth[length=6pt]},line width=0.7pt,gray]
(TreeRes) edge (RevRes)
(TreeRes) edge (uNS)
(RevRes) edge (Res)
(RevRes) edge (uSA)
(uNS) edge (uSA)
(NS) edge (SA)
(uNS) edge (NS)
(NS) edge (FNS)
(uSA) edge (SA)
(Res) edge (SA);

\small
\hypersetup{hidelinks}
\tikzset{new/.style={-{Stealth[length=6pt]},dashed,line width=1pt,YellowOrange}}
\draw[new,bend right=7]
(Res) edge
node[midway,below,inner sep=4pt,sloped] {\cref{thm:res-sa}}
(uSA);
\draw[new,bend right=20]
(RevRes) edge
node[pos=0.17,below,sloped,inner sep=3pt] {\cref{thm:revres-ns}}
(FNS);

\end{tikzpicture}
\newcommand{\sA}{\mathrm{A}}
\newcommand{\sB}{\mathrm{B}}
\vspace{1em}
\caption{Our new separations of proof systems. An arrow $\sA\rightarrow \sB$ means that $\sA$ is $p$-simulated by $\sB$, that is, with polynomial overhead in width/degree and size (when allowing twin variables). A dashed arrow~$\sA \dashrightarrow \sB$ means that $\sA$ is not $p$-simulated by $\sB$.}
\label{fig:systems}
\end{figure}
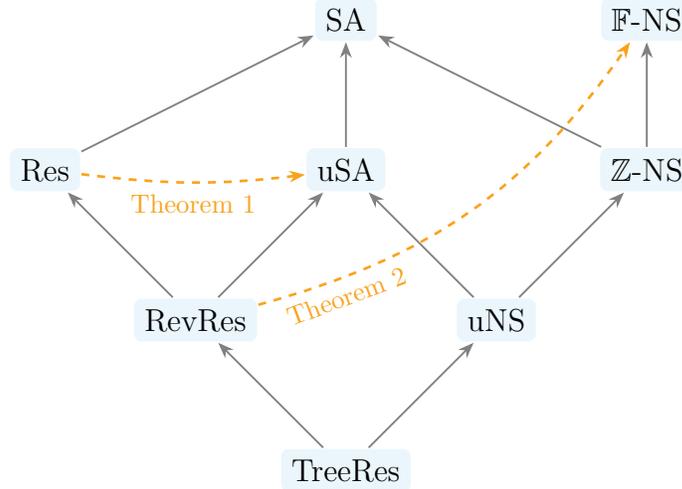

\subsection{Reversible Resolution vs.\ Nullstellensatz}

Our second separation is between the standard algebraic proof system Nullstellensatz~\cite{Beame1994} and a subsystem of Resolution that we call \emph{Reversible Resolution}. The latter is closely related to fragments of Resolution that have been introduced to model the reasoning used by \emph{MaxSAT} solvers (which find an assignment that satisfies as many clauses as possible). Prior work has defined several distinct such \emph{MaxSAT Resolution} systems~\cite{Bonet2007,Larrosa2008,Filmus2023}. Our variant is yet slightly different (see~\cref{sec:char-revres} for a comparison to prior systems). Ultimately, our definition is motivated by results that will be discussed in \cref{sec:intro-tfnp}: Reversible Resolution captures an important~\TFNP class, and, moreover, it equals the ``intersection'' of Resolution and uSA.

\begin{description}
\item[Reversible Resolution (RevRes).]
In this restricted fragment of Resolution we only allow the \emph{symmetric resolution rule} $C\lor x_i, C\lor \bar{x}_i \vdash C$ and its inverse $C \vdash C\lor x_i, C\lor \bar{x}_i$. Moreover, we stipulate that an application of either rule \emph{consumes its premises} in the following sense. The refutation begins with a \emph{multiset} of clauses of $F$---we may choose the multiplicity of each clause freely at start---and a single application of a deduction rule removes a single occurrence of each premise clause from the multiset and then adds the concluded clauses back to the multiset. Once we produce at least one empty clause, the refutation is complete. The \emph{size} and \emph{width} of the refutation are defined as before.
\item[Nullstellensatz ($\F$-NS).]
Let $\F$ be a field. An $\F$-Nullstellensatz refutation of a set of polynomial equations $\{a_i(x)=0:i\in[m]\}$ over $\F$ is given by a set of polynomials $\{p_i(x)\}\subseteq \F[x]$ such that
\begin{equation} \label{eq:ns}
\sum_{i\in[m]} p_i(x)\cdot a_i(x) ~=~ 1.
\end{equation}
The \emph{size} of the refutation is the combined total number of monomials in $p_i$ and $a_i$ and its \emph{degree} is the maximum of $\deg(p_i)+\deg(a_i)$ over all $i$.
\end{description}
Reversible Resolution is $p$-simulated by uSA. Indeed, the usual simulations of Resolution by SA~\cite{Dantchev2009,Atserias2016} have the neat property that if they are applied to a RevRes proof instead, the resulting coefficients become bounded by the size of the RevRes proof (see also~\cite{Filmus2023} for a simulation in a closely related MaxSAT system). This also means that RevRes is strictly less powerful than Resolution, as per our first separation result.

It is a classic result that Resolution is not $p$-simulated by $\F$-NS over any field $\F$. This is witnessed by~CNF formulas expressing the \emph{sink-of-dag} ($\sod$) principle~\cite{Clegg1996,Buss1998} or the \emph{pebbling} principle~\cite{BureshOppenheim2002,Rezende2019}. Our second result strengthens these classical separations showing that RevRes cannot be simulated by low-degree $\F$-NS.

\begin{restatable}{theorem}{RevResVNS}
\label{thm:revres-ns}
There are $n$-variate CNF formulas $F$ that can be refuted by constant-width polynomial-size RevRes, but such that any $\F$-NS refutation (over any $\F$) of $F$ requires degree $n^{\Omega(1)}$.
\end{restatable}

Again, we note that standard lifting techniques can be used to strengthen the degree lower bound in \cref{thm:revres-ns} to an exponential size lower bound. We conclude that RevRes is not $p$-simulated by $\F$-NS. In particular, this strengthens a previous result by Filmus et.\ al.~\cite{Filmus2023} who showed that RevRes (actually, their closely related MaxSAT system) is not $p$-simulated by tree-like Resolution.

\subsection{Techniques}

Our separation between Resolution and uSA (\cref{sec:proof-res-sa}) builds on the separation between RevRes and $\F$-NS (\cref{sec:proof-revres-ns}). We prove the latter separation for $\F=\R$ in a particularly robust form, namely we show that it holds even if we allow some small amount of ``error'' in the NS proof. We introduce what we call \emph{$\epsilon$-approximate Nullstellensatz} ($\epsilon$-NS) refutations where we relax the polynomial identity~\cref{eq:ns} over $\F=\R$ to hold only approximately:
\begin{equation}
\sum_{i\in[m]} p_i(x)\cdot a_i(x) ~=~ 1\pm\epsilon, \quad\qquad\forall x\in\{0,1\}^n.
\end{equation}
In the above expression and for the remainder of the article, ``$= 1\pm \epsilon$'' stands for ``$\in [1 - \epsilon, 1 + \epsilon]$'', meaning that the LHS is a polynomial that takes values in $[1 - \epsilon, 1 + \epsilon]$ when evaluated on boolean inputs. For example, an SA refutation where $J(x) \leq \epsilon$ for all boolean inputs $x$ is also an $\epsilon$-NS refutation (since $J(x) \geq 0$ trivially holds).

We show that there is no low-degree approximate NS proof for the formulas that encode the so called \emph{sink-of-potential-line} ($\sopl$) principle. These formulas are easy for RevRes, and in fact, we later show they are \emph{complete} for RevRes (see \cref{thm:characterisations}). Naturally, our lower-bound proof borrows techniques from polynomial approximation theory. We give a randomised \emph{decision-to-search reduction}, showing how a low-degree $\epsilon$-NS refutation of~$\sopl$ would imply a low-degree approximating polynomial for the $\Or$ function. It is well-known, however, that the $n$-bit $\Or$ requires large approximate polynomial degree, namely $\Omega(\sqrt{n})$. This proof idea is inspired by previous works~\cite{Raz1992,Huynh2012,Goos2018cbs,Itsykson2021} that followed a similar strategy in the context of communication complexity: they studied randomised reductions from set-disjointness (communication analogue of~$\Or$) to various communication search problems. Finally, we also give a separate (non-robust) proof that $\sopl$ is hard for $\F$-NS over any field $\F$ using the intersection theorem (\cref{thm:intersection-thm}).

Our lower bound for $\epsilon$-NS, say with $\epsilon\coloneqq 1/2$, now helps us prove \cref{thm:res-sa}. We consider an~SA refutation~\cref{eq:sa} of the $\sod$ principle (which is a stronger principle than $\sopl$). The non-existence of a low-degree $\epsilon$-NS refutation for $\sopl$ immediately implies that in any SA refutation of~$\sod$, the conical junta~$J$ has to assume a value at least $\epsilon$ on some input: the~RHS equals~$1+J(x)\geq 1+\epsilon = 1.5$ for some $x$. Our idea is to now \emph{iterate} the~$\epsilon$-NS lower-bound argument by combining several $\sopl$ instances inside $\sod$ with the aim of finding large values on the RHS. After~$i$ iterations, we show the~RHS equals $1+J(x_i)\geq 1.5^{\Omega(i)}$ for some carefully constructed input~$x_i$ that embeds $i$ copies of~$\sopl$. Setting $i=\poly(n)$ concludes the proof.

\section{Separations in \texorpdfstring{$\TFNP$}{TFNP}} \label{sec:intro-tfnp}

A major motivation for our proof complexity separations in \cref{sec:intro-proof} is that they have consequences in terms of black-box separations between subclasses of \TFNP. Together with prior work, our new separations resolve all the black-box relationships between classes depicted in~\cref{fig:classes}. To explain this connection in detail, we start with a short introduction to \TFNP.

\subsection{Introduction to \TFNP}

The class \TFNP consists of all \emph{total \NP search problems}, that is, search problems where a solution is guaranteed to exist, and where it can be efficiently checked whether a given candidate solution is feasible. Some very important problems lie in \TFNP, for example, \factoring (given a number, compute a prime factor) or \nash (given a bimatrix game, compute a Nash equilibrium).

A crucial observation is that no \TFNP problem can be \NP-hard, unless $\NP = \coNP$~\cite{Megiddo1991}. Furthermore, it is believed that \TFNP is unlikely to have complete problems~\cite{Pudlak15-Herbrand}. As a result, in order to understand the complexity of important \TFNP problems, researchers have defined syntactic subclasses of \TFNP, such as \PLS~\cite{Johnson1988}, \PPAD, \PPADS, \PPA, \PPP~\cite{Papadimitriou1994}. These subclasses are defined using canonical complete problems that correspond to very simple existence principles.
\begin{itemize}[leftmargin=4em]
\item[\PLS:] Every directed acyclic graph has a sink.
\item[\PPAD:] Every directed graph with an unbalanced node (outdegree $\neq$ indegree) must have another unbalanced node.
\item[\PPADS:] Every directed graph with a positively unbalanced node (outdegree $>$ indegree) must have a negatively unbalanced node (outdegree $<$ indegree).
\item[\PPA:] Every undirected graph with an odd-degree node must have another odd-degree node.
\item[\PPP:] Every function mapping $[n+1]$ to $[n]$ must have a collision. \emph{(Pigeonhole Principle)}
\end{itemize}
These existence principles naturally give rise to corresponding total search problems. For example, for \PPAD that would be: given a directed graph and an unbalanced node in that graph, find another unbalanced node. These problems are defined so that the search space (the set of nodes) has size exponential in the size of the input. Otherwise, it would be trivial to find a solution in polynomial time. In more detail, this is achieved by having the input of the problem consist of a boolean circuit that can be used to compute the neighbours of any given node.

The theory of $\TFNP$ classes has been successful in capturing the complexity of many important natural problems. Indeed, in a celebrated result~\cite{Daskalakis2009,ChenDT09-Nash}, it was shown that \nash is complete for \PPAD. Following this breakthrough, other problems from game theory~\cite{DengQS12-cake,ChenDO15-anonymous-games,Mehta2018-constant-rank} and economics~\cite{CodenottiSVY08-economies-games,ChenDDT09-Arrow-Debreu,ChenPY17-non-monotone-markets} were also proved \PPAD-complete. Similarly, \PLS has been found to capture the complexity of various interesting problems, mainly ones where a local optimum of some sort is sought~\cite{Krentel89-TSP,Krentel90-weighted-CNF,Schaeffer91-local-search,FabrikantPT04-pure-Nash}. Finally, various problems in fair division are \PPA-complete~\cite{FRG22-NS-CH-ham}, while some problems related to cryptography have been shown \PPP-complete~\cite{SotirakiZZ18-PPP}.

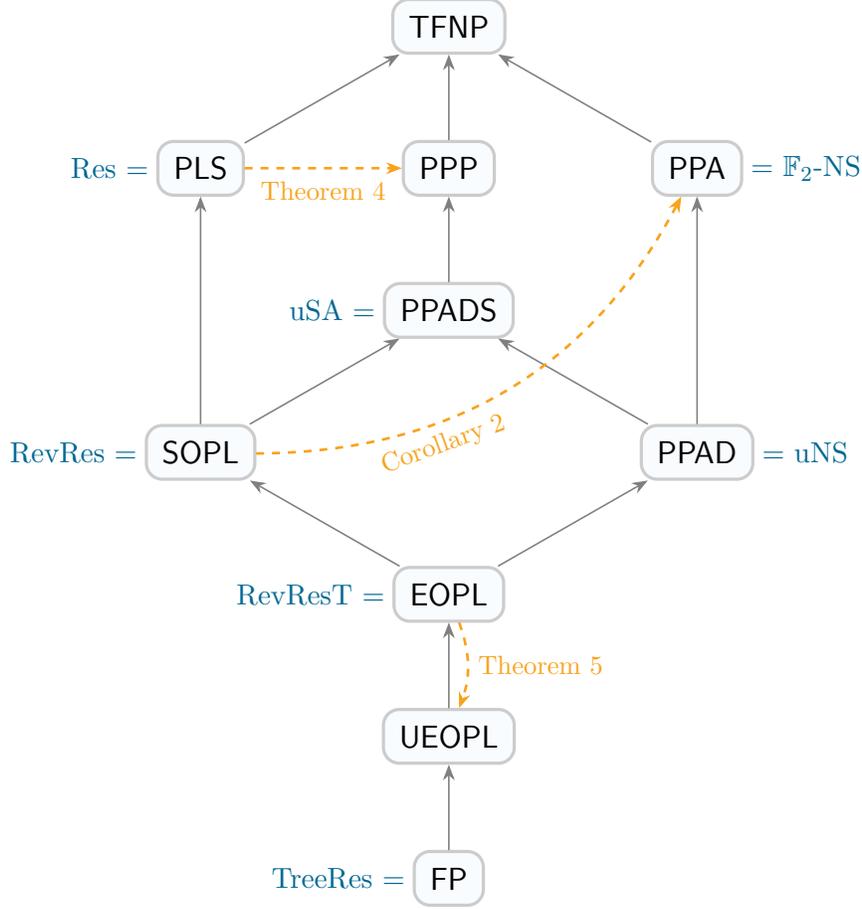
\begin{figure}[t]
\centering
\begin{tikzpicture}[scale=1.1]
\tikzset{inner sep=0,outer sep=3}

% Set bounding box
%\node (a) at (-7.1,0) {\mbox{}};
%\node (b) at (7.1,0) {\mbox{}};

\tikzstyle{a}=[inner sep=6pt, inner ysep=6pt,outer sep=0.5pt,
draw=black!20!white, fill=Cerulean!2!white, very thick, rounded corners=6pt, align=center]

\begin{scope}[yscale=1.145]
\large
\node[a] (FP) at (0,-1.5) {$\FP$};
\node[a] (UEOPL) at (0,0) {$\UEOPL$};
\node[a] (EOPL) at (0,1.5) {$\EOPL$};
\node[a] (SOPL) at (-3,3) {$\SOPL$};
\node[a] (PPAD) at (3,3) {$\PPAD$};
\node[a] (PPADS) at (0,4.5) {$\PPADS$};
\node[a] (PLS) at (-3,6) {$\PLS$};
\node[a] (PPP) at (0,6) {$\PPP$};
\node[a] (PPA) at (3,6) {$\PPA$};
\node[a] (TFNP) at (0,7.5) {$\TFNP$};
\end{scope}

\path[-{Stealth[length=6pt]},line width=.6pt,gray]
(FP) edge (UEOPL)
(UEOPL) edge (EOPL)
(EOPL) edge (SOPL)
(EOPL) edge (PPAD)
(SOPL) edge (PLS)
(SOPL) edge (PPADS)
(PPAD) edge (PPADS)
(PPAD) edge (PPA)
(PPADS) edge (PPP)
(PLS) edge (TFNP)
(PPP) edge (TFNP)
(PPA) edge (TFNP);

\small
\hypersetup{hidelinks}
\tikzset{new/.style={-{Stealth[length=6pt]},dashed,line width=1pt,YellowOrange}}
%\tikzset{old/.style={-{Stealth[length=6pt]},dashed,line width=0.6pt,Gray}}
\draw[new]
(PLS) edge
node[midway,below,inner sep=2pt]{\cref{thm:pls-ppp}}
(PPP);
\draw[new,bend right=31]
(SOPL) edge
node[pos=0.34,below,sloped,inner sep=1pt] {\cref{cor:sopl-ppa}}
(PPA);
\draw[new,bend left=20]
(EOPL) edge
node[right,inner sep=1pt] {\cref{thm:eopl-ueopl}}
(UEOPL);

%\draw[old]
%(PPA) edge
%node[midway,below,inner sep=2pt]{\cite{Beame1998}}
%(PPP);
%\draw[old,bend left=20]
%(UEOPL) edge
%node[right,inner sep=1pt] {\cite{Hubacek2020}}
%(FP);
%\draw[old,bend left=31]
%(PPAD) edge
%node[pos=0.2,below,sloped,inner sep=1pt] {\cite{Morioka2001}}
%(PLS);

\normalsize
\tikzset{model/.style={color=MidnightBlue}}

\node [model, left=0 of FP] {TreeRes $=$};
\node [model, left=0 of PLS] {Res $=$};
\node [model, right=0 of PPA] {$=$ $\mathbb{F}_2$-NS};
\node [model, left=0 of PPADS] {uSA $=$};
\node [model, right=0 of PPAD] {$=$ uNS};
\node [model, left=0 of SOPL] {RevRes $=$};
\def\stackalignment{r}
\node [model, left=0 of EOPL] {RevResT $=$};

\end{tikzpicture}
\vspace{5mm}
\caption{Class inclusion diagram for $\TFNP$. An arrow $\cA\rightarrow\cB$ means $\cA\subseteq \cB$ relative to all oracles. A~dashed arrow~$\cA\dashrightarrow\cB$ means $\cA\not\subseteq \cB$ relative to some oracle. We have only drawn new separations proved in this paper. Together with prior oracle separations~\cite{Beame1998,Morioka2001,Buresh2004}, this resolves all black-box relationships between the classes featured in the diagram. In the black-box model, some classes can be captured using propositional proof systems, as indicated in blue.}
\label{fig:classes}
\end{figure}

\subsubsection*{New classes and collapses}

More recently, newer classes \CLS~\cite{Daskalakis2011}, \EOPL~\cite{Fearnley2020,Hubacek2020}, \SOPL~\cite{Goos2018} were defined, motivated chiefly by problems that were unlikely to be complete for any of the classical classes discussed above. Indeed, it was noted that many interesting problems lie in both \PLS and~\PPAD, but are unlikely to be complete for \PLS $\cap$ \PPAD, a seemingly completely artificial class. To remedy this situation, \CLS, and later \EOPL, were defined as more natural subclasses of \PLS $\cap$ \PPAD. However, in a surprising turn of events, it was discovered that $\CLS = \PLS \cap \PPAD$~\cite{Fearnley2021} and also that $\EOPL = \PLS \cap \PPAD$ and $\SOPL = \PLS \cap \PPADS$~\cite{GoosHJMPRT22-collapses}. In other words, the new classes can be completely defined in terms of the classical ones.

In order to rule out further surprising collapses in the future, it would thus makes sense, whenever one defines a new subclass, to also provide some kind of evidence that the new class is indeed new, and does not collapse to existing classes. Clearly, any unconditional separation is completely out of reach, since it would immediately imply that $\P \neq \NP$. However, it turns out that one can indeed prove separations relative to oracles by proving unconditional separations between black-box versions of the classes.

\subsubsection*{The black-box model}

Recall that \TFNP subclasses are defined in terms of very simple existence principles that are turned into (white-box) total search problems by having the input be implicitly described by a boolean circuit. Another---sometimes more natural---choice is to have the input be described by a black box, instead of a white box. For example, in the case of \PPAD, instead of being given the description of a circuit that can be used to compute neighbours, we can consider the model where we can query an oracle (black-box) to ask for the neighbours of a node.

More formally, a total \emph{query} search problem is a sequence of relations $R_n \subseteq \{0,1\}^n \times O_n$, one for each size $n \in \N$, such that for all inputs $x \in \{0,1\}^n$ there is an output~$o \in O_n$ such that $(x,o) \in R_n$. Here $O_n$ is a finite set of outputs and we say that $o$ is a solution to instance $x$, when~$(x,o) \in R_n$. We think of an instance $x \in \{0,1\}^n$ as a very long bitstring that can only be accessed through queries to individual bits. In this context, an efficient algorithm is a deterministic algorithm that, for any~$x \in \{0,1\}^n$, finds a solution $o$ to $x$ by performing a small number of queries to $x$, namely at most $\poly(\log n)$ queries. Thus, efficient algorithms correspond to decision trees (with leaves labelled by elements of $O_n$) of depth at most $\poly(\log n)$. Note that this model is non-uniform: the problem admits an efficient algorithm, if for each $n \in \N$, there exists a shallow decision tree solving~$R_n$.

The notion of total search problems as defined above does not quite correspond to \TFNP yet, because it is missing the requirement for efficient verification of solutions. We enforce this in the following natural way. A total search problem $\textsc{R} = (R_n)_n$ is in $\TFNP^{dt}$, if for each $o \in O_n$ there is a decision tree $T_o$ with depth $\poly(\log n)$ such that for every $x \in \{0,1\}^n$, $T_o(x) = 1$ if and only if~$(x,o) \in R_n$. We define the class~$\PPAD^{dt}$ as the set of all $\TFNP^{dt}$ problems that have an efficient decision-tree reduction to (the query version of) the canonical complete problem for \PPAD. We denote by $\PPAD^{dt}(R_n)$ the decision tree complexity of a reduction from $R_n$ to the canonical~$\PPAD^{dt}$-complete problem (see \cref{sec:definitions} for a precise definition). Thus, problem $\textsc{R}=(R_n)_n$ lies in~$\PPAD^{dt}$ if and only if~$\PPAD^{dt}(R_n) = \poly(\log n)$. The decision-tree analogues of the other classes are defined in the same way.

\subsubsection*{Black-box separations}
In the black-box model, it is now possible to prove \emph{unconditional} separations, e.g., that $\PPAD^{dt} \not\subseteq \PLS^{dt}$ by showing that there is no shallow decision-tree reduction from some problem in $\PPAD^{dt}$ to a complete problem for $\PLS^{dt}$. Importantly, a black-box separation also provides some evidence that the separation might hold in the white-box setting too, in the following sense: \emph{any black-box separation implies a corresponding separation in the white-box model relative to some oracle}~\cite{Beame1998}. Moreover, all existing containment results (including the recent collapses~\cite{Fearnley2021,GoosHJMPRT22-collapses}) also hold in the black-box setting. Thus, a black-box separation is quite significant, since it rules out any collapse using existing techniques.

Previously, Beame et al.~\cite{Beame1998} proved all possible separations between the classes $\PPA^{dt}$, $\PPAD^{dt}$, $\PPADS^{dt}$, $\PPP^{dt}$. Subsequently, Morioka~\cite{Morioka2001} extended these results by proving that $\PPAD^{dt}$ is not reducible to $\PLS^{dt}$. This implies that none of $\PPA^{dt}$, $\PPAD^{dt}$, $\PPADS^{dt}$ and $\PPP^{dt}$ are contained in $\PLS^{dt}$. Buresh-Oppenheim and Morioka~\cite{Buresh2004} further proved that $\PLS^{dt}$ is not contained in $\PPA^{dt}$. It has so far remained open whether $\PLS^{dt} \subseteq \PPADS^{dt}$ or $\PLS^{dt} \subseteq \PPP^{dt}$.

\subsubsection*{Connection to proof complexity}
Propositional proof complexity is a major tool for proving black-box separations. There is a natural correspondence between total query search problems and CNF contradictions. In one direction, a CNF contradiction $F \coloneqq C_1 \land \dots \land C_m$ over the variables $x=(x_1, \dots, x_n)$ naturally gives rise to a corresponding total search problem $S(F)$: given an assignment $x\in\{0,1\}^n$, find an unsatisfied clause of $F$. Formally, we define $S(F) \subseteq \{0,1\}^n \times [m]$ by $(x,i) \in S(F)$ if and only if~$C_i(x) = 0$. Thus, a sequence of unsatisfiable CNF formulas $\textsc{F} = (F_n)$, where $F_n$ has $n$ variables, defines the total search problem $S(\textsc{F}) = (S(F_n))$. Note that $S(\textsc{F}) \in \TFNP^{dt}$ if $F_n$ has width $\poly(\log n)$.

In the other direction, a problem $\textsc{R} = (R_n)$ in $\TFNP^{dt}$ can be written equivalently as $S(\textsc{F})$ for some sequence of CNF contradictions $\textsc{F}=(F_n)$. Specifically, for $R_n \subseteq \{0,1\}^n \times O_n$ we define the formula $F_n\coloneqq \bigwedge_{o \in O_n} \neg T_o(x)$, where we note that $T_o(x)$ can naturally be written as a DNF formula of width at most $\poly(\log n)$ (with one term per accepting leaf of $T_o$), and thus $\neg T_o(x)$ can be written as a CNF formula of the same width.

\subsection{New characterisations}
The above connection to proof complexity opens up the possibility to characterise search problem classes by propositional proof systems, in the following sense: \emph{the problem $(S(F_n))_n$ lies in class~$X$ if and only if the CNF formulas $(F_n)_n$ have small refutations in proof system $Y$}. To make this more precise, for any proof system $\textsc{P}$ and a CNF formula $F$, we define
\[
\textsc{P}(F) ~\coloneqq~ \min_{\text{$\textsc{P}$-proof $\Pi$ of $F$}}
\big[\log\size(\Pi) + \deg(\Pi)\big].
\]
Here, $\deg(\Pi)$ should be understood as \emph{width} when $\text{P}$ is Resolution (or RevRes) and as \emph{depth} when~$\text{P}$ is tree-like Resolution. Prior work has established the following characterisations.
\begin{itemize}
    \item $\FP^{dt}(S(F)) = \Theta(\textup{TreeRes}(F))$~\cite{Lovasz1995}.
	 \item $\PLS^{dt}(S(F)) = \Theta(\textup{Res}(F))$~\cite{BussKT2014}.
    \item $\PPA^{dt}(S(F)) = \Theta(\mathbb{F}_2\textup{-NS}(F))$~\cite{Goos2018}.
    \item $\PPA_p^{dt}(S(F)) = \Theta(\mathbb{F}_p\textup{-NS}(F))$ for every prime $p$~\cite{Kamath2020}.
\end{itemize}

We contribute the following new characterisations. For one of them, we need to introduce one more proof system, \emph{Reversible Resolution with Terminals} (RevResT), defined in~\cref{sec:char-revres}.
\begin{restatable}{theorem}{Chars} \label{thm:characterisations}
For any unsatisfiable CNF formula $F$, we have:
\begin{itemize}
    \item $\PPAD^{dt}(S(F)) = \Theta(\textup{uNS}(F))$.
    \item $\PPADS^{dt}(S(F)) = \Theta(\textup{uSA}(F))$.
    \item $\SOPL^{dt}(S(F)) = \Theta(\textup{RevRes}(F))$.
    \item $\EOPL^{dt}(S(F)) = \Theta(\textup{RevResT}(F))$.
\end{itemize}
\end{restatable}

Together with our proof complexity separations from \cref{sec:intro-proof}, we immediately obtain the following black-box separations (which yield white-box oracle separations as discussed above).
\begin{samepage}
\begin{corollary} \label{cor:pls-ppads}
$\PLS^{dt} \not\subseteq \PPADS^{dt}$.
\end{corollary}%
\begin{corollary} \label{cor:sopl-ppa}
$\SOPL^{dt} \not\subseteq \PPA^{dt}$.
\end{corollary}
\end{samepage}

Additional characterizations, as well as separation results, were obtained in the subsequent works~\cite{Hubacek2024, Li2024}.

\subsection{Two further separations}

We show two more black-box separations involving classes $\PPP^{dt}$ and $\UEOPL^{dt}$, which currently lack elegant proof system characterisations. The first separation strengthens \cref{cor:pls-ppads}. 
\begin{restatable}{theorem}{PLSPPP} \label{thm:pls-ppp}
$\PLS^{dt} \not\subseteq \PPP^{dt}$.
\end{restatable}
\vspace{-3mm}
\begin{restatable}{theorem}{EOPLUEOPL} \label{thm:eopl-ueopl}
$\EOPL^{dt} \not\subseteq \UEOPL^{dt}$.
\end{restatable}

(An early preprint of this work did not include the above theorems. In an independent work, Bonacina and Thapen~\cite{Bonacina2022} also proved \cref{thm:pls-ppp}, deriving it from \cref{cor:pls-ppads} using essentially the same proof as we do.)

\cref{thm:pls-ppp} settles the last open oracle separation question between the five original~$\TFNP$ classes introduced in~\cite{Johnson1988,Papadimitriou1994}. This question was re-asked recently by Daskalakis in his Nevanlinna Prize lecture~\cite[Open Question 6]{Daskalakis2019}. Previously, Buresh-Oppenheim and Morioka~\cite{Buresh2004} showed a partial result in the direction of \cref{thm:pls-ppp}, namely, that there is no reduction from $\PLS^{dt}$ to~$\PPP^{dt}$ that preserves the number of solutions in each instance. Finally, \cref{thm:eopl-ueopl} answers a question of~\cite{Fearnley2020} who introduced the class $\UEOPL$. They conjectured that $\EOPL\not\subseteq\UEOPL$ and asked whether this could be shown relative to an oracle.

\subsection{Intersection theorems in proof complexity}\label{sec:intro-intersection-thms}

Our new characterisations can be combined with the collapses $\SOPL = \PLS \cap \PPADS$ and $\EOPL = \PLS \cap \PPAD$~\cite{GoosHJMPRT22-collapses} (which hold in the black-box model) to produce completely new types of results in propositional proof complexity that we call \emph{intersection theorems}.

Stated plainly, the first of these results says that a CNF formula $F$ admits an efficient (small degree and size) Reversible Resolution refutation \emph{if and only if} if it admits an efficient Resolution refutation \emph{and} an efficient unary Sherali--Adams refutation.
In other words, Reversible Resolution is the ``intersection'' of Resolution and unary Sherali--Adams. We can similarly show that Reversible Resolution with Terminals is the ``intersection'' of Resolution and unary Nullstellensatz.

\begin{restatable}{theorem}{Intersection} \label{thm:intersection-thm}
For any unsatisfiable CNF formula $F$, we have:
\begin{itemize}
\item $\textup{RevRes}(F) = \Theta(\textup{Res}(F) + \textup{uSA}(F))$.
\item $\textup{RevResT}(F) = \Theta(\textup{Res}(F) + \textup{uNS}(F))$.
\end{itemize}
\end{restatable}

To our knowledge, these are the first theorems of their type, that is, showing that efficient proofs exist in one system $\textsc{P}_0$ \emph{if and only if} efficient proofs exist in two other systems $\textsc{P}_1$ and $\textsc{P}_2$. This is all the more striking given that all three of these proof systems are quite natural, being motivated from boolean logic and SAT-solving ($\textup{Res}$), linear programming ($\textup{uSA}$), and MaxSAT solving ($\textup{RevRes}$).
Moreover, the proof of this theorem (\cref{sec:intersection}) crucially uses \emph{both} perspectives of proof systems and total search problems.
Starting with propositional proofs in Resolution and unary Sherali--Adams, we convert them to efficient formulations of $S(F)$ in $\PLS^{dt}$ and $\PPADS^{dt}$, respectively. 
We then apply the collapse theorem to argue there is an efficient formulation of~$S(F)$ in $\SOPL^{dt}$, which we can finally convert back to a $\textup{RevRes}$ proof.
We see no apparent way to prove this theorem directly using classic proof complexity techniques.

\subsection{Open problems} \label{sec:open-problems}

In our opinion, exploring the interplay between $\TFNP$ and propositional proof complexity holds untapped potential. The results in this work arose from our core belief that \emph{a natural concept introduced in one theory should have a natural counterpart in another theory.} This philosophy suggests many further directions for research and serves as a guiding principle for formulating new beautiful connections between the two theories. For example:
\begin{enumerate}[label=(\arabic*)]
\item Can \cref{thm:res-sa} be strengthened to show that the Sum-of-Squares system needs huge coefficients to simulate Resolution in low degree?
\item Can we characterise the class $\PPP$ by a proof system?
\item Does unary-NS $p$-simulate $\Z$-NS for refuting CNF formulas?
\item Can we prove other intersection theorems in propositional proof complexity?
\item Do Sum-of-Squares and Polynomial Calculus characterise some $\TFNP$ classes?
\item Are there communication complexity analogues of our results? The recent column~\cite{Rezende2022} surveys the connections between total search problems and characterisations of various circuit models in the language of communication complexity (via Karchmer--Wigderson games).
\end{enumerate}
We note here that Buss, Fleming and Impagliazzo~\cite{BussFI2022} have recently provided an answer to question (5) by giving a $\TFNP$ characterization of Polynomial Calculus. In fact, they show a more general connection: every well-behaved proof system which can prove its own soundness is characterized by a $\TFNP$ problem, and vice-versa. This also answers question (2), although ideally we would like to characterize $\PPP$ by a more natural proof system than the one obtained through this generic connection.

\section{Definitions} \label{sec:definitions}

In this section we give formal definitions of the total search problems that we consider in this work.
We emphasise that unlike the standard uniform setting of $\TFNP$, we will be interested in non-uniform variants of $\TFNP$ classes defined by decision trees.

\subsection{Decision tree $\TFNP$}

\begin{definition}
	A \emph{total (query) search problem} is a sequence of relations $\textsc{R} = \set{R_n \subseteq \B^{n} \times O_n}$, where $O_n$ are finite sets, such that for all $x \in \B^n$ there is an $o \in O_n$ such that $(x, o) \in R_n$.
	A total search problem $\textsc{R}$ is in $\TFNP^{dt}$ if for each $o \in O_n$ there is a decision tree $T_o$ with depth $\poly(\log n)$ such that for every $x \in \B^n$, $T_o(x) = 1$ iff $(x, o) \in \textsc{R}$.
\end{definition}

While total search problems are formally defined as sequences $\textsc{R} = (R_n)$, it will often make sense to speak of an individual search problem $R_n$ in the sequence.
We will therefore slightly abuse notation and also call $R_n$ a total search problem.
It will also be convenient to encode total search problems with inputs and outputs chosen from domains other than $\B^n$. One common example will be total search problems where the inputs are chosen from $[n]^n$.
We can simulate this simply by encoding all elements of the non-boolean domain in binary in the usual way.
In all examples in this paper, performing this encoding will change the complexities of the involved problems by no more than a $O(\log n)$ factor.
We also allow the $n$-th problem $R_n$ in a sequence to have $\poly(n)$ input bits (instead of $n$) for notational convenience.

The canonical examples of total search problems in $\TFNP^{dt}$ are the search problems associated with an unsatisfiable CNF formula $F$.

\begin{definition}
	For any unsatisfiable CNF formula $F \coloneqq C_1 \land \cdots \land C_m$ over $n$ variables, define $S(F) \subseteq \B^n \times [m]$ by $(x, i) \in S(F)$ if and only if $C_i(x) = 0$.
\end{definition}

Therefore, given any sequence of unsatisfiable CNF formulas $\textsc{F} = \set{F_1, F_2, \ldots}$ we get a total search problem $S(\textsc{F}) = \set{S(F_1), S(F_2), \dots}$ in the natural way.
Observe that $S(\textsc{F}) \in \TFNP^{dt}$ if each unsatisfiable CNF formula has width $\poly(\log n)$. Conversely, these examples are also \emph{complete}, in the sense that any search problem in $\TFNP^{dt}$ can be re-encoded as unsatisfiable CNF formulas.
\begin{definition}
	For any total search problem $R\subseteq\{0,1\}^n\times O$ with solution verifiers $T_o$, $o\in O$, its encoding as an unsatisfiable CNF formula is given by $F\coloneqq \bigwedge_{o \in O} \neg T_o(x)$ where we think of $\neg T_o(x)$ written as a CNF formula (of width determined by the decision tree depth of $T_o$).
\end{definition}

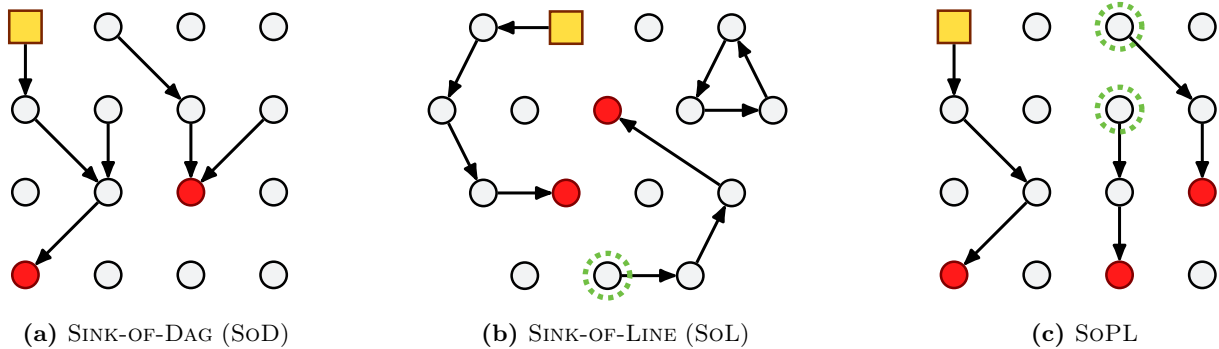
\begin{figure}[t]
\centering
\begin{subfigure}[b]{0.25\textwidth}
    \centering
    \begin{tikzpicture}[y=-1cm, scale=1]
% node grid
\coordinate (p11) at (\labelspace, \labelspace); \coordinate (p12) at (\labelspace + \spacex, \labelspace); \coordinate (p13) at (\labelspace + 2*\spacex, \labelspace); \coordinate (p14) at (\labelspace + 3*\spacex, \labelspace); 
\coordinate (p21) at (\labelspace, \labelspace + \spacey); \coordinate (p22) at (\labelspace + \spacex, \labelspace + \spacey);	 \coordinate (p23) at (\labelspace + 2*\spacex, \labelspace + \spacey); \coordinate (p24) at (\labelspace + 3*\spacex, \labelspace + \spacey); 
\coordinate (p31) at (\labelspace, \labelspace + 2*\spacey); \coordinate (p32) at (\labelspace + \spacex, \labelspace + 2*\spacey); \coordinate (p33) at (\labelspace + 2*\spacex, \labelspace + 2*\spacey); \coordinate (p34) at (\labelspace + 3*\spacex, \labelspace + 2*\spacey); 
\coordinate (p41) at (\labelspace, \labelspace + 3*\spacey); \coordinate (p42) at (\labelspace + \spacex, \labelspace + 3*\spacey); \coordinate (p43) at (\labelspace + 2*\spacex, \labelspace + 3*\spacey); \coordinate (p44) at (\labelspace + 3*\spacex, \labelspace + 3*\spacey); 

% node drawing
\tikzstyle{node_regular} = [node_regular_intro]

\node[node_a]        (P11) at (p11) {};
\node[node_regular]  (P12) at (p12) {};
\node[node_regular]  (P13) at (p13) {};
\node[node_regular] (P14) at (p14) {};
\node[node_regular]  (P21) at (p21) {};
\node[node_regular]  (P22) at (p22) {};
\node[node_regular]  (P23) at (p23) {};
\node[node_regular]  (P24) at (p24) {};
\node[node_regular]  (P31) at (p31) {};
\node[node_regular]  (P32) at (p32) {};
\node[node_solution] (P33) at (p33) {};
\node[node_regular]  (P34) at (p34) {};
\node[node_solution]  (P41) at (p41) {};
\node[node_regular]  (P42) at (p42) {};
\node[node_regular]  (P43) at (p43) {};
\node[node_regular]  (P44) at (p44) {};

% edges drawing
\draw[edge_regular] (P11) -- (P21);
\draw[edge_regular] (P21) -- (P32);
\draw[edge_regular] (P22) -- (P32);
\draw[edge_regular] (P32) -- (P41);
\draw[edge_regular] (P12) -- (P23);
\draw[edge_regular] (P23) -- (P33);
\draw[edge_regular] (P24) -- (P33);

\end{tikzpicture}
    \vspace{2mm}
    \caption{$\sodLong$ ($\sod$)}
    \label{figure:intro_sod}
\end{subfigure}%
\begin{subfigure}[b]{0.49\textwidth}
    \centering
    \begin{tikzpicture}[y=-1cm, scale=1]
% node grid
\coordinate (p11) at (0, 0); 				 \coordinate (p12) at (\spacex, 0); 				 \coordinate (p13) at (2*\spacex, 0); 				\coordinate (p14) at (3*\spacex, 0);
\coordinate (p20) at (-\spacex/2, \spacey); \coordinate (p21) at (\spacex/2, \spacey);	 \coordinate (p22) at (3*\spacex/2, \spacey); 	 \coordinate (p23) at (5*\spacex/2, \spacey); 	\coordinate (p24) at (7*\spacex/2, \spacey);
\coordinate (p31) at (0, 2*\spacey); \coordinate (p32) at (\spacex, 2*\spacey); \coordinate (p33) at (2*\spacex, 2*\spacey); \coordinate (p34) at (3*\spacex, 2*\spacey);
\coordinate (p40) at (-\spacex/2, 3*\spacey); \coordinate (p41) at (\spacex/2, 3*\spacey); \coordinate (p42) at (3*\spacex/2, 3*\spacey); \coordinate (p43) at (5*\spacex/2, 3*\spacey); \coordinate (p44) at (7*\spacex/2, 3*\spacey);

% node drawing
\tikzstyle{node_regular} = [node_regular_intro]

\node[node_regular] (P11) at (p11) {};
\node[node_a]       (P12) at (p12) {};
\node[node_regular] (P13) at (p13) {};
\node[node_regular] (P14) at (p14) {};
\node[node_regular] (P20) at (p20) {};
\node[node_regular] (P21) at (p21) {};
\node[node_solution] (P22) at (p22) {};
\node[node_regular] (P23) at (p23) {};
\node[node_regular] (P24) at (p24) {};
\node[node_regular] (P31) at (p31) {};
\node[node_solution] (P32) at (p32) {};
\node[node_regular] (P33) at (p33) {};
\node[node_regular] (P34) at (p34) {};
\node[node_regular] (P41) at (p41) {};
\node[node_regular] (P42) at (p42) {};
\node[node_regular] (P43) at (p43) {};

% edges drawing
\draw[edge_regular] (P12) -- (P11);
\draw[edge_regular] (P11) -- (P20);
\draw[edge_regular] (P20) -- (P31);
\draw[edge_regular] (P31) -- (P32);

\draw[edge_regular] (P42) -- (P43);
\draw[edge_regular] (P43) -- (P34);
\draw[edge_regular] (P34) -- (P22);

\draw[edge_regular] (P24) -- (P14);
\draw[edge_regular] (P14) -- (P23);
\draw[edge_regular] (P23) -- (P24);

% notice node
%\node[node_notice] at (p20) {};
\node[node_notice] at (p42) {};
\end{tikzpicture}
    \vspace{0.5mm}
    \caption{$\solLong$ ($\sol$)}
    \label{figure:intro_sol}
\end{subfigure}
\begin{subfigure}[b]{0.25\textwidth}
    \centering
    \begin{tikzpicture}[y=-1cm, scale=1]
% node grid
\coordinate (p11) at (0, 0); 				 \coordinate (p12) at (\spacex, 0); 				 \coordinate (p13) at (2*\spacex, 0); 				\coordinate (p14) at (3*\spacex, 0);
\coordinate (p21) at (0, \spacey);	 \coordinate (p22) at (\spacex, \spacey); 	 \coordinate (p23) at (2*\spacex, \spacey); 	\coordinate (p24) at (3*\spacex, \spacey);
\coordinate (p31) at (0, 2*\spacey); \coordinate (p32) at (\spacex, 2*\spacey); \coordinate (p33) at (2*\spacex, 2*\spacey); \coordinate (p34) at (3*\spacex, 2*\spacey);
\coordinate (p41) at (0, 3*\spacey); \coordinate (p42) at (\spacex, 3*\spacey); \coordinate (p43) at (2*\spacex, 3*\spacey); \coordinate (p44) at (3*\spacex, 3*\spacey);

% node drawing
\tikzstyle{node_regular} = [node_regular_intro]

\node[node_a]       (P11) at (p11) {};
\node[node_regular] (P12) at (p12) {};
\node[node_regular] (P13) at (p13) {};
\node[node_regular] (P14) at (p14) {};
\node[node_regular] (P21) at (p21) {};
\node[node_regular] (P22) at (p22) {};
\node[node_regular] (P23) at (p23) {};
\node[node_regular] (P24) at (p24) {};
\node[node_regular] (P31) at (p31) {};
\node[node_regular] (P32) at (p32) {};
\node[node_regular] (P33) at (p33) {};
\node[node_solution] (P34) at (p34) {};
\node[node_solution] (P41) at (p41) {};
\node[node_regular] (P42) at (p42) {};
\node[node_solution] (P43) at (p43) {};
\node[node_regular] (P44) at (p44) {};

% edges drawing
\draw[edge_regular] (P11) -- (P21);
\draw[edge_regular] (P21) -- (P32);
\draw[edge_regular] (P32) -- (P41);

\draw[edge_regular] (P23) -- (P33);
\draw[edge_regular] (P33) -- (P43);

\draw[edge_regular] (P13) -- (P24);
\draw[edge_regular] (P24) -- (P34);

% notice node
\node[node_notice] at (p13) {};
\node[node_notice] at (p23) {};
\end{tikzpicture}
    \vspace{2mm}
    \caption{$\sopl$}
    \label{figure:intro_sopl}
\end{subfigure}%
\caption{Examples of total search problems. The distinguished source node is drawn as a yellow square. Red nodes are associated with solutions. (For visual clarity, we highlight the actual sink nodes for $\sod$ rather than their predecessors.) Nodes circled in green would be solutions for $\eol$ and $\eopl$, respectively.}
\label{figure:intro_problems}
\end{figure}

\subsection{Search problem zoo} \label{sec:zoo}
We now define several search problems that will be of interest to us. See also \cref{figure:intro_problems} for helpful illustrations of some of them. We start with the problems that are complete for the classical classes introduced in~\cite{Johnson1988,Papadimitriou1994}.

\begin{description}
\item[$\PPP$: $\pigeon$ ($\pigeon_n$).]
This problem features $n$ pigeons, denoted by $[n]$, and as input we are given, for each pigeon $u\in[n]$ a hole $s_u \in [n-1]$. The goal is to output
\begin{enumerate}
\item $u,v \in [n]$, if $u \neq v$ and $s_u = s_v$. \hfill \emph{(pigeon collision)}
\end{enumerate}

\item[$\PPADS$: $\solLong$ ($\sol_n$).]
This problem is defined on a set of $n$ nodes, denoted by $[n]$, where the node $1$ is ``distinguished''.
For input, we are given a successor $s_u \in [n]$ for each node $u \in [n]$ and a predecessor $p_u \in [n]$ for each node $u \neq 1$. 
Given this list of successor/predecessor pointers we create a directed graph $G$ where we add an edge $(u, v)$ if and only if $s_u = v$ and $p_v = u$. We say $u$ is a \emph{proper sink} if it has in-degree $1$ and out-degree $0$, and it is a \emph{proper source} if it has in-degree $0$ and out-degree $1$.
The goal of the search problem is to output any of the following
\begin{enumerate}
	\item $1$, if $1$ is not a proper source node in $G$, or \hfill \emph{(no distinguished source)}
	\item $i \neq 1$, if $i$ is a proper sink node in $G$. \hfill \emph{(proper sink)}
\end{enumerate}

\item[$\PPAD$: $\eolLong$ ($\eol_n$).]
Same as $\sol$, except we add the following feasible solution.
\begin{enumerate}[resume]
	\item $i \neq 1$, if $i$ is a proper source node in $G$.
	\hfill \emph{(proper source)}
\end{enumerate}

\item[$\PLS$: $\sodLong$ ($\sod_n$).]
This problem is defined on the $[n] \times [n]$ grid, where the node $(1,1)$ is ``distinguished''.
As input, for each grid node $u=(i,j) \in [n] \times [n]$, we are given a \emph{successor} $s_u \in [n] \cup \set{\nul}$, interpreted as naming a node $(i+1,s_u)$ on the next row.
We say a node $u$ is \emph{active} if $s_u \neq \nul$, otherwise it is \emph{inactive}. A node $u$ is a \emph{proper sink} if $u$ is inactive but some active node has $u$ as a successor.
The goal of the search problem is to output any of the following
\begin{enumerate}
	\item \label{sod1} $(1,1)$, if $(1,1)$ is inactive \hfill \emph{(inactive distinguished source)}
	\item \label{sod2} $(n, j)$, if $(n,j)$ is active, \hfill \emph{(active sink)}
	\item \label{sod3} $(i, j)$ for $i \leq n-1$, if $(i,j)$ is active and its successor is a proper sink. \hfill \emph{(proper sink)}
\end{enumerate}
\end{description}
For $\sod$, it is helpful to think of the successors $s_u$ as describing a fan-out $1$ dag on an $n \times n$ grid of nodes such that all edges are between adjacent rows.
Active nodes are those nodes which have some edge leaving them.
Then, if we require that $(1, 1)$ is active and all nodes on row $n$ are inactive, the goal is to find a \emph{proper sink}, that is, an active node with an inactive successor node.

We next define complete problems for the more modern classes introduced in~\cite{Hubacek2020,Fearnley2020,Goos2018}. They are variations of the $\sod$ problem where all nodes in the grid have predecessor pointers and we only add an edge if the successor and predecessor pointers agree.
In particular, this implies that every node has fan-out \emph{and} fan-in $1$.

\begin{description}
\item[$\SOPL$: $\soplLong$ ($\sopl_n$).]
	As input we are given a \emph{successor} $s_{u} \in [n] \cup \set{\nul}$ for each $u \in [n] \times [n]$ and a \emph{predecessor} $p_u \in [n] \cup \set{\nul}$ for each $u \in \set{2, \dots, n} \times [n]$.
	A node $(i,j) \in [n-1] \times [n]$ is \emph{active} if $s_{(i,j)} = k \neq \nul$ and $p_{(i+1,k)} = j$, otherwise it is \emph{inactive}; a node $(i,j) \in \set{n} \times [n]$ is active if $s_{(i,j)} \neq \nul$ and inactive otherwise.
	A node $u$ is a \emph{proper sink} if $u$ is inactive but some active node has $u$ as a successor.
	The goal is to output any of the following
	\begin{enumerate}
	\item $(1,1)$, if $(1,1)$ is inactive, \hfill \emph{(inactive distinguished source)}
	\item $(n, j)$, if $(n, j)$ is active, \hfill \emph{(active sink)}
	\item $(i, j)$, if $(i,j)$ is a proper sink. \hfill \emph{(proper sink)}
\end{enumerate}
\item[$\EOPL$: $\eoplLong$ ($\eopl_n$).]
Add the following feasible solutions to $\sopl$.
	A node $(i,j)$ is a \emph{proper source} if $(i,j)$ is active and, either, $i = 1$ or $1 < i < n$ and there is no active node with $(i, j)$ as a successor.
	\begin{enumerate}[resume]
		\item $(i,j)$, if $(i, j)\neq(1,1)$ and $(i,j)$ is a proper source. \hfill \emph{(proper source)}
	\end{enumerate}
\item[$\UEOPL$: $\ueoplLong$ ($\ueopl_n$).]
Add the following feasible solution to $\eopl$.
\begin{enumerate}[resume]
	\item $(i, j)$ and $(i, j')$, if $j \neq j'$ and both nodes are active. \hfill \emph{(two parallel lines)}
\end{enumerate}
\end{description}

\subsection{Reductions and formulations}
Given any problem defined above we can consider complexity classes of total search problems obtained by taking reductions to these problems.
In this work we are particularly interested in the case where the reduction is defined by a low-depth decision tree.

\begin{definition}
	\label{def:formulation}
	Let $R \subseteq \B^n \times O$ and $S \subseteq \B^m \times O'$ be total search problems.
	An \emph{$S$-formulation of $R$} is a decision-tree reduction $(f_i, g_o)_{i \in [m], o \in O'}$ from $R$ to $S$.
	Formally, for each $i \in [m]$ and $o \in O'$ there are functions $f_i\colon \B^n \rightarrow \B$ and $g_o\colon \B^n \rightarrow O$ such that
	\[ (x, g_o(x)) \in R \impliedby (f(x), o) \in S\]
	where $f(x) \in \B^m$ is the string whose $i$-th bit is $f_i(x)$.
	The \emph{depth} of the reduction is
	\[
	d ~\coloneqq~ \max\big( \set{D(f_i) : i \in [m]} \cup \set{D(g_o) : o \in O'}\big),
	\]
	where $D(h)$ denotes the decision-tree depth of $h$.
	The \emph{size} of the reduction is $m$, the number of input bits to $S$. The \emph{complexity} of the reduction is $\log m + d$.
	We write $S^{dt}(R)$ to denote the minimum complexity of an $S$-formulation of $R$.

	We extend these notations to sequences in the natural way.
	If $R$ is a single search problem and~$\textsc{S} = (S_m)$ is a sequence of search problems, then we denote by $\textsc{S}^{dt}(R)$ the minimum of $S^{dt}_m(R)$ over all $m$.
	If $\textsc{R} = (R_n)$ is also a sequence, then we denote by $\textsc{S}^{dt}(\textsc{R})$ the function $n\mapsto \textsc{S}^{dt}(R_n)$.
\end{definition}
Using the previous definition we can now define complexity classes of total search problems via reductions.
For total search problems $\textsc{R} = (R_n), \textsc{S} = (S_n)$, we write
\[
\textsc{S}^{dt} ~\coloneqq~ \set{\textsc{R} : \textsc{S}^{dt}(\textsc{R}) = \poly(\log n)}.
\] 
We can now define the decision-tree variants of the standard classes: $\PPP^{dt} = \pigeon^{dt}$, $\PPADS^{dt} = \sol^{dt}$, and so on, according to the problems defined in \cref{sec:zoo}.

\section{Reversible Resolution vs.\ Nullstellensatz}
\label{sec:proof-revres-ns}

In this section we prove \cref{thm:revres-ns}, restated below.
\RevResVNS*

We prove \cref{thm:revres-ns} in two ways. First (\cref{sec:apx-ns,sec:ns-proof,sec:ub-revres}), we give a particularly robust proof in the special case $\F=\R$, which will be useful in \cref{sec:proof-res-sa} when we prove our other separation result. Second (\cref{sec:lb-fns}), we give a (non-robust) proof for all $\F$ using the intersection theorem. In both proofs we consider the $\sopl$ principle and show that it does not admit a low-degree NS proof, and that it can be refuted in low-width small-size RevRes.

\subsection{Approximate Nullstellensatz} \label{sec:apx-ns}

We define a generalisation of $\R$-NS that we call \emph{$\epsilon$-approximate Nullstellensatz} ($\epsilon$-NS) where $\epsilon\in(0,1)$ is an error parameter. An $\epsilon$-NS refutation of a set of real polynomial equations $\{a_i(x)=0:i\in[m]\}$ is a set of polynomials $\{p_i(x)\}$ such that
\begin{equation} \label{eq:apx-ns}
\sum_{i\in[m]} p_i(x)\cdot a_i(x) ~=~ 1\pm\epsilon, \quad\qquad\forall x\in\{0,1\}^n
\end{equation}
where we recall that ``$= 1\pm \epsilon$'' stands for ``$\in [1 - \epsilon, 1 + \epsilon]$'', meaning that the LHS is a polynomial that takes values in $[1 - \epsilon, 1 + \epsilon]$ when evaluated on boolean inputs.
The $\epsilon$-NS system is not a standard proof system in the sense of Cook and Reckhow~\cite{Cook1979}. In particular, it is not hard to show (using the PCP theorem) that testing the condition in \cref{eq:apx-ns} is in fact $\coNP$-complete. Another feature of the new system is that the error parameter can be efficiently reduced using standard error reduction techniques for polynomial approximation. For example, if we compose any $\epsilon$-NS proof $\sum_ip_ia_i=1\pm\epsilon$ with the univariate polynomial $q(z)\coloneqq z(2-z)$, we obtain an $\epsilon^2$-NS proof $q(\sum_ip_ia_i)=1\pm\epsilon^2$.

\subsection{Lower bound for $\epsilon$-NS} \label{sec:ns-proof}
Recall that the input to $\sopl_n$ consists of successor pointers $s_u\in[n]\cup\{\nul\}$ and predecessor pointers $p_u\in[n]\cup\{\nul\}$ for each grid node $u\in [n]\times[n]$. For the purposes of NS, we encode this input in binary by a string~$y\in\{0,1\}^{n'}$ over $n' =O(n^2\log n)$ variables. Moreover, we can think of~$\sopl_n$ as an unsatisfiable set of polynomial equations $\{a_i(y)=0\}$ each of degree $O(\log n)$. These equations can be obtained by taking the unsatisfiable CNF encoding of $\sopl_n$ (\cref{def:SOPL-CNF}) and encoding each clause as the corresponding polynomial equation in the usual way.

Our goal is to prove the following lemma.
\begin{lemma} \label{lem:apx-ns}
Every $\frac{1}{2}$-NS refutation of $\sopl_n$ requires degree $n^{\Omega(1)}$.
\end{lemma}

It suffices to prove the lemma for error $\epsilon\coloneqq 0.01$, because of efficient error reduction. Fix any~$\epsilon$-NS refutation $\sum_i p_i(y) a_i(y) = 1 \pm \epsilon$ of degree $k$ for $\sopl_n$. Our goal is to show a lower bound on $k$. We will give a randomised \emph{decision-to-search reduction}, in the style of~\cite{Raz1992,Huynh2012,Goos2018cbs,Itsykson2021}, showing that a low-degree $\epsilon$-NS refutation would imply a low-degree approximating polynomial for the~$(n-1)$-bit $\Or$ function. The following well-known fact then concludes the proof.

\begin{fact}[\cite{Nisan1994}] \label{fact:or-deg}
Suppose that $p$ is an $n$-variate real polynomial such that $p(x)=\Or_n(x)\pm 1/3$ for all $x\in\{0,1\}^n$. Then $\deg(p)\geq\Omega(\sqrt{n})$.
\end{fact}

\newcommand{\f}{\bm{f}}
\newcommand{\y}{\bm{y}}
\renewcommand{\u}{\bm{u}}

\paragraph{Definition of reduction.}
We define a depth-$d$ deterministic reduction as a pair $(f,u)$ such that
\begin{enumerate}[label={(\arabic*)}]
\item \label{it:r1}
$f\colon\{0,1\}^{n-1}\to\{0,1\}^{n'}$ is a function that maps an input $x$ of $\Or_{n-1}$ to an input $y=f(x)$ of~$\sopl_n$. Moreover, each output bit $f_i(x)\in\{0,1\}$ is a depth-$d$ decision tree function of $x$.
\item For any input $x$, the only solutions of $y=f(x)$ are active sinks on the last row $\{n\}\times [n]$. We write $\sols(y)\subseteq\{n\}\times [n]$ for the set of solutions in $y$. Moreover, $u\in\sols(y)$ is a solution called the \emph{planted} solution. (Note that $u$ does not depend on $x$.)
\item If $\Or(x)=0$, then $y=f(x)$ contains a unique solution, namely $\sols(y)=\{u\}$.
\item \label{it:r4}
If $\Or(x)=1$, then $y=f(x)$ contains at least two solutions, $|\sols(y)|\geq 2$.
\end{enumerate}
We then define a depth-$d$ \emph{randomised} reduction $\calR$ as a probability distribution over depth-$d$ deterministic reductions $(\f,\u)\sim \calR$. For every $x$, we write $\calR_x$ for the distribution of $(\f(x),\u) = (\y,\u)$. We say that a pair $(\y,\u)$ is \emph{ideal} if it satisfies the following.
\begin{quote}
\emph{Ideal $(\y,\u)$:}~ Let $y$ be any outcome of $\y$ and consider $\u$ conditioned on $\y=y$, namely, $\u'\coloneqq(\u\mid\y=y)$. Then $\u'$ is uniformly distributed over $\sols(y)$; in short, $\u'\sim\sols(y)$.
\end{quote}
We say $\calR$ is \emph{ideal} if $\calR_x$ is ideal for every $x$.

\paragraph{Ideal reduction \boldmath$\Rightarrow$ Approximation to $\Or$.}
Next, we show that if we had an ideal reduction, we could construct an approximating polynomial for $\Or$. We write $i_u$ for the unique~$i$ such that the polynomial equation $a_i(y)=0$ encodes the $\sopl_n$ constraint that $u$ is not an active sink. Namely, this corresponds to the equation $s_u = 0$, where the bit $s_u \in \{0,1\}$ of the input $y$ encodes whether or not $u$ is active (see \cref{def:SOPL-CNF}). If we think of $u\in\{n\}\times[n]$ as encoded by an~$O(\log n)$-bit string, we can define an $[n'+O(\log n)]$-variate polynomial
\begin{equation} \textstyle
q(y,u)
~\coloneqq~ p_{i_u}(y) a_{i_u}(y)
~=~ \sum_i\mathds{1}[i=i_u] p_i(y)a_i(y).
\end{equation}
Here, for every $i$, the indicator function $\mathds{1}[i=i_{u}]\in\{0,1\}$ is computed by an $O(\log n)$-degree polynomial. This means $q$ has degree $\deg(q)\leq O(k\log n)$. If $(\y,\u)$ is ideal, then
\begin{align}
\E[q(\y,\u)]
&~=~\textstyle \E_{y\sim\y}\big[\E_{u'\sim (\u\mid\y=y)}[p_{i_{u'}}(y) a_{i_{u'}}(y)]\big] \notag \\
&~=~\textstyle \E_{y\sim\y}\big[\E_{u'\sim \sols(y)}[p_{i_{u'}}(y) a_{i_{u'}}(y)]\big] \notag \\
&~=~\textstyle \E_{y\sim\y}\big[|\sols(y)|^{-1}\sum_{u'\in \sols(y)}p_{i_{u'}}(y) a_{i_{u'}}(y)\big] \notag \\
&~=~\textstyle \E_{y\sim\y}\big[|\sols(y)|^{-1}\sum_i p_i(y) a_i(y)\big] \notag \\
&~=~\textstyle \E_{y\sim\y}\big[|\sols(y)|^{-1}]\cdot (1\pm\epsilon) \notag \\
&~=~\textstyle (1\pm\epsilon)\cdot \E\big[|\sols(\y)|^{-1}\big] \label{eq:ideal}
\end{align}
where we used the fact that $\sum_{u'\in \sols(y)}p_{i_{u'}}(y) a_{i_{u'}}(y) = \sum_i p_i(y) a_i(y)$, because $a_i(y) = 0$ for all $i \notin \{i_{u'}: u' \in \sols(y)\}$, given that $y$ satisfies all the $\sopl_n$ constraints, except the equations requiring that $u'$ not be an active sink, for $u' \in \sols(y)$.

Suppose for a moment we had an ideal depth-$d$ randomised reduction $\calR$. Then, we could construct the polynomial
\[\textstyle
r(x) ~\coloneqq~ \E_{\calR_x}[q(\y,\u)] ~=~ \sum_{f,u}\Pr_{\calR}[(\f,\u)=(f,u)]\cdot q(f(x),u).
\]
We have $\deg(r)\leq O(dk\log n)$. Moreover, if $\Or(x)=0$ then $r(x)=1\pm\epsilon$; and if $\Or(x)=1$ then~$r(x)\in[0,(1+\epsilon)/2]$, since $\E\big[|\sols(\y)|^{-1}\big] \in [0,1/2]$. Thus for $\epsilon=0.01$, if we consider $t(x) \coloneqq 1-r^2(x)$ we get that $t$ approximates~$\Or$ to within error $1/3$. Using \cref{fact:or-deg}, we deduce that $k\geq \Omega(\sqrt{n}/(d\log n))$.

In summary, all that remains is to find an ideal reduction of shallow depth. Unfortunately, we do not know how to design an ideal reduction for $\sopl$. We instead give a reduction that is \emph{locally indistinguishable} from an ideal one, which will suffice for us.

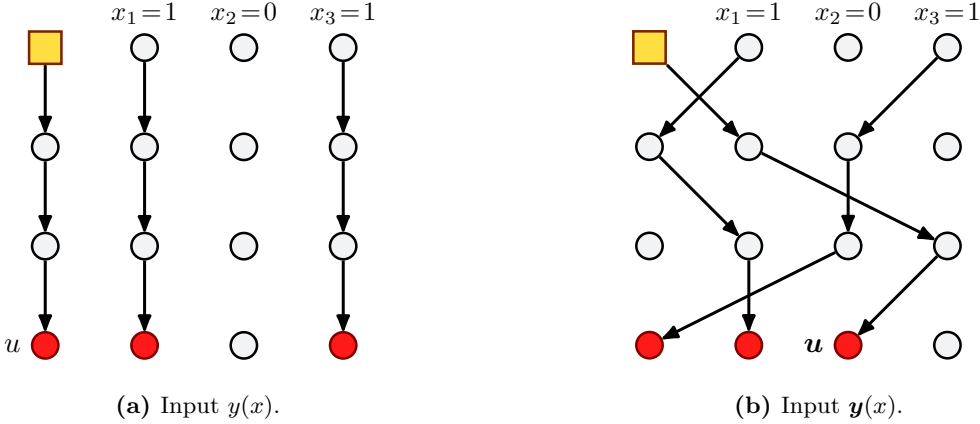
\begin{figure}
\centering
\begin{subfigure}[b]{0.5\textwidth}
    \centering
    \begin{tikzpicture}[y=-1cm, scale=1.2]
% node grid
\coordinate (p11) at (0, 0); 				 \coordinate (p12) at (\spacex, 0); 				 \coordinate (p13) at (2*\spacex, 0); 				\coordinate (p14) at (3*\spacex, 0);
\coordinate (p21) at (0, \spacey);	 \coordinate (p22) at (\spacex, \spacey); 	 \coordinate (p23) at (2*\spacex, \spacey); 	\coordinate (p24) at (3*\spacex, \spacey);
\coordinate (p31) at (0, 2*\spacey); \coordinate (p32) at (\spacex, 2*\spacey); \coordinate (p33) at (2*\spacex, 2*\spacey); \coordinate (p34) at (3*\spacex, 2*\spacey);
\coordinate (p41) at (0, 3*\spacey); \coordinate (p42) at (\spacex, 3*\spacey); \coordinate (p43) at (2*\spacex, 3*\spacey); \coordinate (p44) at (3*\spacex, 3*\spacey);

% node drawing
\tikzstyle{node_regular} = [node_regular_intro]

\node[node_a]       (P11) at (p11) {};
{\small
\node[node_regular, label=above:{$x_1\!=\!1$}] (P12) at (p12) {};
\node[node_regular, label=above:{$x_2\!=\!0$}] (P13) at (p13) {};
\node[node_regular, label=above:{$x_3\!=\!1$}] (P14) at (p14) {};
}
\node[node_regular] (P21) at (p21) {};
\node[node_regular] (P22) at (p22) {};
\node[node_regular] (P23) at (p23) {};
\node[node_regular] (P24) at (p24) {};
\node[node_regular] (P31) at (p31) {};
\node[node_regular] (P32) at (p32) {};
\node[node_regular] (P33) at (p33) {};
\node[node_regular] (P34) at (p34) {};
\node[node_solution, label=left:{$u$}] (P41) at (p41) {};
\node[node_solution] (P42) at (p42) {};
\node[node_regular] (P43) at (p43) {};
\node[node_solution] (P44) at (p44) {};

% edges drawing
\draw[edge_regular] (P11) -- (P21);
\draw[edge_regular] (P21) -- (P31);
\draw[edge_regular] (P31) -- (P41);

\draw[edge_regular] (P12) -- (P22);
\draw[edge_regular] (P22) -- (P32);
\draw[edge_regular] (P32) -- (P42);

\draw[edge_regular] (P14) -- (P24);
\draw[edge_regular] (P24) -- (P34);
\draw[edge_regular] (P34) -- (P44);

\end{tikzpicture}
    \vspace{2mm}
    \caption{Input $y(x)$.}
    \label{figure:Rtilde_before}
\end{subfigure}%
\begin{subfigure}[b]{0.5\textwidth}
    \centering
    \begin{tikzpicture}[y=-1cm, scale=1.2]
% node grid
\coordinate (p11) at (0, 0); 				 \coordinate (p12) at (\spacex, 0); 				 \coordinate (p13) at (2*\spacex, 0); 				\coordinate (p14) at (3*\spacex, 0);
\coordinate (p21) at (0, \spacey);	 \coordinate (p22) at (\spacex, \spacey); 	 \coordinate (p23) at (2*\spacex, \spacey); 	\coordinate (p24) at (3*\spacex, \spacey);
\coordinate (p31) at (0, 2*\spacey); \coordinate (p32) at (\spacex, 2*\spacey); \coordinate (p33) at (2*\spacex, 2*\spacey); \coordinate (p34) at (3*\spacex, 2*\spacey);
\coordinate (p41) at (0, 3*\spacey); \coordinate (p42) at (\spacex, 3*\spacey); \coordinate (p43) at (2*\spacex, 3*\spacey); \coordinate (p44) at (3*\spacex, 3*\spacey);

% node drawing
\tikzstyle{node_regular} = [node_regular_intro]

\node[node_a]       (P11) at (p11) {};
{\small
\node[node_regular, label=above:{$x_1\!=\!1$}] (P12) at (p12) {};
\node[node_regular, label=above:{$x_2\!=\!0$}] (P13) at (p13) {};
\node[node_regular, label=above:{$x_3\!=\!1$}] (P14) at (p14) {};
}
\node[node_regular] (P21) at (p21) {};
\node[node_regular] (P22) at (p22) {};
\node[node_regular] (P23) at (p23) {};
\node[node_regular] (P24) at (p24) {};
\node[node_regular] (P31) at (p31) {};
\node[node_regular] (P32) at (p32) {};
\node[node_regular] (P33) at (p33) {};
\node[node_regular] (P34) at (p34) {};
\node[node_solution] (P41) at (p41) {};
\node[node_solution] (P42) at (p42) {};
\node[node_solution, label=left:{$\bm{u}$}] (P43) at (p43) {};
\node[node_regular] (P44) at (p44) {};

% edges drawing
\draw[edge_regular] (P11) -- (P22);
\draw[edge_regular] (P22) -- (P34);
\draw[edge_regular] (P34) -- (P43);

\draw[edge_regular] (P12) -- (P21);
\draw[edge_regular] (P21) -- (P32);
\draw[edge_regular] (P33) -- (P41);

\draw[edge_regular] (P14) -- (P23);
\draw[edge_regular] (P23) -- (P33);
\draw[edge_regular] (P32) -- (P42);

\end{tikzpicture}
    \vspace{2mm}
    \caption{Input $\y(x)$.}
    \label{figure:Rtildeafter}
\end{subfigure}
\caption{Randomised reduction $\calR$. First, we compute $y(x)$ deterministically from $x$. This input always contains a path down the left-most column, which terminates at the active sink $u$ (\emph{planted} solution). Moreover, for every $i$ with $x_i=1$ there is a path down the $(i+1)$-st column. The number of active sinks is~$|\sols(y)|=1+|x|$, where $|x|$ denotes the Hamming weight. In the second step, we randomly permute every row of nodes, except the first one. This yields the random output $(\y,\u)$ of the reduction.}
\label{fig:reduction}
\end{figure}

\paragraph{A locally ideal reduction.}
Consider the following depth-$1$ randomised reduction $\calR$; see \cref{fig:reduction}.
\begin{enumerate}
\item Let $y=y(x)$ be the input to $\sopl_n$ that has a directed path running down the first column of nodes, starting at distinguished node $(1,1)$ and terminating at the active sink $u\coloneqq (n,1)$ (say $u$ is made active by being assigned $1$ as successor). Moreover, we activate a path in $y$ down column $i\geq2$ iff $x_{i-1}=1$. Note that $y$ is a depth-$1$ decision tree function of $x$, and $u$ does not depend on $x$ at all.
\item Let $\y=\y(x)$ be obtained from $y$ so that, for each row except the first, $i\in[n]\setminus\{1\}$, randomly permute the nodes~$\{i\}\times[n]$ on that row (updating the successor/predecessor pointers). Let $\u$ be the sink node that $u$ is mapped to.
\item Output $(\f,\u)$ where $\f(x)\coloneqq \y(x)$.
\end{enumerate}
It is easy to check that $\calR$ satisfies items \ref{it:r1}--\ref{it:r4} for every outcome of randomness. In particular, we have $|\sols(\y)|=1+|x|$. Unfortunately, $\calR$ is \emph{not} ideal: $\u$ is always the active sink at the end of the path starting at the distinguished node. What we would really like instead is that $\calR_x$ was distributed as the ideal pair $(\y,\u)\sim\calI_x$ defined by the following procedure: Sample $(\y,\u')\sim \calR_x$; define $\u$ such that for every outcome $y$, $(\u\mid \y=y)\sim\sols(y)$; and output $(\y,\u)$.

Define two functions $\{0,1\}^{n-1}\to\R$ by
\begin{align}
r(x) ~\coloneqq~&\textstyle \E_{\calR_x}[q(\y,\u)], \\
r'(x) ~\coloneqq~&\textstyle \E_{\calI_x}[q(\y,\u)].
\end{align}
We know that $r$ has low degree as a polynomial, $\deg(r)\leq O(k\log n)$, and $r'$ has the ideal output behaviour, $r'(x)= (1\pm\epsilon)\cdot \E\big[|\sols(\f(x))|^{-1}\big]$ by~\cref{eq:ideal}. The following claim shows that, in fact, $r=r'$, and hence we can get the best of both worlds. By the discussion above, we are then able to construct an~$O(k\log n)$-degree approximating polynomial for $\Or$, which concludes the proof of \cref{lem:apx-ns}.
\begin{claim}
\label{claim:1a}
We have $r(x)=r'(x)$ for all $x\in\{0,1\}^{n-1}$.
\end{claim}
\begin{proof}
By linearity of expectation, it suffices to show $\E_{\calR_x}[m(\y,\u)]=\E_{\calI_x}[m(\y,\u)]$ for any monomial~$m$ of $q$ and every $x$. Fix a monomial $m$. We claim that $\calR_x$ and $\calI_x$ have the same marginal distribution over the variables read by $m$, which would prove the claim. We may assume that $\deg(m)\leq O(k\log n)\leq o(n)$ because otherwise \cref{lem:apx-ns} is proved. Hence there exist two consecutive rows $i,i+1\in[n/3,2n/3]$ such that $m$ does not read any variables associated with either row. Starting with a sample $(\y,\u)\sim\calR_x$ we can generate a sample from $\calI_x$ as follows: Consider active nodes $A\subseteq\{i\}\times[n]$ and $B\subseteq\{i+1\}\times[n]$ on rows $i$ and $i+1$ in $\y$ and the $|A|=|B|=1+|x|$ many directed edges joining them (defined by successor pointers for row $i$ and predecessor pointers for row $i+1$). Reroute these edges by choosing a random bijection $A\to B$, and denote the resulting input by $\y'$. Then $(\y',\u)\sim\calI_x$. This proves our claim about the marginals, since our modification to the input $\y$ was done outside the variables read by $m$.
\end{proof}

\subsection{Upper bound for RevRes} \label{sec:ub-revres}

Our characterisation of $\SOPL^{dt}$ by RevRes in \cref{sec:char-revres} involves proving that $\sopl_n$ (understood as an~$O(\log n)$-width CNF contradiction) admits an $O(\log n)$-width polynomial-size RevRes refutation (\cref{thm:sopl-in-revres}). If we want to further optimise this down to a constant-width polynomial-size RevRes refutation, as claimed by \cref{thm:revres-ns}, then we can consider instead a \emph{sparse constant-width} variant of~$\sopl_n$. Indeed, the following sparsifying construction is standard and so we only sketch it.

We start by defining a bounded-degree dag $G$ that models the connectivity structure of the~$[n]\times [n]$ grid with successor/predecessor pointers. The nodes of $G$ include all the grid nodes $[n]\times[n]$. Moreover, for each $u\in[n-1]\times[n]$ we include in $G$ a \emph{successor tree} $S_u$ that is a full binary tree with $n$ leaves, and has edges directed from the root towards the leaves. Similarly, for each $u\in([n]\setminus\{1\})\times[n]$ we include in $G$ a \emph{predecessor tree} $P_u$ whose edges are directed from leaves towards the root. We identify the root nodes of $S_u$ and $P_u$ with $u$. Moreover, for grid nodes $(i,j)$ and $(i+1,k)$ appearing on consecutive rows, we identify the $k$-th leaf of $S_{(i,j)}$ and the $j$-th leaf of $P_{(i+1,k)}$. This completes the description of $G$. Note that the in/out-degree of every node is at most $2$.

We can now define a search problem $\sopl_G$ relative to $G$. As input, each node $u$ in $G$ gets a successor $s_u$ and a predecessor $p_u$ picked from $\{0,1\}\cup\{\nul\}$. For example, $s_u=0$ ($s_u=1$) means that $u$'s successor is the left (right) child of $u$ in $G$. The constraints of $\sopl_G$ can now be written in constant width. The RevRes upper bound in \cref{thm:sopl-in-revres} can be adapted to yield a constant-width polynomial-size refutation of $\sopl_G$. Moreover, the original grid version $\sopl_n$ can be reduced to the graph version~$\sopl_G$ using an $O(\log n)$-depth decision tree reduction; see, for example,~\cite[\S4.2]{Fleming2022} for details (but for $\sod$ instead of $\sopl$). The existence of this reduction implies that $\sopl_G$ needs large $\epsilon$-NS degree, because we showed that $\sopl_n$ does.

This concludes the proof of \cref{thm:revres-ns} in case $\F=\R$.

\subsection{Lower bound for $\F$-NS} \label{sec:lb-fns}

We now prove the lower bound in \autoref{thm:revres-ns} for any field $\F$.

\begin{lemma}
$\F\textup{-NS}(\sopl_n)\geq n^{\Omega(1)}$.
\end{lemma}
\begin{proof}
Prior work has shown that $\sod_n$ (understood as an $O(\log n)$-CNF) requires $n^{\Omega(1)}$-degree $\F$-NS refutations~\cite{Buss1998,BureshOppenheim2002,Rezende2019}, and similarly that $\sol_n$ (undestood as an $O(\log n)$-CNF) requires $n^{\Omega(1)}$-degree $\F$-NS refutations~\cite{Beame1998,Beame1998a}. Define the CNF formula
\[
F_n ~\coloneqq~ \sod_n\land \sol_n
\]
where $\sod_n$ and $\sol_n$ are defined on disjoint sets of variables. The following claim (proved below) states that $F_n$ requires $n^{\Omega(1)}$-degree $\F$-NS refutations, or, in other words, $\F\textup{-NS}(F_n) \geq n^{\Omega(1)}$.

\begin{claim}\label{clm:meet}
Let $F$ and $G$ be two CNF contradictions over disjoint sets of variables. If $F$ and $G$ require $\F$-NS refutations of degree $\geq d$, then $F\land G$ requires $\F$-NS refutations of degree $\geq d$.
\end{claim}

By the definition of $F_n$, we have
\begin{align*}
\Theta(\textup{Res}(F_n))~=~\PLS^{dt}(S(F_n)) ~=~& \sod^{dt}(S(F_n)) ~\leq~ O(\log n),\\
\Theta(\textup{uSA}(F_n))~=~\PPADS^{dt}(S(F_n)) ~=~& \sol^{dt}(S(F_n)) ~\leq~ O(\log n).
\end{align*}
By the intersection theorem (\cref{thm:intersection-thm}) corresponding to $\SOPL = \PLS \cap \PPADS$ we conclude that $S(F_n)$ has an efficient $\sopl$-formulation:
\[
\Theta(\textup{RevRes}(F_n)) ~=~ \SOPL^{dt}(S(F_n)) ~=~ \sopl^{dt}(S(F_n)) ~\leq~ O(\log n).
\]
If we had $\F\textup{-NS}(\sopl_n)\leq n^{o(1)}$, then because $S(F_n)$ reduces to $\sopl_{n^{O(1)}}$ via an $O(\log n)$-depth decision tree reduction, we would have $\F\textup{-NS}(F_n)\leq n^{o(1)}$, which is a contradiction.
\end{proof}
\begin{proof}[Proof of \cref{clm:meet}]
The least degree of an $\F$-NS refutation of a set of polynomial equations $\calF\coloneqq \{a_i(x)=0\}$ can be characterised by the maximum $d$ such that $\calF$ admits a \emph{$d$-design}~\cite[\S2]{Buss1998}, that is, an~$\F$-linear map~$\varphi\colon \F[x]\to\F$ satisfying (i) $\varphi(1)=1$, and (ii) $\varphi(q(x)\cdot a_i(x))=0$ for all $a_i$ and~$q\in \F[x]$ such that $\deg(q)+\deg(a_i)< d$. Let $\varphi$ and $\varphi'$ be $d$-designs for $F$ and $G$ (encoded as sets of polynomial equations) over variables $x$ and $y$, respectively. For each monomial $m(x)m'(y)$ in variables $x,y$, we define $\Phi(m(x)m'(y))\coloneqq\varphi(m(x))\cdot\varphi(m'(y))$. We can extend this definition linearly into a map~$\Phi\colon \F[x,y]\to\F$. We claim that $\Phi$ is a $d$-design for $F\land G$. Indeed, for (i) we have $\Phi(1)=\varphi(1)\varphi'(1)=1$. For (ii) it suffices to check the condition for each monomial $q(x,y)=m(x)m'(y)$ and an axiom $a_i(x)$ of $F$ (the case of $G$ is analogous) with $\deg(q)+\deg(a_i)< d$. We have $\Phi(m(x)m'(y)\cdot a_i(x))=\varphi(m(x)\cdot a_i(x))\varphi'(m'(y))=0\cdot \varphi'(m'(y))=0$.
\end{proof}

\section{Resolution vs.\ Sherali--Adams} \label{sec:proof-res-sa}

In this section we prove \cref{thm:res-sa}, restated below.
\ResVuSA*

We consider the $\sod$ principle. We first show that it requires large coefficients to refute in low-degree SA, and then we recall why it has low-width Resolution refutations.

\subsection{Lower bound for SA}
We consider the $\sod_{n^2}$ search problem on the grid $[n^2]\times [n^2]$.  We think of this large grid as being further subdivided into $n^2$ many subgrids, each of size $n\times n$. The $(i,j)$-subgrid consists of nodes
\[
((i-1)n,(j-1)n)+[n]\times[n]
~\coloneqq~\big\{((i-1)n+i',(j-1)n+j'):(i',j')\in[n]\times[n]\big\}.
\]
Recall that the input to this search problem consists of a successor $s_u\in[n^2]\cup\{\nul\}$ for each grid node~$u$. For the purposes of SA, we encode this input by a string~$x\in\{0,1\}^{n'}$ over~$n' =O(n^4\log n)$ variables. Moreover, we can think of $\sod_{n^2}$ as a set of unsatisfiable polynomial equations~$\{a_i(x)=0\}$ each of degree $O(\log n)$. Our goal is to prove the following lemma.
\begin{lemma}\label{lem:sa-lb}
Any degree-$n^{o(1)}$ SA proof of $\sod_{n^2}$ requires coefficients of magnitude $\exp(\Omega(n))$.
\end{lemma}

Suppose we are given a degree-$n^{o(1)}$ SA refutation of $\sod_{n^2}$ over the reals,
\begin{equation} \label{eq:sa-ref}
\sum_{i\in[m]} p_i(x)a_i(x) ~=~ 1 + J(x).
\end{equation}
Our idea is to apply \cref{lem:apx-ns} iteratively in stages to find a sequence of inputs $x_1,\ldots,x_n$ with a RHS value $1+J(x_i)\geq 2^{\Omega(i)}$. Hence \cref{lem:sa-lb} follows at stage $i=n$, since there are at most $\exp(n^{o(1)})$ many monomials, and so one of them must have a coefficient of exponential magnitude.

We start by preprocessing the SA refutation~\cref{eq:sa-ref} for technical convenience. We may assume wlog that each term $t$ appearing in $J=\sum_t \alpha_t t$ satisfies the following.
\begin{enumerate}
\item $t$ is \emph{node-aligned}: if $t$ reads some variable associated with a node $u$, then it reads all the $O(\log n)$ variables associated with $u$. To ensure this, we may replace a term $t$ with an equivalent sum of two terms, $t=tx_i+t\bar{x}_i$, which reads one more variable. Adding more literals to terms like this will only increase the degree of the proof by an $O(\log n)$ factor.
\item $t$ is \emph{curious}: if $t$ reads a node $u$ that lies on the last row of a subgrid, that is, $u\in\{in\}\times[n^2]$ for some $i\in[n]$, then $t$ also reads the successor $s_u$ of $u$ (if any) on the next row. Similarly as above, this can be ensured by at most doubling the degree of the proof.
\item $t$ is \emph{non-witnessing}: it does not witness a solution to the search problem. Formally, $t$ witnesses a violation $a_i\neq 0$ if for all $x$, $t(x)=1\Rightarrow a_i(x)\neq 0$ (or contrapositively, $a_i(x)=0\Rightarrow t(x)=0$). To ensure this, if $t$ is witnessing, we can factor\footnote{The existence of such a factorization is easy to see here since both $a_i$ and $t$ are conjunctions of literals. More generally, the existence of such a factorization is guaranteed for any multilinear polynomials $a_i$ and $t$ satisfying $a_i(x)=0\Rightarrow t(x)=0$ on the boolean hypercube, where we simplify expressions using the constraints $x_i^2-x_i=0$. One way to prove this is by using the fact that two multilinear polynomials are syntactically identical if and only if they agree on the boolean hypercube.} $t=p'_i a_i$ and move $t$ to the LHS of the proof.
\end{enumerate}

\paragraph{First stage.}
Let $y_1$ be an input to $\sopl_n$ defined on nodes $[n]\times[n]$. We can embed $y_1$ inside an input to $\sod_{n^2}$ as follows. We write $(\nul^*\gets y_1)$ for the input to $\sod_{n^2}$ where we start with an assignment of $\nul$ to all nodes $[n^2]\times[n^2]$ (denoted $\nul^*$), and then overwrite the top-left~$(1,1)$-subgrid with the successor pointers in $y_1$ (aligning the distinguished nodes of $\sod_{n^2}$ and $\sopl_n$). In this reduction, we can forget the predecessor pointers, as they are not part of the input to $\sod$. Now every solution of $(\nul^*\gets y_1)$ for $\sod_{n^2}$ corresponds naturally to a solution of~$y_1$ for $\sopl_n$. (A minor detail is that the \emph{active sinks} in $y_1$ correspond to \emph{proper sinks} in $(\nul^*\gets y_1)$.) Using this reduction, we can view our SA refutation of $\sod_{n^2}$ also as a refutation of $\sopl_n$.

We claim that there is some input $y_1$ to $\sopl_n$ such that for $x'_1\coloneqq(\nul^*\gets y_1)$ we have a RHS value $1+J(x'_1)\geq 1.5$. Suppose not: then the RHS is always in $[1,1.5]$ for all $y_1$, which means we have a low-degree $\frac{1}{2}$-NS proof of $\sopl_n$. But this contradicts \cref{lem:apx-ns}.

We have now found an input $x'_1=(\nul^*\gets y_1)$ with RHS at least $1.5$. Before we iterate this argument in the second stage, we have to clean up $x'_1$ slightly.

\begin{figure}
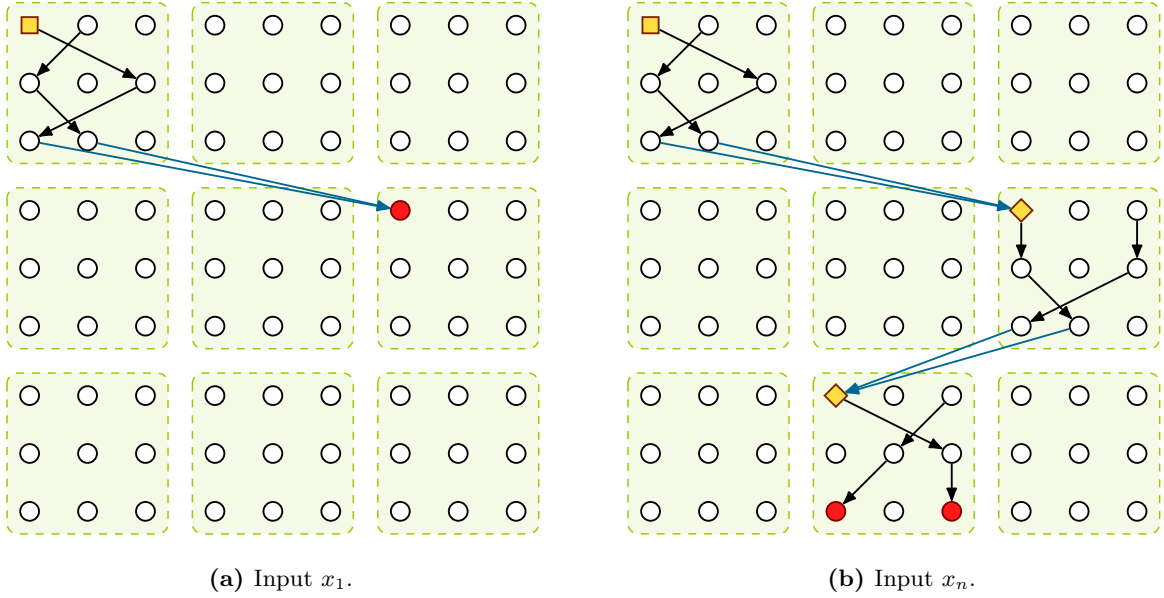

    \centering
\begin{subfigure}[b]{0.5\textwidth}
    \centering
    \input{figures/embedded_instance.tex}
    \vspace{2mm}
    \caption{Input $x_1$.}
    \label{figure:embedded}
\end{subfigure}%
\begin{subfigure}[b]{0.5\textwidth}
    \centering
    \input{figures/iterated_instance.tex}
    \vspace{2mm}
    \caption{Input $x_n$.}
    \label{figure:iterated}
\end{subfigure}
    \caption{Illustration of the proof of \cref{lem:sa-lb}. (a) In the first stage, we construct an input $x_1$ to $\sod_{n^2}$ that embeds an input $y_1$ to $\sopl_n$ in the top-left $(1,1)$-subgrid, and moreover, all the active sinks of $y_1$ are assigned as successor the top-left corner of some $(2,j)$-subgrid. (b) The completed construction after $n$ stages.}
    \label{fig:sa-lb}
\end{figure}

\paragraph{First stage: Clean-up.}
Recall that the instances considered in the proof of \cref{lem:apx-ns} consist of some number of directed paths that terminate at sinks $\sols(y_1)\subseteq\{n\}\times[n]$. We will modify~$x'_1$ by making the nodes $\sols(y_1)$ point to the same \emph{top-left corner} of a $(2,j)$-subgrid for some $j\in[n]$. Indeed, let $\rho_j\colon\sols(y_1)\to[n^2]$ be the partial assignment that assigns $(n,(j-1)n)+(1,1)$ (top-left corner of the $(2,j)$-subgrid) as the successor of all nodes in $\sols(y_1)$. Let $(x'_1\gets\rho_j)$ be the input obtained from $x'_1$ by applying $\rho_j$. (We actually have $x'_1=x'_1\gets\rho_1$, as this is how we decided to make every node in $\sols(y_1)$ an active sink in $y_1$.) By defining $x_1\coloneqq (x'_1\gets\rho_j)$ for a carefully chosen~$j\in[n]$ (see~\cref{figure:embedded}), we establish the following properties for the start of the next stage.
\begin{enumerate}[label=(1\alph*)]
    \item\label{l1} The only solutions in $x_1$ are proper sinks pointing to the corner of the $(2,j)$-subgrid.
    \item\label{l2} We have $1+J(x_1\gets y_2) \geq 1.4$ for any partial assignment $y_2$ to nodes in the $(2,j)$-subgrid.
\end{enumerate}
Property \ref{l1} is true by construction and we prove property \ref{l2} below. 
\begin{claim}
There exists a $j\in[n]$ such that \ref{l2} holds.
\end{claim}
\begin{proof}
Let us first prove that for every term $t$ appearing in $J=\sum_t\alpha_t t$, we have
\begin{equation} \label{eq:terms}
    t(x'_1) ~=~ t(x'_1\gets\rho_j)\qquad \forall j.
\end{equation}
It suffices to show that any term $t$ in $J$ with $t(x'_1)=1$ (or $t(x'_1\gets\rho_j)=1$) does not read any nodes in~$\sols(y_1)$. Assume for contradiction that such a $t$ reads a node $u\in \sols(y_1)$. Then, because $t$ is curious, it also reads $u$'s successor node (note that $s_u\neq\nul$ in both $x'_1$ and $x'_1\gets\rho_j$) on the next row. This successor node is set to $\nul$ in $x'_1$ (and $(x'_1\gets\rho_j)$) and hence $t$ witnesses that $u$ is a solution (proper sink). But this contradicts our assumption that $t$ is non-witnessing. This proves~\cref{eq:terms}.

Define $J_j \coloneqq \sum_{\smash{t\in T_j}}\alpha_t t$ where $T_j$ is the set of terms $t$ in $J$ that do \emph{not} read any node from the~$(2,j)$-subgrid.  Note that each $t$ can read from at most $\deg(t)\leq n^{o(1)}$ many different subgrids, and hence if we choose $\bm{j}\sim[n]$ at random, $\Pr[t\in T_{\bm{j}}]\geq 99\%$. We now have
\[\textstyle 
\E[1+J_{\bm{j}}(x'_1)] ~=~ 1+\sum_t\Pr[t\in T_{\bm{j}}]\alpha_t t(x'_1) ~\geq~ 99\%\cdot (1+J(x'_1)) ~\geq~ 99\%\cdot 1.5 ~\geq~ 1.4.
\]
By averaging, there is some fixed $j\in[n]$ such that $1+J_j(x'_1)\geq 1.4$. Defining $x_1\coloneqq (x'_1\gets\rho_j)$ for this particular $j$, we have, for every assignment $y_2$ to the $(2,j)$-subgrid,
\[
1+J(x_1\gets y_2)
~\geq~ 1+J_j(x_1\gets y_2)
~=~ 1+J_j(x_1)
~\overset{\cref{eq:terms}}{=}~ 1+J_j(x'_1)
~\geq~ 1.4. \qedhere
\]
\end{proof}

\paragraph{Second stage.}
Here we start with the input $x_1$ satisfying \ref{l1}--\ref{l2} for some $j\in[n]$. Let $y_2$ be any input to~$\sopl_n$. We think of $y_2$ (ignoring predecessor pointers) as embedded in the $(2,j)$-subgrid. Consider the input~$(x_1\gets y_2)$ where the distinguished node of $y_2$ is aligned with corner of the $(2,j)$-subgrid, which is the only sink in $x_1$ by \ref{l1}. Then every solution of $(x_1\gets y_2)$ for $\sod_{n^2}$ corresponds to a solution of $y_2$ for $\sopl_n$. Hence we can view our SA refutation of $\sod_{n^2}$ as a refutation of $\sopl_n$ (this time in the $(2,j)$-subgrid). Moreover, we have from~\ref{l2} that the RHS of the proof evaluates to~$1+J(x_1\gets y_2)\geq 1.4$ for all $y_2$. If we scale our original SA proof by a factor~$1/1.4$, we get another polynomial identity
\begin{equation}
\frac{1}{1.4}\sum_{i\in[m]} p_i(x)a_i(x) ~=~ \frac{1}{1.4}(1 + J(x)),
\end{equation}
where the RHS evaluates to at least $1$ on any input of the form $x=(x_1\gets y_2)$. Using \cref{lem:apx-ns} we can now conclude that there must exist an input $x_2'=(x_1\gets y_2)$ such that $\frac{1}{1.4}(1 + J(x'_2))\geq 1.5$, or equivalently, $1 + J(x_2') \geq 1.5\cdot1.4$.

\paragraph{Second stage: Clean-up.}
Using exactly the same argument as in the first clean-up stage, we conclude that $x'_2$ can be cleaned up into $x_2$ such that for some $j\in[n]$ (different $j$ than in first stage):
\begin{enumerate}[label=(2\alph*)]
    \item The only solutions in $x_2$ are proper sinks pointing to the corner of the $(3,j)$-subgrid.
    \item We have $1+J(x_2\gets y_3) \geq 1.4^2$ for any partial assignment $y_3$ to nodes in the $(3,j)$-subgrid.
\end{enumerate}

By continuing this argument in the same fashion, we can eventually, at stage $n$, find an input~$x_n$ with $1+J(x_n) \geq 1.4^n$ (see~\cref{figure:iterated}). This concludes the proof of \cref{lem:sa-lb}.

\subsection{Upper bound for Resolution}

It is well-known that $\sod_n$ (understood as an $O(\log n)$-width CNF contradiction) admits an~$O(\log n)$-width Resolution refutation (e.g.,~\cite[Theorem 8.18]{Kamath2020}). If we want to further optimise this down to a constant-width refutation, as claimed by \cref{thm:res-sa}, then we can consider a sparse variant of~$\sod_n$ similarly as we did in~\cref{sec:ub-revres}. We omit the details.

\section{Proofs of Characterisations} \label{sec:characterisations}

In this section we prove \cref{thm:characterisations}, restated below.
\Chars*

Recall that the notation $\textsc{A}^{dt}(\textsc{B})$ for total search problems $\textsc{A}, \textsc{B}$ is the minimum complexity (namely,~$\log \text{size} + \text{depth}$) of an $\textsc{A}$-formulation of $\textsc{B}$.
Similarly, for a proof system $\textup{P}$ and CNF formula $F$ the notation $\textup{P}(F)$ is the minimum of $\log \text{size}(\Pi) + \deg(\Pi)$ where $\Pi$ is a $\textup{P}$-proof of $F$.

\subsection{Unary Nullstellensatz and $\PPAD$}

We first argue that unary Nullstellensatz corresponds to the decision tree class $\PPAD^{dt}$.

\begin{theorem}
	\label{thm:unaryns-ppad}
	Let $F$ be an unsatisfiable CNF formula.
	Then,
	\begin{itemize}
		\item If $F$ has a degree-$d$ size-$L$ uNS proof, then $S(F)$ has a depth-$O(d)$ $\eol_{O(L)}$-formulation.
		\item If $S(F)$ has a depth-$d$ $\eol_{L}$-formulation, then $F$ has a degree-$O(d)$ size-$L2^{O(d)}$ uNS proof.
	\end{itemize}
	In particular, $\PPAD^{dt}(S(F)) = \Theta(\textup{uNS}(F))$.
\end{theorem}
\begin{corollary}
	For any sequence $F_n$ of $\poly(\log n)$-width CNF formulas, $F_n$ has a degree-$\poly(\log n)$, size-$n^{\poly(\log n)}$ unary Nullstellensatz proof if and only if $S(F) \in \PPAD^{dt}$.
\end{corollary}

We prove \cref{thm:unaryns-ppad} in the next two lemmas.
The proof of this theorem is itself modelled on a similar characterization of $\PPA$-formulations by $\mathbb{F}_2$-Nullstellensatz, proved by \cite{Beame1998, Goos2018}.
It turns out to be easier to show that $\eol$-formulations imply Nullstellensatz proofs, so we do that first.
Furthermore, we will assume that all of our Nullstellensatz proofs are \emph{multilinearized}: that is, we work modulo the $x_i^2 - x_i = 0$ equations, and so the individual degree of any variable in the proof is at most $1$. 
It is well-known that making this assumption will not change the degree or size of the proof by more than a constant factor \cite{Buss1998}.

\begin{lemma}
	\label{lem:ppad-unary-ns}
	Let $F$ be an unsatisfiable CNF formula.
	If there is a depth-$d$ $\eol_L$-formulation of $S(F)$ then there is a unary Nullstellensatz refutation of $F$ with degree $O(d)$ and size $L2^{O(d)}$.
\end{lemma}
\begin{proof}
	Suppose $F \coloneqq C_1 \wedge \cdots \wedge C_m$ is on $n$ variables $x_1, \dots, x_n$, and let $\overline {C}_i$ be the negation of $C_i$ represented as a polynomial.
	Assume that there is a depth-$d$ $\eol_L$-formulation of $S(F)$.
	Let $V \coloneqq [L]$ be the set of nodes in the $\eol$ formulation and let $v^* = 1$ denote the distinguished source node.
	Each node $v \in V$ is equipped with successor and predecessor functions $s_v, p_v : \B^n \rightarrow V$, respectively, each computed by decision trees of depth at most $d$, as well as a solution decision tree $g_v : \B^n \rightarrow [m]$ that outputs a corresponding solution of $S(F)$.
	For any input assignment $x \in \B^n$ let $G_x$ denote the directed graph obtained by evaluating all the successor and predecessor decision trees on input $x$ and adding an edge $(u, v)$ iff $s_u(x) = v$ and $p_v(x) = u$.

	For each $v \in V$ define the function $S_v : \B^n \rightarrow \set{-1, 0, 1}$ by \[ S_v(x) \coloneqq \begin{cases} -1 & \text{if } v \neq v^* \text{ is a source in } G_x \\
		1 & \text{if } v \neq v^* \text{ is a proper sink in } G_x \text{ or } v = v^* \text{ and } v^* \text{ is not a source} \\
		0 & \text{otherwise}.
	\end{cases}\]
	We compute $S_v$ for each node $v$ by a depth at most $5d$ decision tree as follows.
	First, we compute $s_v(x) = u$ and $p_v(x) = w$, and then compute $p_u(x)$ and $s_w(x)$.
	From this information we can determine the output value of $S_v$, and we have used at most $4d$ queries.
	If $S_v = 0$ then the leaf of the decision tree is labelled with $0$.
	Otherwise, if $S_v \neq 0$ then $v$ is a solution to $\eol$, and so in this case we will also run the decision tree for $g_v$ and label each leaf with either $1$ or $-1$ according to the output value of $S_v$.
	Overall this requires at most $5d$ queries. 

	Now, for any leaf $\ell$ in the decision tree for $S_v$ let $D_\ell$ denote the polynomial representation of the conjunction of literals on the path from the root of the tree to $\ell$.
	Observe that we can represent \[ S_v = \sum_{(-1)\text{-leaf } \ell} - D_\ell + \sum_{\text{1-leaf } \ell} D_\ell,\] where the first sum is over leaves of $S_v$ labelled with $-1$ and the second is over leaves of $S_v$ labelled with $1$.
	If $\ell$ is a non-zero leaf then $v$ is a solution to the $\eol$ instance, so let $C_\ell$ denote the solution of $S(F)$ output by the decision tree $g_v$ at this leaf.
	Observe that at every non-zero leaf $\ell$, the clause $C_\ell$ must be falsified by the assignment on the path to $\ell$, since $C_\ell$ is a solution to $S(F)$ by the correctness of the $\eol$ formulation and by the fact that we ran the $g_v$ decision tree in $S_v$.
	This implies that for each non-zero leaf $\ell$ of $S_v$ we can write $D_\ell = D'_\ell\overline C_\ell$, and thus \[ S_v = \sum_{(-1)\text{-leaf } \ell} - D_\ell + \sum_{\text{1-leaf } \ell} D_\ell = \sum_{(-1)\text{-leaf } \ell} - D'_\ell \overline C_\ell + \sum_{\text{1-leaf } \ell} D'_\ell \overline C_\ell. \]
	If we sum up these polynomials for each $v \in V$ and gather terms then \[ \sum_{v \in V} S_v = \sum_{i=1}^m p_i \overline C_i \] for some polynomials $p_i$.
	Note that each polynomial has degree at most $5d$ since they are obtained from the underlying $S_v$ decision trees.

	To see that $\sum_{i=1}^m p_i\overline C_i$ is a unary Nullstellensatz refutation of $F$, observe that since each $S_v$ came from an $\eol$ formulation we have
	\[ \sum_{i=1}^m p_i(x) \overline C_i(x) = \sum_{v \in V} S_v(x) = (\#\text{ sinks in }G_x) - (\#\text{ non-distinguished sources in } G_x) = 1 \] for any input $x \in \B^n$.
	Finally, we observe that all coefficients used in this proof are integers, and the number of distinct monomials produced is at most $|V|2^{O(d)} = L2^{O(d)}$ from expanding the depth-$d$ decision trees as polynomials.
\end{proof}

The more difficult direction is the converse, proved next.

\begin{lemma}
	\label{lem:unary-ns-ppad}
	Let $F$ be an unsatisfiable CNF formula.
	If there is a unary Nullstellensatz refutation of $F$ with degree $d$ and size $L$ then there is a depth-$O(d)$ $\eol_{O(L)}$-formulation of $S(F)$.
\end{lemma}

\begin{figure}
    \centering
    \begin{tikzpicture}[y=-1cm, scale=1]
\newcommand{\hspacex}{0.7*\spacex}
\newcommand{\hspacey}{0.7*\spacey}

\newcommand{\polygonshiftsmall}{.3}

% node grid
\coordinate (p11) at (0, 0); 				 \coordinate (p12) at (\hspacex, 0); 				 \coordinate (p13) at (2*\hspacex, 0); 				\coordinate (p14) at (3.2*\hspacex, 0);    \coordinate (p15) at (4.2*\hspacex, 0);    \coordinate (p16) at (5.2*\hspacex, 0); \coordinate (p17) at (6.4*\hspacex, 0);    \coordinate (p18) at (7.4*\hspacex, 0);    \coordinate (p19) at (8.4*\hspacex, 0); 

% \coordinate (p21) at (0, \hspacey); 				 \coordinate (p22) at (\hspacex, \hspacey); 				 \coordinate (p23) at (2*\hspacex, \hspacey); 				\coordinate (p24) at (3.2*\hspacex, \hspacey);    \coordinate (p25) at (4.2*\hspacex, \hspacey);    \coordinate (p26) at (5.2*\hspacex, \hspacey); \coordinate (p27) at (6.4*\hspacex, \hspacey);    \coordinate (p28) at (7.4*\hspacex, \hspacey);    \coordinate (p29) at (8.4*\hspacex, \hspacey); 

\coordinate (p31) at (0, 2*\hspacey); 				 \coordinate (p32) at (\hspacex, 2*\hspacey); 				 \coordinate (p33) at (2*\hspacex, 2*\hspacey); 				\coordinate (p34) at (3.2*\hspacex, 2*\hspacey);    \coordinate (p35) at (4.2*\hspacex, 2*\hspacey);    \coordinate (p36) at (5.2*\hspacex, 2*\hspacey); \coordinate (p37) at (6.4*\hspacex, 2*\hspacey);    \coordinate (p38) at (7.4*\hspacex, 2*\hspacey);    \coordinate (p39) at (8.4*\hspacex, 2*\hspacey); 

\coordinate (p41) at (0, 3.2*\hspacey); 				 \coordinate (p42) at (\hspacex, 3.2*\hspacey); 				 \coordinate (p43) at (2*\hspacex, 3.2*\hspacey); 				\coordinate (p44) at (3.2*\hspacex, 3.2*\hspacey);    \coordinate (p45) at (4.2*\hspacex, 3.2*\hspacey);    \coordinate (p46) at (5.2*\hspacex, 3.2*\hspacey); \coordinate (p47) at (6.4*\hspacex, 3.2*\hspacey);    \coordinate (p48) at (7.4*\hspacex, 3.2*\hspacey);    \coordinate (p49) at (8.4*\hspacex, 3.2*\hspacey); 

\coordinate (p51) at (0, 4.2*\hspacey); 				 \coordinate (p52) at (\hspacex, 4.2*\hspacey); 				 \coordinate (p53) at (2*\hspacex, 4.2*\hspacey); 				\coordinate (p54) at (3.2*\hspacex, 4.2*\hspacey);    \coordinate (p55) at (4.2*\hspacex, 4.2*\hspacey);    \coordinate (p56) at (5.2*\hspacex, 4.2*\hspacey); \coordinate (p57) at (6.4*\hspacex, 4.2*\hspacey);    \coordinate (p58) at (7.4*\hspacex, 4.2*\hspacey);    \coordinate (p59) at (8.4*\hspacex, 4.2*\hspacey); 

\coordinate (p61) at (0, 5.2*\hspacey); 				 \coordinate (p62) at (\hspacex, 5.2*\hspacey); 				 \coordinate (p63) at (2*\hspacex, 5.2*\hspacey); 				\coordinate (p64) at (3.2*\hspacex, 5.2*\hspacey);    \coordinate (p65) at (4.2*\hspacex, 5.2*\hspacey);    \coordinate (p66) at (5.2*\hspacex, 5.2*\hspacey); \coordinate (p67) at (6.4*\hspacex, 5.2*\hspacey);    \coordinate (p68) at (7.4*\hspacex, 5.2*\hspacey);    \coordinate (p69) at (8.4*\hspacex, 5.2*\hspacey); 

% \coordinate (p71) at (0, 6.4*\hspacey); 				 \coordinate (p72) at (\hspacex, 6.4*\hspacey); 				 \coordinate (p73) at (2*\hspacex, 6.4*\hspacey); 				\coordinate (p74) at (3.2*\hspacex, 6.4*\hspacey);    \coordinate (p75) at (4.2*\hspacex, 6.4*\hspacey);    \coordinate (p76) at (5.2*\hspacex, 6.4*\hspacey); \coordinate (p77) at (6.4*\hspacex, 6.4*\hspacey);    \coordinate (p78) at (7.4*\hspacex, 6.4*\hspacey);    \coordinate (p79) at (8.4*\hspacex, 6.4*\hspacey); 

% \coordinate (p81) at (0, 7.4*\hspacey); 				 \coordinate (p82) at (\hspacex, 7.4*\hspacey); 				 \coordinate (p83) at (2*\hspacex, 7.4*\hspacey); 				\coordinate (p84) at (3.2*\hspacex, 7.4*\hspacey);    \coordinate (p85) at (4.2*\hspacex, 7.4*\hspacey);    \coordinate (p86) at (5.2*\hspacex, 7.4*\hspacey); \coordinate (p87) at (6.4*\hspacex, 7.4*\hspacey);    \coordinate (p88) at (7.4*\hspacex, 7.4*\hspacey);    \coordinate (p89) at (8.4*\hspacex, 7.4*\hspacey); 

% \coordinate (p91) at (0, 8.4*\hspacey); 				 \coordinate (p92) at (\hspacex, 8.4*\hspacey); 				 \coordinate (p93) at (2*\hspacex, 8.4*\hspacey); 				\coordinate (p94) at (3.2*\hspacex, 8.4*\hspacey);    \coordinate (p95) at (4.2*\hspacex, 8.4*\hspacey);    \coordinate (p96) at (5.2*\hspacex, 8.4*\hspacey); \coordinate (p97) at (6.4*\hspacex, 8.4*\hspacey);    \coordinate (p98) at (7.4*\hspacex, 8.4*\hspacey);    \coordinate (p99) at (8.4*\hspacex, 8.4*\hspacey); 

%Zones
\draw[gadget_ITER_small] ([shift={(-\polygonshiftsmall, -\polygonshiftsmall)}] p12) -- ([shift={(-\polygonshiftsmall, \polygonshiftsmall)}] p52) -- ([shift={(\polygonshiftsmall, \polygonshiftsmall)}] p53) -- ([shift={(\polygonshiftsmall, -\polygonshiftsmall)}] p13) -- cycle;

\draw[gadget_ITER_small] ([shift={(-\polygonshiftsmall, -\polygonshiftsmall)}] p14) -- ([shift={(-\polygonshiftsmall, \polygonshiftsmall)}] p54) -- ([shift={(\polygonshiftsmall, \polygonshiftsmall)}] p56) -- ([shift={(\polygonshiftsmall, -\polygonshiftsmall)}] p16) -- cycle;

\draw[gadget_ITER_small] ([shift={(-\polygonshiftsmall, -\polygonshiftsmall)}] p17) -- ([shift={(-\polygonshiftsmall, \polygonshiftsmall)}] p57) -- ([shift={(\polygonshiftsmall, \polygonshiftsmall)}] p59) -- ([shift={(\polygonshiftsmall, -\polygonshiftsmall)}] p19) -- cycle;

% node drawing
\tikzstyle{node_regular} = [node_regular_intro]

\node[node_a, label=left:{$v^*$}]       (P11) at (p11) {};
\node[node_regular] (P12) at (p12) {};
\node[node_regular] (P13) at (p13) {};
\node[node_regular] (P14) at (p14) {};
\node[node_regular] (P15) at (p15) {};
\node[node_regular] (P16) at (p16) {};
\node[node_regular] (P17) at (p17) {};
\node[node_regular] (P18) at (p18) {};
%\node[node_regular] (P19) at (p19) {};

%\node[node_regular] (P21) at (p21) {};
%\node[node_regular] (P22) at (p22) {};
%\node[node_regular] (P23) at (p23) {};
%\node[node_regular] (P24) at (p24) {};
%\node[node_regular] (P25) at (p25) {};
%\node[node_regular] (P26) at (p26) {};
%\node[node_regular] (P27) at (p27) {};
%\node[node_regular] (P28) at (p28) {};
%\node[node_regular] (P29) at (p29) {};

%\node[node_regular] (P31) at (p31) {};
\node[node_regular] (P32) at (p52) {};
\node[node_regular] (P33) at (p53) {};
\node[node_regular] (P34) at (p54) {};
\node[node_regular] (P35) at (p55) {};
\node[node_regular] (P36) at (p56) {};
\node[node_regular] (P37) at (p57) {};
\node[node_regular] (P38) at (p58) {};
\node[node_regular] (P39) at (p59) {};

\draw[edge_regular, MidnightBlue] (P32) -- (P12);
\draw[edge_regular, MidnightBlue] (P33) -- (P13);
\draw[edge_regular, MidnightBlue] (P34) -- (P16);
\draw[edge_regular, MidnightBlue] (P35) -- (P14);
\draw[edge_regular, MidnightBlue] (P36) -- (P15);

\draw[edge_regular] (P11) -- (P32);
\draw[edge_regular] (P12) -- (P34);
\draw[edge_regular] (P13) -- (P35);
\draw[edge_regular] (P14) -- (P38);
\draw[edge_regular] (P15) -- (P39);
\draw[edge_regular] (P16) -- (P37);
\draw[edge_regular] (P17) -- (P33);
\draw[edge_regular] (P18) -- (P36);

\node at (1.5*\hspacex, 5.2*\hspacey) {$V_1$};
\node at (p65) {$V_2$};
\node at (p68) {$V_3$};

\end{tikzpicture}
    \vspace{2mm}
    \caption{High level illustration of the \eol instance constructed in the proof of \cref{lem:unary-ns-ppad}. Edges of the outer matching are shown in black, while those from the inner matching are in blue. In this example, the clause corresponding to $V_3$ is not satisfied by assignment $x$ and as a result no internal edges are added in $V_3$. Note that \eol solutions indeed only occur in $V_3$.}
    \label{fig:enter-label}
\end{figure}

\begin{proof}
	Let $F = C_1 \wedge \cdots \wedge C_m$ and consider a degree-$d$, size-$L$ unary Nullstellensatz refutation of $F$, which we write as \[ \sum_{i=1}^m p_i\overline C_i = 1 \] where each $p_i$ is a multilinear polynomial over $x_1, \dots, x_n$ and all coefficients are integers.

	To build the $\eol$ formulation, we expand the above proof out into its constituent monomials with multiplicity.
	That is, for each $i \in [m]$ write the polynomial \[ p_i\overline C_i = \sum_{j} c_{i,j}q_{i,j}\] where $c_{i,j} \in \Z$ and $q_{i,j}$ is a monomial obtained by expanding the polynomial directly and performing all necessary cancellations.
	Each node in our $\eol$ formulation will represent one of the above monomials $q_{i,j}$ and is considered a ``$+$'' or a ``$-$'' node, depending on that monomial's sign.
	In total, we create $m + 1$ sets of nodes $V^*, V_1, \dots V_m$, defined as follows.
	The set $V^*$ only contains the distinguished source node $v^*$, which we consider as a ``$-$'' node.
	For each $i \in [m]$ the set $V_i$ contains a node for each monomial $q_{i,j}$ from the above expansion with multiplicity.
	So, in particular, we add $|c_{i,j}|$ copies of the monomial $q_{i,j}$ to $V_i$ for each monomial $q_{i,j}$ in the above expansion.
	Let $V$ denote the set of all nodes produced by this construction.
	For every node $v \in V$, the decision tree for $g_v$ will query no variables and output $C_i$ if $v \in V_i$ and an arbitrary clause if $v = v^*$; our construction will explicitly prevent the source node $v^*$ from being a solution.

	Now, we must describe the successor and predecessor decision trees $s_v, p_v$ at each node.
	It will be easier to describe the possible edges in $G_x$ on an input $x \in \B^n$; all of the edges are organized into two different matchings as detailed next. See \cref{fig:enter-label} for a high-level illustration.
	\begin{description}
    \item[Outer Matching.]
			In this matching we add edges between nodes in different node groups. 
			All directed edges will be oriented from ``$-$'' nodes to ``$+$'' nodes.
			Since the polynomials form a Nullstellensatz refutation over $\mathbb{Z}$, we know that each time the monomial $q$ appears with a ``$+$'' sign, it must also appear with a ``$-$'' sign, except for the single $1$ term. 
			Thus by treating the distinguished source $v^*$ as ``$-1$'', we can create a perfect matching $M$ on the nodes of $V$ where all matched nodes are between a ``$+$'' and a ``$-$'' node standing for the same monomial. 
			Since we have gathered terms within the expansions $p_i\overline C_i$, all occurrences of monomials $q$ within a set $V_i$ have the same sign, and thus all the edges in this matching will be between nodes in different sets.
			Formally, in the $\eol$ formulation, for each edge $e = (u,v)$ in $M$ corresponding to a monomial $q$, we add a directed edge from the ``$-$'' to the ``$+$'' node if and only if $q(x) = 1$. 
			This condition can be determined by $s_u$ and $p_v$ by querying the variables occurring in $q$. 

		\item[Inner Matching.]
			In this matching we add directed edges from ``$+$'' nodes to ``$-$'' nodes within the same node group.
			Consider any set $V_i$.
			Formally, at each node occurring in the group $V_i$, we query all variables of the corresponding clause $C_i$ in both the successor and predecessor functions for that node.
			For any $x \in \B^n$, if $C_i(x) = 1$ then $\overline C_i = 0$ and thus $p_i(x) \overline C_i(x) = 0$.
			This means that under the partial restriction $\rho$ consistent with $x$ at the variables of $C_i$, all monomials remaining in $p_i\overline C_i \restriction \rho$ must cancel.
			We can therefore fix a perfect matching between the negative and positive instances of monomials in $V_i$ under $\rho$, representing the cancellation of monomials under $\rho$. Then, each edge of this matching is included in the graph if and only if the monomials corresponding to its endpoints evaluate to $1$ at $x$ (note that the two endpoints will both evaluate to the same value, since they are matched under $\rho$).
			On the other hand, if $C_i(x) = 0$ then we will simply not add any edges to the internal matching of $V_i$.
	\end{description}
	Let $x \in \B^n$ be any assignment to the variables of $F$.
	The edges of any node $v  \in  V_i$ associated with a monomial $q$ are determined by querying the variables of $C_i$ and $q$.
	This implies that the depth of each decision tree $T_v$ is at most $d$, and the size is clearly $O(L)$ since every monomial in the proof is represented as a node.

	We now verify correctness of the $\eol$ formulation.
	Since it is well-defined, on every input $x$ the graph $G_x$ will have a solution.
	Let $v$ be such a solution (either a sink or proper source node) in $G_x$.
	By construction, $v \neq v^*$ since the node $v^*$ is always a source node.
	This implies that $v \in V_i$ for some $i \in [m]$, and so $v$ must be associated with a monomial $q$.
	By the construction of the inner and outer matching, $v$ can only be a source or sink node in $V_i$ if the inner matching is empty.
	But this can only happen if $C_i(x) = 0$, and thus $C_i$ is a valid solution to $S(F)$.
\end{proof}

\subsection{Unary Sherali--Adams and $\PPADS$}

We now show that low-degree unary Sherali--Adams proofs characterise $\PPADS^{dt}$.
The proof of this fact follows the proof from the previous section quite closely, but requires some extra work to handle the extra conical junta terms.

\begin{theorem}
	\label{thm:unarysa-ppads}
	Let $F$ be an unsatisfiable CNF formula.
	Then,
	\begin{itemize}
		\item If $F$ has a degree-$d$, size-$L$ unary Sherali--Adams proof, then $S(F)$ has a depth-$O(d)$ $\sol_{O(L)}$-formulation.
		\item If $S(F)$ has a depth-$d$ $\sol_{L}$-formulation, then $F$ has a degree-$O(d)$, size-$L2^{O(d)}$ unary Sherali--Adams proof.
	\end{itemize}
	In particular, $\PPADS^{dt}(S(F)) = \Theta(\textup{uSA}(F))$.
\end{theorem}

\begin{corollary}
	For any sequence $F$ of $\poly(\log n)$-width CNF formulas, $F$ has a $\poly(\log n)$-degree, $n^{\poly(\log n)}$-size unary Sherali--Adams proof if and only if $S(F) \in \PPADS^{dt}$.
\end{corollary}

Before we prove the theorem, it will be convenient to have the following simple normal form for Sherali--Adams proofs.
Just like in the previous section we will assume that all Sherali--Adams proofs are multilinearized, and it is known that this assumption does not change the degree or size of the proof by more than a constant factor \cite{Fleming2019}.

\begin{lemma}
	\label{lem:sa-normal-form}
	Let $F$ be an unsatisfiable CNF formula. 
	If $\sum_{i=1}^m p_i \overline C_i = 1 + J$ is a unary Sherali--Adams refutation of $F$ with degree $d$ and size $L$, then there is a degree-$d$, size-$L$ unary Sherali--Adams refutation of $F$ of the form $\sum_{i=1}^m J_i\overline C_i = 1 + J_0$ where $J_i$ is a conical junta for each $i = 0, 1, \dots, m$.
\end{lemma}
\begin{proof}
	For each $i \in [m]$ we can expand $p_i = \sum_{j} c_{i,j} q_{i,j}$ where $c_{i,j}$ are integers and $q_{i,j}$ are monomials. 
	Each monomial $q_{i,j}$ is a conjunction, so the expressions \[ J_i^{-} = \sum_{j: c_{i,j} < 0} |c_{i,j}|q_{i,j}, \quad J^+_i = \sum_{j: c_{i,j} > 0} c_{i,j}q_{i,j}\]
	are conical juntas for each $i \in [m]$.
	Writing $p_i = J_i^+ - J_i^-$, substituting into the Sherali--Adams refutation, and rearranging completes the proof.
\end{proof}

We now begin the proof of \cref{thm:unarysa-ppads}.
As before we split the proof into two lemmas, one for each direction of the characterisation.
The easier direction is again that an $\sol$-formulation implies a unary Sherali--Adams proof, and it almost exactly follows the proof of \cref{lem:ppad-unary-ns}.

\begin{lemma}
	Let $F$ be an unsatisfiable CNF formula.
	If there is a depth-$d$ $\sol_{L}$-formulation of $S(F)$ then there is a unary Sherali--Adams refutation of $F$ with degree $O(d)$ and size $L2^{O(d)}$.
\end{lemma}
\begin{proof}
	The proof of this lemma is essentially the same as the proof of \cref{lem:ppad-unary-ns}, so we will simply sketch it and note what needs to be modified.
	Suppose $F \coloneqq C_1 \wedge \cdots \wedge C_m$ and let $\overline {C}_i$ be the negation of $C_i$ represented as a polynomial. 
	We have an $\sol$-formulation for $S(F)$, and so we have decision trees computing successors $s_v$ and predecessors $p_v$ for each of the nodes $v \in V$.
	As in the proof of \cref{lem:ppad-unary-ns}, for each $v \in V$ we define a depth at most $5d$ decision tree $S_v$, defined by \[ S_v(x) = \begin{cases} 1 & \text{if } v \neq v^* \text{ is a source in } G_x \\
		-1 & \text{if either } v \text{ is a proper sink in } G_x \text{ or  if } v = v^* \text{ and } v^* \text{ is not a source} \\
		0 & \text{otherwise},
	\end{cases}\]
	where we note that we have switched the ``$-1$'' and the ``$+1$'' in the definition of $S_v$ when compared to \cref{lem:ppad-unary-ns}.
	As before, $S_v(x)$ can be determined by first running the decision trees for $s_v(x) = u$ and $p_v(x) = w$, then the decision trees for $p_u(x), s_w(x)$, and finally the decision tree for $g_v(x)$ if the node $v$ is a solution to $\sol$.
	From this, we can again represent
	\[ S_v = \sum_{(-1)\text{-leaf } \ell} - D_\ell + \sum_{\text{1-leaf } \ell} D_\ell,\] where the first sum is over leaves of $S_v$ labelled with $-1$ and the second is over leaves of $S_v$ labelled with $1$.
	However, now a node $v$ is only a solution to $\sol$ if $S_v(x) = -1$, and so for each $(-1)$-leaf $\ell$ of $S_v$ we can write $D_\ell = D'_\ell \cdot \overline C_\ell$ where $C_\ell$ is the clause of $F$ falsified at that leaf.
	This allows us to write \[ S_v = \sum_{(-1)\text{-leaf } \ell} - D_\ell + \sum_{\text{1-leaf } \ell} D_\ell = \sum_{(-1)\text{-leaf } \ell} - D'_\ell\cdot \overline C_v + \sum_{\text{1-leaf } \ell} D_\ell. \]
	If we sum up these polynomials for each $v \in V$ and gather terms we get \[ \sum_{v \in V} S_v = \sum_{i=1}^m -J_i \overline C_i + J_0 \] for some degree-$O(d)$ conical juntas $J_0, J_1, \dots, J_m$.
	As in the proof of \cref{lem:ppad-unary-ns} we have that $\sum_{v} S_v(x) = -1$ and the size and degree calculations are identical. 
\end{proof}

The proof of the converse direction is also similar to the proof of \cref{lem:unary-ns-ppad}, but requires some more substantial modification when compared to the previous proof.
The main issue is how to handle the extra conical junta terms $J_0$ in the unary Sherali--Adams refutation.
As in the proof of \cref{lem:unary-ns-ppad}, we will create a graph representing all the monomials in the unary Sherali--Adams proof.
However, we will do some extra work to ensure that the nodes corresponding to monomials from the conical junta term $J_0$ will always be source nodes.
This ensures that any solutions will occur at nodes corresponding to some falsified clause in the formula.
\begin{lemma} 
	Let $F$ be an unsatisfiable CNF formula.
	If there is a unary Sherali--Adams refutation of $F$ with degree $d$ and size $L$ then there is a degree-$O(d)$ $\sol_{O(L)}$-formulation of $S(F)$.
\end{lemma}
\begin{proof}
	Suppose $F \coloneqq C_1 \wedge \cdots \wedge C_m$ is on $n$ variables and consider a unary Sherali--Adams refutation \[ \sum_{i=1}^m -J_i\overline C_i + J_0 = -1 \] of $F$ where each $J_i$ for $i = 0, 1, \ldots, m$ are integral conical juntas.
	For notational convenience, we will let $\overline C_0 \coloneqq -1$ and we will expand each conical junta $J_i$ as a non-negative sum of conjunctions. 
	While this notation is somewhat unusual, it allows us to write the refutation in a uniform way as \[ \sum_{i=1}^m -J_i\overline C_i + J_0 = \sum_{i=0}^m \sum_{j=1}^{t_i} -\lambda_{i,j} D_{i,j} \overline C_i = -1\]
	where $t_i$ is a non-negative integer, $\lambda_{i,j}$ is a positive integer, and $D_{i,j}$ is a conjunction for every $i, j$.

	To build the $\sol$ formulation, we expand the above proof out into its constituent monomials with multiplicity.
	As in the proof of \cref{lem:unary-ns-ppad}, each node in our $\sol$ formulation will represent a monomial in the proof and is either a ``$+$'' or a ``$-$'' node, depending on that monomial's sign.
	This time, however, we create a group of nodes $V_{i,j}$ for each $i = 0, 1, \dots, m$ and each $j \in [t_i]$, as well as a special group $V^*$.
	The group $V^*$ only contains the distinguished node $v^*$, which we now consider as a ``$+$'' node.
	On the other hand, for each $i, j$, the group $V_{i,j}$ will correspond to the polynomial $-\lambda_{i,j}D_{i,j}\overline C_{i,j}$. 
	We expand this polynomial into a sum of monomials $ -\lambda_{i,j} D_{i,j}\overline C_{i,j} = \sum_q c_q q$ for some integers $c_q$ and monomials $q$, and for each monomial $q$ in this expansion we create $|c_q|$ nodes in $V_{i,j}$, each of which are ``$+$'' nodes if $c_q > 0$ 
	and ``$-$'' nodes otherwise.
	Let $V$ denote the set of all nodes produced by this construction.
	For any node $v \in V$, if $v \in V_{i,j}$ for some $i > 0$ then the solution decision tree $g_v$ will query no variables and simply output $C_i$	as the solution to $S(F)$.
	Otherwise, $g_v$ will output an arbitrary solution, as in this case by construction of the formulation the node $v$ will never be a solution to $\sol$.

	Now, we must describe the successor and predecessor decision trees at each node.
	As in the proof of \cref{lem:unary-ns-ppad}, it will be easier to describe the possible edges in $G_x$ as all of the edges are organized into two different matchings. 
	\begin{description}
		\item[Outer Matching.] 
			The definition of the outer matching is the same as in \cref{lem:unary-ns-ppad}.
			In this matching we add edges between nodes in different node groups. 
			All directed edges will be oriented from ``$+$'' nodes to ``$-$'' nodes.
			Since the polynomials form an $\SA$ refutation over $\mathbb{Z}$, we know that each time the monomial $q$ appears with a ``$+$'' sign, it must also appear with a ``$-$'' sign, except for the single $-1$ term. 
			Thus by considering $v^*$ as ``$+1$'', we can create a perfect matching $M$ of the nodes of $V$ where all edges are between a ``$+$'' and a ``$-$'' node standing for the same monomial. 
			Since we have gathered terms within the expansions of $-\lambda_{i,j} D_{i,j} \overline C_i$, all occurrences of monomials $q$ within a single group $V_{i,j}$ have the same sign and thus all the matchings are between nodes in different sets.
			For each edge $e$ in $M$, we will add a directed edge between the ``$+$'' and the ``$-$'' node if and only if $q(x) = 1$; this can be determined by querying all variables in $q$.

		\item[Inner Matching.]
			The inner matching is constructed similarly as in the proof of \cref{lem:unary-ns-ppad}, but requires some modification.
			As in that proof, in the inner matching we add directed edges from ``$-$'' nodes to ``$+$'' nodes within the same node group.
			However, we will now be careful to force any solution (i.e.~a sink node) to occur at a ``$-$'' node in $G_x$.
			By our construction, the $V^*$ group has no ``$-$'' nodes, and all ``$-$'' nodes in the group $V_{0, j}$ for any $j \in [t_0]$ will have successors, and thus any sink node must be associated with $V_{i,j}$ for some~$i > 0$.

			Consider any set of the form $V_{i,j}$, since $V^*$ has a single node corresponding to $+1$ and so no internal edges will be matched.
			Formally, at each node occurring in the group $V_{i,j}$, we query all variables of $\overline C_i$ and $D_{i,j}$ (note that when $i = 0$, $\overline C_0 = 1$ and so we only query $D_{i,j}$ variables).
			For any assignment $x \in \B^n$, if $C_i(x) = 1$ then $\overline C_i(x) = 0$ and thus $D_{i,j}(x) \overline C_i(x) = 0$.
			This means that under the partial restriction $\rho$ consistent with $x$ at the variables of $C_i$, all monomials in $D_{i,j}\overline C_i \restriction \rho$ must cancel to $0$.
			We can therefore fix a directed perfect matching between the negative and positive copies of monomials in $V_{i,j}$, as in the proof of \cref{lem:unary-ns-ppad}.
			
			On the other hand, if $\overline C_i(x) \neq 0$ then $-\lambda_{i,j}D_{i,j}(x)\overline C_i(x) = c$ for some integer $c$.
			If $i > 0$ then $c \leq 0$, and so in this case, there will be $|c|$ copies of ``$-$'' monomials	in $V_i$ that are not cancelled by $+$ monomials internally.
			We can then fix a directed partial matching between monomials accordingly, but leaving the $|c|$ ``$-$'' monomials without successors if required (these will become sink nodes).
			If $i = 0$ then $c \geq 0$ since $\overline C_i = -1$, and so in this case there may be more ``$+$'' monomials than ``$-$'' monomials evaluating to $1$.
			We can therefore fix a directed partial matching between monomials, now leaving some ``$+$'' monomials without predecessors (these will become new source nodes), but all ``$-$'' monomials will have successors and so they will not become proper sink nodes.
	\end{description}
	As we have described above, we will need at most $d$ queries in any decision tree in the reduction, and also the number of nodes in the final $\sol$ instance is no more than the size (number of monomials) of the underlying unary Sherali--Adams proof.

	We finally verify correctness of the $\sol$-formulation.
	This is a well-defined $\sol$ formulation and thus on every input $x \in \B^n$ the graph $G_x$ will have a solution $v \in V$.
	This must be a sink node by the definition of $\sol$ and therefore, by construction, $v$ must be a ``$-$'' node since ``$+$'' nodes always have successors by the construction of the outer matching.
	As we have described in the definition of the inner matching, any ``$-$'' node $v \in V_{0, j}$ for any $j$ will have a successor, and thus $v \in V_{i, j}$ for some $i > 0$.
	But then, by definition of the inner matching, if $v$ is a sink node in $V_{i,j}$ for $i > 0$ then $C_i(x) = 0$ and the label of $v$ is $C_i$, thus the $\sol$ formulation correctly outputs a solution to $S(F)$.
\end{proof}

\subsection{Reversible Resolution, $\SOPL$, and $\EOPL$} \label{sec:char-revres}

In this section we define the Reversible Resolution systems (RevRes and RevResT), and prove our final characterisations capturing $\SOPL^{dt}$ and $\EOPL^{dt}$.

\begin{theorem}
	\label{thm:revres-sopl}
	Let $F$ be an unsatisfiable CNF formula.
	Then,
	\begin{itemize}
		\item If $F$ has a width-$d$, size-$L$ Reversible Resolution proof (with Terminals, resp.), then $S(F)$ has a depth-$O(d)$ $\sopl_{O(L)}$-formulation ($\eopl$-formulation, resp.).
		\item If $S(F)$ has a depth-$d$ $\sopl_L$-formulation ($\eopl_L$-formulation, resp.), then $F$ has a width-$O(d)$, size-$L^{O(1)}2^{O(d)}$ Reversible Resolution proof (with Terminals, resp.).
	\end{itemize}
	In particular, $\SOPL^{dt}(S(F)) = \Theta(\textup{RevRes}(F))$ and $\EOPL^{dt}(S(F)) = \Theta(\textup{RevResT}(F))$.
\end{theorem}

\begin{corollary}
	For any sequence $F$ of $\poly(\log n)$-width CNF formulas, $F$ has a $\poly(\log n)$-width, $n^{\poly(\log n)}$-size Reversible Resolution proof (with Terminals, resp.) if and only if $S(F) \in \SOPL^{dt}$ ($S(F) \in \EOPL^{dt}$, resp.).
\end{corollary}

\subsubsection*{Reversible Resolution and MaxSAT}
We begin by formally defining Reversible Resolution refutations and comparing them to MaxSAT systems from the literature \cite{Bonet2007,Larrosa2008,Filmus2023}.

\begin{definition}
	Let $F$ be an unsatisfiable CNF formula.
	If $C$ is a clause then the \emph{reversible weakening rule} is the proof rule $C \vdash C \vee x, C \vee \overline x$, and the \emph{reversible resolution rule} is the proof rule $C \vee x, C \vee \overline x \vdash C$.
	A \emph{reversible resolution refutation} (RevRes) of $F$ is a sequence of multisets of clauses $\mathcal{C}_1, \mathcal{C}_2, \ldots, \mathcal{C}_t$ such that the following holds:
	\begin{enumerate}
		\item Every clause in $\mathcal{C}_1$ occurs in $F$, possibly with multiplicity.
		\item The multiset $\mathcal{C}_t$ contains the empty clause $\bot$.
		\item For each $i = 1, 2, \ldots, t-1$, the multiset $\mathcal{C}_{i+1}$ is obtained from $\mathcal{C}_i$ by selecting clauses in $\mathcal{C}_{i}$ and \emph{replacing} them with the result of one of the two reversible rules applied to those clauses.
	\end{enumerate}
	The proof is a \emph{reversible resolution refutation with terminals} (RevResT) if every clause in $\mathcal{C}_t$ other than $\bot$ is a weakening of a clause from $F$.
	The size of the proof is $\sum_{i=1}^t |\mathcal{C}_i|$ --- the number of clauses in all configurations.
	The width of the proof is the maximum width of any clause occuring in any configuration.
\end{definition}

The key difference between the reversible resolution rule and the standard resolution rule is that the output of the reversible rule (as a CNF formula) is logically equivalent to the input of the rule.
Despite this restriction, it is clear that we can use the reversible rule to simulate tree-like resolution.
If we use clauses $C \lor x$ and $D \lor \overline x$ to derive $C \lor D$, then we can derive this in RevRes as follows.
First, for each literal in $C \lor x$, apply the reversible weakening rule repeatedly to derive $C \lor D \lor x$ (along with some extra clauses which we can ignore).
Similarly, derive $D \lor C \lor \overline x$. 
Then apply the reversible resolution rule to these two clauses to derive $C \lor D$.

However, \cref{thm:res-sa} implies that RevRes \emph{cannot} efficiently simulate Resolution.
Intuitively, this is because of property (3) in the definition of a reversible refutation: we must \emph{replace} the clauses used in the rule with new clauses.
Therefore we cannot ``duplicate'' derived clauses for free, which is essential to obtain the full power of Resolution.

Indeed, the RevRes proof system is a slight strengthening of the proof system \emph{MaxSAT Resolution with Weakening} (also denoted $\textup{MaxResW}$) studied in the literature on MaxSAT solvers \cite{Bonet2007,Larrosa2008,Filmus2023}.
The principal difference between MaxSAT Resolution and standard Resolution is that MaxSAT Resolution seeks to preserve the number of satisfied clauses under any assignment.
For completeness, we define the MaxSAT Resolution proof system next.
\begin{definition}
	Let $A = a_1 \lor \cdots \lor a_s$ and $B = b_1 \lor \cdots \lor b_t$ be clauses over boolean literals $a_i, b_j$.
	The \emph{MaxSAT resolution rule} is the proof rule that, given $x \vee A$ and $\overline x \vee B$, deduces the following set of clauses:
	\begin{align*}
		&a_1 \lor \cdots \lor a_s \lor b_1 \lor \cdots \lor b_t \\
		&x \lor A \lor \bigvee_{i=1}^j b_i \lor \overline b_{j+1} \quad \forall j = 0, 1, \dots, t \\
		&\overline x \lor B \lor \bigvee_{i=1}^j a_i \lor \overline a_{j+1} \quad \forall j = 0, 1, \dots, s.
	\end{align*}
	A \emph{MaxRes refutation} of an unsatisfiable CNF $F$ is a sequence of multisets of clauses $\mathcal{C}_1, \dots, \mathcal{C}_t$ where $\mathcal{C}_1$ contains exactly the clauses in $F$, $\mathcal{C}_t$ contains a copy of the empty clause $\bot$, and the configuration $\mathcal{C}_i$ for $i > 1$ is obtained from $\mathcal{C}_{i-1}$ by applying the MaxSAT resolution rule to some clauses in $\mathcal{C}_{i-1}$ and replacing those clauses with the output of the rule.
	A \emph{MaxResW refutation} is a $\textup{MaxRes}$ refutation that is also allowed to use the weakening rule $C \vdash C \vee x, C \vee \overline x$.
\end{definition}

RevRes can simulate MaxResW proofs without much difficulty.
The weakening rule in MaxResW is the reversible weakening rule.
To simulate the MaxSAT resolution rule, starting from $x \vee A, \overline x \vee B$, apply the reversible weakening rule on $x \vee A$ to weaken it with the variable $b_1$, obtaining $x \vee A \vee b_1, x \vee A \vee \overline b_1$.
Then, weaken $x \vee A \vee b_1$ on the variable $b_2$ to obtain the clauses $x \vee A \vee b_1 \vee b_2, x \vee A \vee b_1 \vee \overline b_2$.
Repeating in this fashion on all literals in $B$, and similarly weakening $\overline x \vee B$, we obtain $x \vee A \vee B$, $\overline x \vee A \vee B$, and all the extra clauses output by the MaxSAT rule.
Finally applying the reversible resolution rule to $x \vee A \vee B$ and $\overline x \vee A \vee B$ deduces $A \vee B$.

The converse direction, however, is not clear and could very well be false.
A significant difference between RevRes and MaxResW is the fact that MaxResW proofs must have the initial configuration \emph{exactly} equal to $F$, while RevRes can start with any multiset of clauses from $F$.
As discussed above, this is because the goal of MaxRes is to preserve the number of satisfied clauses under any assignment, while RevRes has no such requirements and simply seeks to prove unsatisfiability.

We can formally interpret this as follows.
Suppose we are given an unsatisfiable CNF formula $F = C_1 \land \cdots \land C_m$, where every clause $C_i$ is equipped with a positive integer weight $w_i$.
Since $F$ is unsatisfiable, the maximum possible weight of satisfied clauses in any assignment to the variables of $F$ is at most $\sum_{i=1}^m w_i - 1$.
Thus, if we could \emph{prove} that this is true for some choice of weights $w_i > 0$, then we have verified that the formula $F$ is unsatisfiable.

RevRes implements this idea. 
Given $F$, we start by choosing positive integer weights $w_i$ for each clause $C_i$, and make $w_i$ copies of $C_i$ in the initial configuration $\mathcal{C}_1$.
The two proof rules in RevRes preserve the number of satisfied clauses under any assignment, and so it follows that if $\mathcal{C}_1, \dots, \mathcal{C}_t$ is a RevRes refutation of $F$ then, since $\mathcal{C}_t$ contains at least one instance of $\bot$, it must be that the maximum weight of satisfied clauses under any assignment is at most $\sum_{i=1}^m w_i - 1$ since $\bot$ is always false.
Hence the formula $F$ must be unsatisfiable.
Interpreted in this way, RevRes sits between MaxResW and the weighted MaxSAT resolution systems defined in \cite{Larrosa2008}.

\subsubsection*{Characterisation theorems}

Unlike the characterisation theorems for unary Nullstellensatz and unary Sherali--Adams, the easier direction for this characterisation theorem is showing that RevRes proofs imply $\sopl$-formulations.

\begin{lemma}
	\label{lem:revres-sopl}
	Let $F$ be an unsatisfiable CNF formula.
	If there is a RevRes refutation of $F$ with width $d$ and size $L$, then there is a depth-$(d+1)$ $\sopl_L$-formulation of $S(F)$.
	Furthermore, if there is a RevResT refutation, then there is a depth-$(d+1)$ $\eopl_L$-formulation of $S(F)$.
\end{lemma}
\begin{proof}
	We focus on the case of RevRes and then describe what needs to be modified in the case of RevResT.
	Let $F = C_1 \land \cdots \land C_m$ be an unsatisfiable CNF formula.
	Let $\mathcal{C}_1, \mathcal{C}_2, \ldots, \mathcal{C}_\ell$ be a RevRes refutation of $F$ of the prescribed size and width and let $t \coloneqq \max_{i \in [\ell]} |\mathcal{C}_i|$.
	By the size bound we know that $t, \ell \leq L$.

	We create an $\sopl$-formulation of $S(F)$ on a grid of size $L \times L$, although we will only use the subgrid of size $\ell \times t$ and hardwire all other nodes to be inactive.
	This can be done for each node $(i, j)$ outside of the $\ell \times t$ grid by setting the successor for $(i,j)$ to be $\nul$ and the predecessor to be arbitrary.
	The relationship between the grid of the $\sopl$-formulation and the RevRes proof is straightforward: the node $(i, j) \in [\ell] \times [t]$ corresponds to the $j$-th clause in the multiset $\mathcal{C}_{\ell-i+1}$.
	Without loss of generality, we assume $\mathcal{C}_l$ is ordered so that the first clause is $\bot$, and thus the distinguished node $(1,1)$ in the $\sopl$ instance corresponds to $\bot$.

	Let $(i, j) \in [\ell] \times [t]$ be any node in the grid and let $C_{i,j}$ denote the corresponding clause in the proof.
	We define the successor function $s_{i,j} : \B^n \rightarrow [t] \cup \set{\nul}$, the predecessor function $p_{i,j} : \B^n \rightarrow [t]$, and the solution function $g_{i,j} : \B^n \rightarrow [m]$.
	The solution function $g_{i,j}$ queries no variables and outputs $C_{i,j}$ if $C_{i,j} \in F$, and otherwise outputs an arbitrary solution (in the second case, by construction $(i,j)$ will never be a solution to $\sopl$).
	To define $s_{i,j}$ and $p_{i,j}$ we introduce some notation.
	If $C \in \mathcal{C}_i$ and $C' \in \mathcal{C}_{i+1}$ are clauses in adjacent configurations then $C'$ is \emph{derived from} $C$, written $C \vdash C'$, if either $C'$ is the output of a reversible proof rule applied to $C$ or if no proof rule was applied to $C$ and $C' = C$ is just the same copy of $C$ in the next configuration.
	For any $x \in \B^n$ define
	\[ 
	s_{i,j}(x) \coloneqq \begin{cases}
		\nul & \text{ if } C_{i,j}(x) = 1 \\
		k & \text{ if } i < \ell, C_{i,j}(x) = C_{i+1, k}(x) = 0, \text{ and } C_{i+1,k} \vdash C_{i,j} \\
		1 & \text{ if } i = \ell \text{ and } C_{i,j}(x) = 0,
	\end{cases}
	\]
	and similarly, if $i > 1$, define
	\[ 
	p_{i,j}(x) \coloneqq \begin{cases}
		1 & \text{ if } C_{i,j}(x) = 1 \\
		k & \text{ if } C_{i,j}(x) = C_{i-1, k}(x) = 0 \text{ and } C_{i,j} \vdash C_{i-1,k}.
	\end{cases}
	\]
	Intuitively, if $C_{i,j}(x) = 0$ then we will make the successor and predecessors of $C_{i,j}$ point to the unique clauses in the adjacent configurations that are guaranteed to be false.
	These functions are well-defined since the reversible rules are of the form $C \vee x_i, C \vee \overline x_i \vdash C$ and $C \vdash C \vee x_i, C \vee \overline x_i$.
	In particular, under any assignment to the variables, the number of false clauses in the input and output of the rules are equal and at most $1$, and thus if $C$ is false then there are unique false clauses in the adjacent configurations that are derived from or used to derive $C$.
	Finally, we note that the successor and predecessor functions can each be computed by querying all the variables in $C_{i,j}$ and possibly one more variable (the one that was resolved or weakened on), and thus the decision tree depth of both of these functions is at most $d+1$.

	Now we argue that the $\sopl$ formulation correctly solves $S(F)$.
	By the definition of the successor and predecessor functions, if any node $(i,j)$ on layer $i < \ell$ is active, then that node has consistent pointers to successor nodes and predecessor nodes on the adjacent layers.
	This means that the node $(i,j)$ is a solution only if it is an active node on layer $i = \ell$, but such a node is active only if the corresponding clause $C_{i, j} \in \mathcal{C}_1$ is false. 
	But all such clauses occur in $F$, and in this case the solution function $g_{i,j}$ outputs $C_{i,j}$, which is a correct solution to $S(F)$.

	In case we started with a RevResT refutation, we observe that the same argument described above also works for $\eopl$ with one extra observation: any clause in the final configuration $\mathcal{C}_t$ that is falsified under an input $x$ is now a weakening of an input clause of $F$, and so this is a valid source node solution to the $\eopl$ problem.
\end{proof}

It remains to prove the converse, which is harder.
As a warmup, we begin by showing that the encoding of $\sopl$ ($\eopl$) as an unsatisfiable CNF formula can be efficiently refuted in RevRes (RevResT resp.).
The general case will follow the structure of this proof closely.
For the warmup it will be helpful to explicitly write the CNF encoding of $\sopl$  and $\eopl$ (\cref{sec:definitions}).
	
\paragraph{Explicit Encodings for $\sopl$ and $\eopl$.}
As we have discussed in \cref{sec:intro-tfnp}, any total search problem $R_n \subseteq \B^n \times O_n$ has a natural encoding as an unsatisfiable CNF formula by $\bigwedge_{o \in O_n} \neg T_o(x)$ where $T_o(x)$ is the decision tree that checks if $(x, o) \in R_n$.
Since $T_o$ is a low-depth decision tree we can encode it as a low-width DNF formula, and thus the resulting CNF formula also has low width.
In this section we describe the unsatisfiable CNF formulas corresponding to $\sopl_n$ and $\eopl_n$ explicitly.

The successor and predecessor pointers in the $\sopl_n$ instance will be encoded in binary, so, for the sake of convenience assume $n = 2^{\lambda} - 1$ for some integer $\lambda \geq 1$ and other cases can be handled similarly.
For each node $(i, j)$ the successor and predecessor pointers will be encoded by blocks of boolean variables $s_{i,j} \in \B^{\lambda}, p_{i,j} \in \B^\lambda$ encoding the value of the pointer in binary.
The pointer $\nul$ will always be encoded by the all-$0$ string.
We will abuse notation and often consider $s_{i,j}$ and $p_{i,j}$ as actual elements of $[n] \cup \set{\nul}$, rather than as short boolean strings.
So, we may write things like $s_{i,j} = k$ for $k \in [n]$ to mean that the bits of $s_{i,j}$ are equal to the binary encoding of $k$.

As everything is encoded in binary, it will be helpful to introduce the following notation.
In general, for a predicate $P : \B^n \rightarrow \B$ we let $\llbracket P \rrbracket$ represent the CNF encoding of $P$ over the $n$ underlying boolean variables.
For example, $\llbracket s_{i,j} = \ell \rrbracket$ for $\ell \in [n]$ represents the CNF encoding of the predicate ``$s_{i, j} = \ell$'' over the boolean variables underlying $s_{i,j}$.
Explicitly, $\llbracket s_{i,j} = \nul \rrbracket = \bigwedge_{t=1}^{\lambda} \overline s_{i,j,t},$ and similarly $\llbracket s_{i,j} \neq \nul \rrbracket$ can be represented by the clause $\bigvee_{t=1}^{\lambda} s_{i,j,t}$.
We can also form more complicated statements, writing e.g.~$\llbracket s_{i,j} = k \wedge p_{i+1, k} = j \rrbracket$ to mean the CNF encoding of ``the successor of $(i,j)$ is $(i+1, k)$ and the predecessor of $(i+1,k)$ is $(i,j)$''.

\begin{definition}\label{def:SOPL-CNF}
	Let $n$ be a positive integer, and for simplicity assume $n = 2^\lambda - 1$ for some integer $\lambda \geq 1$.
	Consider the following unsatisfiable CNF formula $\sopl_n$.
	For each $(i, j) \in \set{2, \dots, n-1} \times [n]$ we have two blocks of $\lambda$ variables $s_{i,j} \in \B^\lambda, p_{i,j} \in \B^\lambda$ encoding the successor and predecessor pointers of the node $(i,j)$ in binary, where $\nul$ is encoded by $0^\lambda$.
	For each $j \in [n]$, we additionally have a block of $\lambda$ variables $s_{1,j} \in \B^\lambda$ encoding the successor of $(1, j)$, a block of $\lambda$ variables $p_{n, j} \in \B^\lambda$ encoding the predecessor of $(n, j)$, and a single variable $s_{n,j} \in \B$ 
	encoding whether or not $(n, j)$ is active.

	The clauses of $\sopl_n$ are the following:
	\begin{itemize}
		\item For each $j \in [n]$, $\llbracket s_{1,1} \neq j \vee p_{2,j} = 1\rrbracket$ and $\llbracket s_{1,1} \neq 0 \rrbracket$ \hfill \emph{(active distinguished source)}
		\item For each $j \in [n]$, $\overline s_{n,j}$ for each $j \in [n]$, \hfill \emph{(inactive sink)}
		\item For each $(i, j) \in \set{1, \dots, n-2} \times [n]$ and each $a, b \in [n]$, $c \in [n] \cup \{0\}$, $a \neq c$, \hfill \emph{(no proper sinks)}
			\[\llbracket s_{i,j} \neq a \lor p_{i+1, a} \neq j \lor s_{i+1,a} \neq b \lor p_{i+2,b} \neq c \rrbracket\]
   as well as $\llbracket s_{i,j} \neq a \lor p_{i+1, a} \neq j \lor s_{i+1,a} \neq 0 \rrbracket$.
			Similarly, for each $a, b \in [n]$, \[ \llbracket s_{n-1, a} \neq b \lor p_{n, b} \neq a \lor s_{n,b} = 1 \rrbracket\]
	\end{itemize}
	The $\eopl_n$ formula is obtained by adding the following extra clauses to $\sopl_n$:
	\begin{itemize}
		\item For each $(i, j) \in \set{2, \dots, n-1} \times [n]$ and each $a, b \in [n]$, $c \in [n] \cup \{0\}$, $c \neq j$, \hfill \emph{(no proper sources)}
			\[\llbracket s_{i,j} \neq a \lor p_{i+1,a} \neq j \lor p_{i,j} \neq b \lor s_{i-1,b} \neq c \rrbracket\]
   as well as $\llbracket s_{i,j} \neq a \lor p_{i+1,a} \neq j \lor p_{i,j} \neq 0 \rrbracket$.
			Similarly, for any $a, b \in [n]$ with $a \neq 1$, \[ \llbracket s_{1, a} \neq b \lor p_{2,b} \neq a \rrbracket.\]
	\end{itemize}
\end{definition}

From the above definition we can see that both $\sopl_n$ and $\eopl_n$ are polynomial-size, $O(\log n)$-width CNF formulas, and they are unsatisfiable since the families of clauses simply encode the contradictory statements ``the $\sopl/\eopl$ problem has no solution''.

\paragraph{Proofs of Characterisations.}
Now, before proving that we can refute $\sopl_n$ in RevRes, we first prove a technical lemma that allow us to manipulate binary encodings in RevRes.

\begin{lemma}
	\label{lem:reversible-weakening}
	Let $\lambda > 0$ be a positive integer, and let $n = 2^{\lambda} - 1$.
	Let $C$ be a width-$k$ clause that does not depend on a block of boolean variables $z \in \B^\lambda$.
	Using the reversible weakening rule we can prove, from $C$, the set of clauses $\set{\llbracket C \lor z \neq i \rrbracket : i = 0, \dots, n}$ in width $k + \lambda$ and size $2^{\lambda}$. 
	Conversely, from the above set of clauses we can prove $C$ using the reversible resolution rule in the same size and width.
\end{lemma}
\begin{proof}
	Starting from $C$, apply the reversible weakening rule on the first bit $z_1$ to obtain $C \vee z_1$ and $C \vee \overline z_1$.
	Weakening each of the results on $z_2$, $z_3$, \dots, $z_\lambda$ in turn yields exactly the CNF formula described in the lemma, and the second statement follows from the reversibility of RevRes.
\end{proof}

\begin{theorem} \label{thm:sopl-in-revres}
	For each positive integer $n$, there is a $O(\log n)$-width, polynomial-size RevRes refutation (RevResT refutation, resp.) of $\sopl_n$ ($\eopl_n$, resp.).
\end{theorem}
\begin{proof}
	We give the proof for $\sopl_n$ and then describe what needs to be modified for $\eopl_n$.
	For each $(i,j) \in [n-1] \times [n]$ and each $k \in [n]$ define the clause $I_{i,j,k} \coloneqq \llbracket s_{i,j} \neq k \lor p_{i+1, k} \neq j \rrbracket$, and note that $I_{i,j,k}$ has width $2\log n$ in the variables $s_{i,j}$ and $p_{i+1,k}$.
	With this notation, the set of clauses
	\[ I_{i,j} \coloneqq \set{\llbracket s_{i,j} \neq k \lor p_{i+1,k} \neq j\rrbracket \st k \in [n]}\] encodes the statement ``the node $(i,j)$ is inactive''.
	Similarly, for any $j \in [n]$ we define \[ I_{n,j} \coloneqq \overline s_{n,j}\] encoding that the node $(n, j)$ is inactive, and note that $I_{n,j}$ is a clause in $\sopl_n$.
	Thus, for any $i \in [n]$, the collection of clauses $\mathcal{I}_i \coloneqq \bigcup_{j=1}^n I_{i,j}$ encodes the statement ``every node on layer $i$ is inactive''.
	We now state the main claim of the proof.

	\begin{claim}
		\label{claim:main}
		For any $i \in \set{2, \dots, n}$, there is a polynomial-size, $O(\log n)$-width RevRes proof of $\mathcal{I}_{i-1}$ from $\mathcal{I}_{i}$ and a polynomial-size collection of clauses from $\sopl_n$.
	\end{claim}

	Let us first use the claim to finish the proof of the theorem.	
	We start with the collection of clauses $\mathcal{I}_n = \bigcup_{j=1}^n I_{n, j}$, each of which is a clause from $\sopl_n$.
	Applying the claim yields the collection $\mathcal{I}_{n-1}$ in polynomial-size and $O(\log n)$ width from $\mathcal{I}_n$ and a polynomial-size collection of clauses from $\sopl_n$.
	Applying the claim $n-2$ more times then yields $\mathcal{I}_1$ in polynomial-size and $O(\log n)$ width.
	However, the clauses $I_{1,1} \subseteq \mathcal{I}_1$ are exactly \[ \llbracket s_{1,1} \neq j \vee p_{2, j} \neq 1 \rrbracket\] for each $j \in [n]$.
	By resolving these clauses with the clauses in $\llbracket s_{1,1} \neq j \vee p_{2, j} = 1 \rrbracket$ in $\sopl_n$, we can deduce the family of clauses $\llbracket s_{1,1} \neq j \rrbracket$ for all $j \neq 0$, and the clause $\llbracket s_{1,1} \neq 0 \rrbracket$ is already in $\sopl_n$.
	Applying \cref{lem:reversible-weakening} to the clauses $\set{\llbracket s_{1,1} \neq i \rrbracket \st i = 0, \dots, n}$ deduces the empty clause $\bot$ in $O(\log n)$ width and $O(n)$ size.
	In sum, the entire proof will have polynomial size and $O(\log n)$ width.

	So, it suffices to prove the claim.
	\begin{proof}[Proof of Claim.]
		We show how to prove the general case where $i \leq n-1$, and the case where $i = n$ is handled by an essentially identical argument.
		Consider the family of clauses $\mathcal{I}_{i} = \bigcup_{j=1}^n I_{i,j}$.
		For each clause $I_{i,j,k}$ apply \cref{lem:reversible-weakening} to weaken as follows.
		Initially, we weaken over all values of the predecessor pointer $p_{i, j}$, obtaining the family of clauses $\llbracket I_{i,j,k} \vee p_{i,j} \neq a \rrbracket$ for each $a \in [n]$.
		Then, from the clause in this family containing $p_{i, j} \neq a$, we weaken over all values of the successor pointer $s_{i-1, a}$, obtaining the family of clauses 
		\begin{align*}
			\mathcal{A}_i & = \set{\llbracket I_{i,j,k} \vee p_{i,j} \neq a \vee s_{i-1,a} \neq b\rrbracket \st j,k,a,b \in [n]} \\
			& = \set{\llbracket s_{i,j} \neq k \vee p_{i+1, k} \neq j \vee p_{i,j} \neq a \vee s_{i-1,a} \neq b \rrbracket \st j,k,a,b \in [n]}.
		\end{align*}
		Partition this family of clauses into two sets as follows.
			Define \[\mathcal{T}_i = \set{\llbracket s_{i,j} \neq k \lor p_{i+1,k} \neq j \lor p_{i,j} \neq a \lor s_{i-1,a} \neq j \rrbracket \st j, k, a \in [n]},\]
		which is the subfamily of clauses in $\mathcal{A}_i$ that have $b = j$, and let $\mathcal{J}_i = \mathcal{A}_i \setminus \mathcal{T}_i$ denote the subfamily where $b \neq j$.
		Next, we show how to use $\mathcal{T}_i$, along with some clauses in $\sopl_n$, to deduce $\mathcal{I}_{i-1}$ in width $O(\log n)$ and polynomial size.
		The clauses $\mathcal{J}_i$ are ``junk'' clauses that are maintained for the rest of the proof and output along with the bottom clause $\bot$ in the final configuration.

		\begin{figure}
		\centering
		\begin{subfigure}[b]{0.3\textwidth}
			\centering
			\begin{tikzpicture}[y=-1cm, scale=1]
	% node grid
	\coordinate (p11) at (-0.5*\spacex, -0.5*\spacey);
	\coordinate (p12) at (0.5*\spacex, -0.5*\spacey);
	\coordinate (p13) at (1.5*\spacex, -0.5*\spacey);
	\coordinate (p14) at (2.5*\spacex, -0.5*\spacey);
	\coordinate (p15) at (3.5*\spacex, -0.5*\spacey);

	\coordinate (p31) at (-0.5*\spacex, 0.5*\spacey);
	\coordinate (p32) at (0.5*\spacex, 0.5*\spacey);
	\coordinate (p33) at (1.5*\spacex, 0.5*\spacey);
	\coordinate (p34) at (2.5*\spacex, 0.5*\spacey);
	\coordinate (p35) at (3.5*\spacex, 0.5*\spacey);

	\coordinate (p51) at (-0.5*\spacex, 1.5*\spacey);
	\coordinate (p52) at (0.5*\spacex, 1.5*\spacey);
	\coordinate (p53) at (1.5*\spacex, 1.5*\spacey);
	\coordinate (p54) at (2.5*\spacex, 1.5*\spacey);
	\coordinate (p55) at (3.5*\spacex, 1.5*\spacey);

	\coordinate (p71) at (-0.5*\spacex, 2.5*\spacey);
	\coordinate (p72) at (0.5*\spacex, 2.5*\spacey);
	\coordinate (p73) at (1.5*\spacex, 2.5*\spacey);
	\coordinate (p74) at (2.5*\spacex, 2.5*\spacey);
	\coordinate (p75) at (3.5*\spacex, 2.5*\spacey);

	\coordinate (p91) at (-0.5*\spacex, 3.5*\spacey);
	\coordinate (p92) at (0.5*\spacex, 3.5*\spacey);
	\coordinate (p93) at (1.5*\spacex, 3.5*\spacey);
	\coordinate (p94) at (2.5*\spacex, 3.5*\spacey);
	\coordinate (p95) at (3.5*\spacex, 3.5*\spacey);

	\coordinate (p21) at (0, 0); 				 \coordinate (p22) at (\spacex, 0); 				 \coordinate (p23) at (2*\spacex, 0); 				\coordinate (p24) at (3*\spacex, 0);

	\coordinate (p41) at (0, \spacey);
	\coordinate (p42) at (\spacex, \spacey); 	 \coordinate (p43) at (2*\spacex, \spacey); 	\coordinate (p44) at (3*\spacex, \spacey);

	\coordinate (p61) at (0, 2*\spacey);
	\coordinate (p62) at (\spacex, 2*\spacey); \coordinate (p63) at (2*\spacex, 2*\spacey); \coordinate (p64) at (3*\spacex, 2*\spacey);

	\coordinate (p81) at (0, 3*\spacey);
	\coordinate (p82) at (\spacex, 3*\spacey); \coordinate (p83) at (2*\spacex, 3*\spacey); \coordinate (p84) at (3*\spacex, 3*\spacey);

	%Zones
	%\begin{scope}[fill opacity=0.5]
		\draw[gadget_ITER_small] (p31) -- (p35) -- (p95) -- (p91) -- cycle;
		%\draw[gadget_ITER_small] (p31) -- (p35) -- (p55) -- (p51) -- cycle;
		%\draw[gadget_ITER_small] (p14) -- (p94) -- (p93) -- (p13) -- cycle;
	%\end{scope}

	% node drawing

	%\node[node_regular]  (P11) at (p21) {};
	%\node[node_regular]  (P12) at (p22) {};
	%\node[node_regular]  (P13) at (p23) {};
	%\node[node_regular] (P14) at (p24) {};
	\node[node_regular, label={[label distance=12*\spacey]90:{$j$}}]  (P21) at (p41) {};
	\node[node_regular, label={[label distance=12*\spacey]90:{$k$}}]  (P22) at (p42) {};
	\node[node_regular, label={[label distance=12*\spacey]90:{$a$}}]  (P23) at (p43) {};
	\node[node_regular, label={[label distance=12*\spacey]90:{$b$}}]  (P24) at (p44) {};
	\node[node_regular]  (P31) at (p61) {};
	\node[node_regular]  (P32) at (p62) {};
	\node[node_regular]  (P33) at (p63) {};
	\node[node_regular]  (P34) at (p64) {};
	\node[node_regular]  (P41) at (p81) {};
	\node[node_regular]  (P42) at (p82) {};
	\node[node_regular]  (P43) at (p83) {};
	\node[node_regular]  (P44) at (p84) {};

	% edges drawing
	\draw[edge_eopl_small] (P31) to [bend right=15] (P42);
	\draw[edge_php_small] (P42) to [bend right=15] (P31);
	\draw[edge_php_small] (P31) to (P23);
	\draw[edge_eopl_small] (P23) to (P34);

\end{tikzpicture}
			\vspace{2mm}
			\caption{A clause in $\mathcal{J}_i \subseteq \mathcal{A}_i$.}
			\label{figure:claim3All}
		\end{subfigure}%
		\begin{subfigure}[b]{0.3\textwidth}
			\centering
			\begin{tikzpicture}[y=-1cm, scale=1]
	% node grid
	\coordinate (p11) at (-0.5*\spacex, -0.5*\spacey);
	\coordinate (p12) at (0.5*\spacex, -0.5*\spacey);
	\coordinate (p13) at (1.5*\spacex, -0.5*\spacey);
	\coordinate (p14) at (2.5*\spacex, -0.5*\spacey);
	\coordinate (p15) at (3.5*\spacex, -0.5*\spacey);

	\coordinate (p31) at (-0.5*\spacex, 0.5*\spacey);
	\coordinate (p32) at (0.5*\spacex, 0.5*\spacey);
	\coordinate (p33) at (1.5*\spacex, 0.5*\spacey);
	\coordinate (p34) at (2.5*\spacex, 0.5*\spacey);
	\coordinate (p35) at (3.5*\spacex, 0.5*\spacey);

	\coordinate (p51) at (-0.5*\spacex, 1.5*\spacey);
	\coordinate (p52) at (0.5*\spacex, 1.5*\spacey);
	\coordinate (p53) at (1.5*\spacex, 1.5*\spacey);
	\coordinate (p54) at (2.5*\spacex, 1.5*\spacey);
	\coordinate (p55) at (3.5*\spacex, 1.5*\spacey);

	\coordinate (p71) at (-0.5*\spacex, 2.5*\spacey);
	\coordinate (p72) at (0.5*\spacex, 2.5*\spacey);
	\coordinate (p73) at (1.5*\spacex, 2.5*\spacey);
	\coordinate (p74) at (2.5*\spacex, 2.5*\spacey);
	\coordinate (p75) at (3.5*\spacex, 2.5*\spacey);

	\coordinate (p91) at (-0.5*\spacex, 3.5*\spacey);
	\coordinate (p92) at (0.5*\spacex, 3.5*\spacey);
	\coordinate (p93) at (1.5*\spacex, 3.5*\spacey);
	\coordinate (p94) at (2.5*\spacex, 3.5*\spacey);
	\coordinate (p95) at (3.5*\spacex, 3.5*\spacey);

	\coordinate (p21) at (0, 0); 				 \coordinate (p22) at (\spacex, 0); 				 \coordinate (p23) at (2*\spacex, 0); 				\coordinate (p24) at (3*\spacex, 0);

	\coordinate (p41) at (0, \spacey);
	\coordinate (p42) at (\spacex, \spacey); 	 \coordinate (p43) at (2*\spacex, \spacey); 	\coordinate (p44) at (3*\spacex, \spacey);

	\coordinate (p61) at (0, 2*\spacey);
	\coordinate (p62) at (\spacex, 2*\spacey); \coordinate (p63) at (2*\spacex, 2*\spacey); \coordinate (p64) at (3*\spacex, 2*\spacey);

	\coordinate (p81) at (0, 3*\spacey);
	\coordinate (p82) at (\spacex, 3*\spacey); \coordinate (p83) at (2*\spacex, 3*\spacey); \coordinate (p84) at (3*\spacex, 3*\spacey);

	%Zones
	%\begin{scope}[fill opacity=0.5]
		\draw[gadget_ITER_small] (p31) -- (p35) -- (p95) -- (p91) -- cycle;
		%\draw[gadget_ITER_small] (p31) -- (p35) -- (p55) -- (p51) -- cycle;
		%\draw[gadget_ITER_small] (p14) -- (p94) -- (p93) -- (p13) -- cycle;
	%\end{scope}

	% node drawing

	%\node[node_regular]  (P11) at (p21) {};
	%\node[node_regular]  (P12) at (p22) {};
	%\node[node_regular]  (P13) at (p23) {};
	%\node[node_regular] (P14) at (p24) {};
	\node[node_regular, label={[label distance=12*\spacey]90:{$j$}}]  (P21) at (p41) {};
	\node[node_regular, label={[label distance=12*\spacey]90:{$k$}}]  (P22) at (p42) {};
	\node[node_regular, label={[label distance=12*\spacey]90:{$a$}}]  (P23) at (p43) {};
	\node[node_regular, label={[label distance=12*\spacey]90:{$b$}}]  (P24) at (p44) {};
	\node[node_regular]  (P31) at (p61) {};
	\node[node_regular]  (P32) at (p62) {};
	\node[node_regular] (P33) at (p63) {};
	\node[node_regular]  (P34) at (p64) {};
	\node[node_regular]  (P41) at (p81) {};
	\node[node_regular]  (P42) at (p82) {};
	\node[node_regular]  (P43) at (p83) {};
	\node[node_regular]  (P44) at (p84) {};

	% edges drawing
	\draw[edge_eopl_small] (P23) to [bend right=10] (P31);
	\draw[edge_php_small] (P31) to [bend right=10] (P23);
	\draw[edge_eopl_small] (P31) to (P42);
	\draw[edge_php_small] (P42) to (P34);

\end{tikzpicture}
			\vspace{2mm}
			\caption{A clause in $\mathcal{F}_i \subseteq \mathcal{B}_i$.}
			\label{figure:claim3B}
		\end{subfigure}%
		\begin{subfigure}[b]{0.3\textwidth}
			\centering
			\begin{tikzpicture}[y=-1cm, scale=1]
	% node grid
	\coordinate (p11) at (-0.5*\spacex, -0.5*\spacey);
	\coordinate (p12) at (0.5*\spacex, -0.5*\spacey);
	\coordinate (p13) at (1.5*\spacex, -0.5*\spacey);
	\coordinate (p14) at (2.5*\spacex, -0.5*\spacey);
	\coordinate (p15) at (3.5*\spacex, -0.5*\spacey);

	\coordinate (p31) at (-0.5*\spacex, 0.5*\spacey);
	\coordinate (p32) at (0.5*\spacex, 0.5*\spacey);
	\coordinate (p33) at (1.5*\spacex, 0.5*\spacey);
	\coordinate (p34) at (2.5*\spacex, 0.5*\spacey);
	\coordinate (p35) at (3.5*\spacex, 0.5*\spacey);

	\coordinate (p51) at (-0.5*\spacex, 1.5*\spacey);
	\coordinate (p52) at (0.5*\spacex, 1.5*\spacey);
	\coordinate (p53) at (1.5*\spacex, 1.5*\spacey);
	\coordinate (p54) at (2.5*\spacex, 1.5*\spacey);
	\coordinate (p55) at (3.5*\spacex, 1.5*\spacey);

	\coordinate (p71) at (-0.5*\spacex, 2.5*\spacey);
	\coordinate (p72) at (0.5*\spacex, 2.5*\spacey);
	\coordinate (p73) at (1.5*\spacex, 2.5*\spacey);
	\coordinate (p74) at (2.5*\spacex, 2.5*\spacey);
	\coordinate (p75) at (3.5*\spacex, 2.5*\spacey);

	\coordinate (p91) at (-0.5*\spacex, 3.5*\spacey);
	\coordinate (p92) at (0.5*\spacex, 3.5*\spacey);
	\coordinate (p93) at (1.5*\spacex, 3.5*\spacey);
	\coordinate (p94) at (2.5*\spacex, 3.5*\spacey);
	\coordinate (p95) at (3.5*\spacex, 3.5*\spacey);

	\coordinate (p21) at (0, 0); 				 \coordinate (p22) at (\spacex, 0); 				 \coordinate (p23) at (2*\spacex, 0); 				\coordinate (p24) at (3*\spacex, 0);

	\coordinate (p41) at (0, \spacey);
	\coordinate (p42) at (\spacex, \spacey); 	 \coordinate (p43) at (2*\spacex, \spacey); 	\coordinate (p44) at (3*\spacex, \spacey);

	\coordinate (p61) at (0, 2*\spacey);
	\coordinate (p62) at (\spacex, 2*\spacey); \coordinate (p63) at (2*\spacex, 2*\spacey); \coordinate (p64) at (3*\spacex, 2*\spacey);

	\coordinate (p81) at (0, 3*\spacey);
	\coordinate (p82) at (\spacex, 3*\spacey); \coordinate (p83) at (2*\spacex, 3*\spacey); \coordinate (p84) at (3*\spacex, 3*\spacey);

	%Zones
	%\begin{scope}[fill opacity=0.9]
		\draw[gadget_ITER_small] (p31) -- (p35) -- (p95) -- (p91) -- cycle;
		%\draw[gadget_ITER_small] (p31) -- (p35) -- (p55) -- (p51) -- cycle;
		%\draw[gadget_ITER_small] (p14) -- (p94) -- (p93) -- (p13) -- cycle;
	%\end{scope}

	% node drawing

	%\node[node_regular]  (P11) at (p21) {};
	%\node[node_regular]  (P12) at (p22) {};
	%\node[node_regular]  (P13) at (p23) {};
	%\node[node_regular] (P14) at (p24) {};
	\node[node_regular, label={[label distance=12*\spacey]90:{$j$}}]  (P21) at (p41) {};
	\node[node_regular, label={[label distance=12*\spacey]90:{$k$}}]  (P22) at (p42) {};
	\node[node_regular, label={[label distance=12*\spacey]90:{$a$}}]  (P23) at (p43) {};
	\node[node_regular, label={[label distance=12*\spacey]90:{$b$}}]  (P24) at (p44) {};%
	\node[node_regular]  (P31) at (p61) {};
	\node[node_regular]  (P32) at (p62) {};
	\node[node_regular] (P33) at (p63) {};
	\node[node_regular]  (P34) at (p64) {};
	\node[node_regular]  (P41) at (p81) {};
	\node[node_regular]  (P42) at (p82) {};
	\node[node_regular]  (P43) at (p83) {};
	\node[node_regular]  (P44) at (p84) {};

	% edges drawing
	\draw[edge_eopl_small] (P23) to [bend right=10] (P31);
	\draw[edge_php_small] (P31) to [bend right=10] (P23);
	\draw[edge_eopl_small] (P31) to [bend right=15] (P42);
	\draw[edge_php_small] (P42) to [bend right=15] (P31);

\end{tikzpicture}
			\vspace{2mm}
			\caption{A clause in $\mathcal{T}_i = \mathcal{A}_i \cap \mathcal{B}_i$.}
			\label{figure:claim3T}
		\end{subfigure}
		\caption{Illustration of the objects in the proof of \autoref{claim:main}. The blue edges are successor pointers and the red edges are predecessor pointers. Each clause says that at least one of the pointers in the above configuration must not be present.}
		\label{fig:claim_main}
	\end{figure}
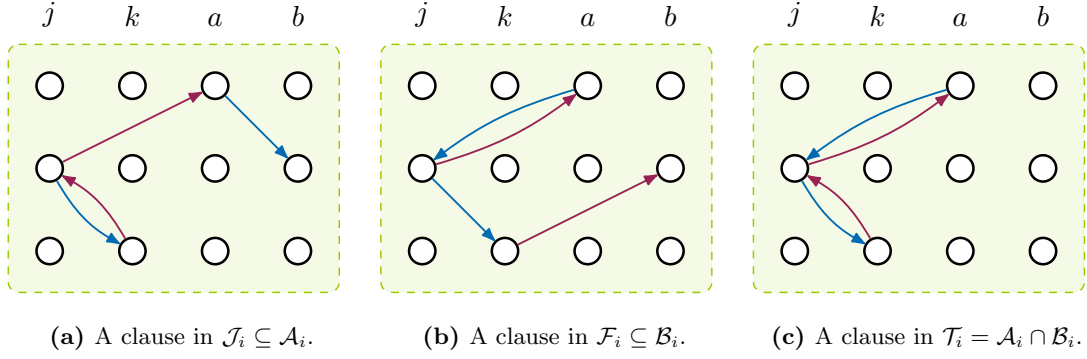

		To do this, we exploit the reversibility of RevRes and show how to deduce from $\mathcal{I}_{i-1}$ the collection $\mathcal{T}_i \cup \mathcal{F}_i$ using the reversible weakening rule, where $\mathcal{F}_i$ is a polynomial-size set of clauses all from $\sopl_n$.
		By running this proof in reverse and connecting it with the proof described above we prove $\mathcal{I}_{i-1}$ from $\mathcal{I}_{i}$, and we can add the clauses $\mathcal{F}_i$ to the initial configuration of the RevRes proof.
		
		This proof is very similar to the proof of $\mathcal{A}_i$ from $\mathcal{I}_{i}$.
		Starting from an arbitrary clause $I_{i-1,a,j} \in \mathcal{I}_{i-1}$, we apply \cref{lem:reversible-weakening} to weaken the clause on all possible values of the successor pointer $s_{i, j}$, obtaining $\llbracket I_{i-1,a,j} \vee s_{i, j} \neq k \rrbracket$ for all $k \in [n]$.
		Then, starting from the clause containing $s_{i, j} \neq k$, we weaken on all values of $p_{i+1, k}$, obtaining the family 
		\begin{align*}
			\mathcal{B}_i & = \set{\llbracket I_{i-1,a,j} \vee s_{i, j} \neq k \lor p_{i+1, k} \neq b \rrbracket \st j,k,a,b \in [n]} \\
			& = \set{\llbracket s_{i-1,a} \neq j \vee p_{i,j} \neq a \vee s_{i, j} \neq k \lor p_{i+1, k} \neq b \rrbracket \st j,k,a,b \in [n]}
		\end{align*}
		We again partition into two sets.
		The first is, of course, $\mathcal{T}_i$, which is the case where $j = b$ in $\mathcal{B}_i$.
		The second set is $\mathcal{F}_i = \mathcal{B}_i \setminus \mathcal{T}_i$, which is the case where $j \neq b$, and observe that every clause in $\mathcal{F}_i$ is a \emph{no proper sink} clause from $\sopl_n$.

		We can now finish the proof of the claim.
		Starting from $\mathcal{I}_{i} \cup \mathcal{F}_i$, use the clauses in $\mathcal{I}_{i}$ to deduce the set of clauses $\mathcal{T}_i \cup \mathcal{F}_i \cup \mathcal{J}_i$.
		Then, run the proof deducing $\mathcal{T}_i \cup \mathcal{F}_i$ from $\mathcal{I}_{i-1}$ in reverse to finally deduce $\mathcal{I}_{i-1} \cup \mathcal{J}_i$.
		The total proof has polynomial size and $O(\log n)$-width, and therefore the claim is proved.
	\end{proof}

	To modify the proof for $\eopl_n$ and RevResT, we make the following changes. 	
	First, we observe that all clauses in $\mathcal{I}_1$ that come from $I_{1, a}$ for $a \neq 1$ are proper source clauses from $\eopl_n$.
	All other clauses in the above proof that occur in the final line come from sets of the form $\mathcal{J}_i$ in the proof of the above claim.
	However, just like the clauses $\mathcal{F}_i$ are proper sink clauses from $\sopl_n$, the clauses in $\mathcal{J}_i$ are exactly proper source clauses from the $\eopl_n$.
	This completes the proof.
\end{proof}

It remains to modify the previous proof in order to accommodate decision-tree reductions to $\sopl$ and $\eopl$.
To do this we mimic the previous proof, but replace the construction of the sets of clauses in the proof with appropriate queries to the decision trees (which RevRes can simulate) in the reduction.

Before we prove the theorem we introduce some helpful notation for manipulating decision trees.
If $T$ is a decision tree then $\mathcal{P}(T)$ is the set of root-to-leaf paths in $T$.
If $o$ is an output (i.e.~leaf label) of $T$, then define $\mathcal{P}_o(T)$ to be the set of root-to-leaf paths in $T$ that output $o$.
Given any path $P \in \mathcal{P}(T)$, let $C_P \coloneqq \bigvee_{\ell \in P} \neg \ell$ be the negation of the literals along $P$; so, $C_P(x) = 1$ iff $P$ is not followed when $T$ is evaluated on $x$.
We also need an appropriate modification of \cref{lem:reversible-weakening} to arbitrary decision trees, which we prove next.
\begin{lemma}
	\label{lem:dt-weakening}
	Let $C$ be a width-$k$ clause, and let $T$ be a depth-$d$ decision tree querying a set of variables disjoint from $C$.
	Using the reversible weakening rule we can prove, from $C$, the set of clauses $\set{C \vee C_P \st P \in \mathcal{P}(T)}$ in width $d + k$ and size at most $2^{d}$.
	Conversely, from the above set of clauses we can prove $C$ using the reversible resolution rule in the same size and width.
\end{lemma}
\begin{proof}
	This proof is essentially the same as in \cref{lem:reversible-weakening}.
	Now, starting from $C$, apply the reversible weakening rule on the first variable $x_i$ queried in the decision tree $T$ to derive the clauses $\set{C \vee x_i, C \vee \overline x_i}$.
	From there we can continue to apply the reversible weakening rule to simulate the queries of the decision tree.
	For instance, if after the decision tree learns $x_i = 0$ it queries $x_j$, we apply the reversible weakening rule to $x_i$ to obtain $x_i \vee x_j, x_i \vee \overline x_j$.
	Continuing in this manner we can derive all clauses $C \vee C_P$ for $P \in \mathcal{P}(T)$, and running the proof in reverse yields the lemma.
\end{proof}

\begin{theorem}
	Let $F$ be an unsatisfiable CNF formula.
	If there is a depth-$d$ $\sopl_L$-formulation ($\eopl_L$-formulation, resp.) of $S(F)$ then there is a RevRes refutation (with terminals, resp.) of $F$ with width $O(d)$ and size $L^{O(1)}2^{O(d)}$.
\end{theorem}
\begin{proof}
	We follow the proof of \cref{thm:sopl-in-revres} and focus on the case of $\sopl$.
	Assume $F = C_1 \land \cdots \land C_m$ is defined on $n$ variables $x_1, \dots, x_n$.
	In this proof we think of CNF formulas and sets of clauses interchangeably.
	In the $\sopl_L$-formulation of $S(F)$ we have functions \[s_{i,j} : \B^n \rightarrow [L] \cup \set{\nul},\ p_{i,j}: \B^n \rightarrow [L] \cup \set{\nul},\  g_{i,j} : \B^n \rightarrow [m]\] computing successors, predecessors, and solutions for each internal node, and we identify each function with the depth-$d$ decision tree computing it.

	For each $(i, j) \in [L-1] \times [L]$ consider the CNF formula \[I_{i,j,k}(x) = \llbracket s_{i,j}(x) \neq k \vee p_{i+1,k}(x) \neq j \rrbracket. \] 
	In other words, $I_{i,j,k}$ is the analogue of the clause using the same notation from the proof of \cref{thm:sopl-in-revres}.
	We can use the decision trees for $s_{i,j}$ and $p_{i,j}$ to encode $I_{i,j,k}$ as a CNF formula explicitly.
	To do this, define the decision tree $T_{i,j}$ as follows: take the decision tree $s_{i,j}$ and at each leaf labelled $k$, simulate the decision tree $p_{i+1,k}$ (skipping queries to variables already made) to obtain an output $a$, and then output the pair $(k, a)$.
	With this decision tree we can define $I_{i,j,k} = \set{C_P \st P \in \mathcal{P}_{(k,j)}(T_{i,j})}$.
	As in the proof of \cref{thm:sopl-in-revres}, define \[ I_{i,j} := \bigcup_{k=1}^n I_{i,j,k}, \quad \mathcal{I}_i := \bigcup_{j=1}^n I_{i,j},\] where we recall that we consider CNFs and sets of clauses interchangeably.
	When $i = L$, then for any $j \in [L]$ define the decision tree $T_{L,j}$ that simulates the decision tree $s_{L,j}$ and outputs $1$ if $(L,j)$ is active and $0$ otherwise.
	With this we define $I_{L,j} = \set{C_P \st P \in \mathcal{P}_1(T_{L,j})}$, and similarly define $\mathcal{I}_L = \bigcup_{j=1}^L I_{L,j}$.
	In this notation, the set of clauses $\mathcal{I}_i$ again encodes ``every node on layer $i$ is inactive'', where now the activity of a node is determined by the underlying decision trees in the formulation.

		The main step in this theorem is the following claim.
	\begin{claim}
		For any $i \in \set{2,\dots,L}$, there is a size $L^{O(1)}2^{O(d)}$, $O(d)$-width RevRes proof of $\mathcal{I}_{i-1}$ from $\mathcal{I}_{i}$ and a collection of weakenings of clauses from $F$.
	\end{claim}

	First we use the claim to finish the proof of the theorem.
	We begin by deriving from $F$ the clauses $\mathcal{I}_{L}$ (let us briefly postpone this argument), and then apply the claim $L-1$ times to derive $\mathcal{I}_1$.
	Let $\mathcal{Q} = \bigcup_{k \neq 0} \mathcal{P}_{(k,1)}(T_{1,1})$ be the set of paths of $T_{1,1}$ that end in a leaf labelled with $(k, 1)$ for some $k \neq 0$, and let $\mathcal{R} = \mathcal{P}(T_{1,1}) \setminus \mathcal{Q}$.
	Observe that $I_{1,1} \subseteq \mathcal{I}_1$ is, by definition, the set of clauses $\set{C_P \st P \in \mathcal{Q}}$.

	Consider any path $P \in \mathcal{R}$, and note that $P$ ends in a leaf labelled with $(k, a)$ where either $k = 0$ or $a \neq 1$. 
	Each leaf witnesses that the distinguished node $(1,1)$ is inactive, and so we can then simulate the decision tree $g_{1,1}$ and learn a solution of $S(F)$. 
	Therefore, for every path $P' \in \mathcal{P}(g_{1,1})$ the clause $C_P \vee C_{P'}$ is either a weakening of a clause in $F$, or, is trivially true if it contains both a literal and its negation. 
	Therefore, by applying \cref{lem:dt-weakening} we can deduce the clause $C_P$ from weakenings of clauses in $F$ in size $2^{O(d)}$ and width $O(d)$.
	Applying this argument for every $P \in \mathcal{R}$ allows us to deduce the clauses $\set{C_P \st P \in \mathcal{R}}$.
	We have now deduced all the clauses $\set{C_P \st P \in \mathcal{P}(T_{1,1})}$, and so applying \cref{lem:dt-weakening} to all of these clauses allows us to deduce $\bot$.

	Let us now describe how to derive from $F$ the clauses \[\mathcal{I}_{L} = \set{\llbracket (L,j) \text{ is inactive} \rrbracket \st j \in [L]} = \bigcup_{j = 1}^n \set{C_P \st P \in \mathcal{P}_1(T_{L,j})}.\]
	For any $j \in [L]$ consider the following decision tree $T'_{L,j}$: first run the decision tree $T_{L,j}$ that checks if $(L,j)$ is active and then, if $(L,j)$ is active, simulate the decision tree $g_{L,j}$ to find a solution to $S(F)$.
	It follows that for any $P \in \mathcal{P}_1(T_{L,j})$ and any $P' \in \mathcal{P}(g_{L,j})$ the clause $C_P \vee C_{P'}$ is a weakening of a clause of $F$ or is trivially true.
	We can therefore deduce $C_P$ from weakenings of clauses of $F$ using \cref{lem:dt-weakening}, and repeating this argument for every $j \in [L]$ and every $P \in \mathcal{P}_j$ we can derive every clause in $\mathcal{I}_L$.
	So, all that remains is to prove the claim.

	\begin{proof}[Proof of Claim.]
		The proof of this claim is modelled on the proof of the similar claim from the previous theorem.
		We again do the general case where $i \leq L-1$; the case where $i = L$ proceeds similarly.
		Consider the set of clauses $\mathcal{I}_{i}$ and $\mathcal{I}_{i-1}$.
		Our first goal is to derive the analogue of the set $\mathcal{A}_i$ in the proof of \cref{claim:main}.

		Let $j \in [L]$ be arbitrary and consider any clause $C \in I_{i, j}$.
		By definition, there is a $k \neq 0$ such that $C = C_P$ for some $P \in \mathcal{P}_{(k,j)}(T_{i, j})$.
		Starting from $C$ in the proof apply \cref{lem:dt-weakening} to the decision tree $p_{i,j}$ to derive a set of clauses, each of the form $C \vee C_P$, where $P \in \mathcal{P}(p_{i,j})$.
		Then, for every $a \in [L]$ and any $P' \in \mathcal{P}_a(p_{i,j})$, apply \cref{lem:dt-weakening} again to $C \vee C_P$ and the decision tree $s_{i-1, a}$ to obtain $C \vee C_P \vee C_{P'}$ for every $P' \in \mathcal{P}(s_{i-1,a})$.
		Performing this procedure for all $C \in \mathcal{I}_i$ yields
		\begin{align*}
			\mathcal{A}_i & \coloneqq \set{\llbracket s_{i,j} \neq k \vee p_{i+1, k} \neq j \vee p_{i,j} \neq a \vee s_{i-1,a} \neq b \rrbracket \st j,k,a,b \in [n]} \\
			& = \set{C \vee C_P \vee C_{P'} \st a \in [L], P \in \mathcal{P}_a(p_{i,j}), P' \in \mathcal{P}(s_{i-1,a})}.
		\end{align*}
		We partition $\mathcal{A}_i$ into two sets: the clauses in $\mathcal{T}_i$ where $b = j$, and the clauses in $\mathcal{J}_i = \mathcal{A}_i \setminus \mathcal{T}_i$.

		Now, as in the proof of \cref{claim:main}, we use $\mathcal{T}_i$ along with some clauses in $F$ to deduce $\mathcal{I}_{i-1}$, and we again will exploit the reversibility of RevRes to do so.
		Namely, starting from $\mathcal{I}_{i-1}$ we deduce $\mathcal{T}_i \cup \mathcal{F}_i$, where $\mathcal{F}_i$ is a collection of (weakenings of) clauses from $F$, and we can then just run the proof in reverse.

		Let $D$ be any clause in $\mathcal{I}_{i-1}$, and note that there is a $j \in [L]$ such that $D = C_P$ for some $P \in \mathcal{P}_{(j,a)}(T_{i-1,a})$.
		Starting from $D$, apply \cref{lem:dt-weakening} with the decision tree $T_{i, j}$ to obtain a collection of clauses of the form $D \vee C_{P'}$ where $P' \in \mathcal{P}(T_{i,j})$.
		Let $(k, b)$ be the output of the decision tree $T_{i,j}$ on the path $P'$.
		If $b = j$, then the clause $D \vee C_{P'}$ belongs to $\mathcal{T}_i$. 
		Moreover, if we repeat this argument for all $D \in \mathcal{I}_{i-1}$ then the collection of all such clauses obtained is \emph{exactly} $\mathcal{T}_i$.
		This is because from $\mathcal{I}_{i}$, the collection of clauses $\mathcal{T}_i$ was obtained by starting from all clauses at leaves of $T_{i, j}$ labelled with $(k, j)$ and then querying $p_{i, j}$ and $s_{i-1, a}$; here, we have performed the exact same queries except we have reversed the order in which we simulated the decision trees $p_{i,j}$ and $s_{i-1,a}$.

		On the other hand, if $b \neq j$, then the literals queried on the paths $P \in \mathcal{P}_{(j,a)}(T_{i-1,a})$ and $P' \in \mathcal{P}_{(k,b)}(T_{i,j})$ together witness that the node $(i, j)$ is a proper sink node, and thus is a solution to the $\sopl$ problem.
		Therefore, at the end of the path $P \cup P'$ we can run the decision tree $g_{i, j}$ to determine a solution to $S(F)$.
		This means that if $P'' \in \mathcal{P}(g_{i,j})$ is any root-to-leaf path in $g_{i,j}$, then the clause $D \vee C_{P'} \vee C_{P''} = C_P \vee C_{P'} \vee C_{P''}$ must be a weakening of a clause in $F$ (or, again, is trivially true).
		Let $\mathcal{F}_i$ denote the set of all of these weakenings of clauses of $F$, obtained by running the above procedure for every clause $D \in \mathcal{I}_{i-1}$.
		We have therefore shown that from $\mathcal{I}_{i-1}$ we can derive $\mathcal{T}_i \cup \mathcal{F}_i$.

		To finish the proof of the claim, we start with the clauses in $\mathcal{I}_{i} \cup \mathcal{F}_i$, deduce $\mathcal{T}_i \cup \mathcal{J}_i$ from $\mathcal{I}_i$ to obtain the clauses $\mathcal{T}_i \cup \mathcal{J}_i \cup \mathcal{F}_i$, and then deduce $\mathcal{I}_{i-1}$ from $\mathcal{T}_i \cup \mathcal{F}_i$.
		This yields the clauses $\mathcal{I}_{i-1} \cup \mathcal{J}_i$, and all of these steps required size $L^{O(1)}2^{O(d)}$ and width at most $O(d)$, completing the proof of the claim and the theorem.
	\end{proof}

	The above proof can be modified to capture $\eopl$ in the same manner as the proof of \cref{thm:sopl-in-revres}.
	In particular, we can argue via the same techniques that the ``junk'' clauses in $\mathcal{J}_i$ and the clauses in $\mathcal{I}_1 \setminus I_{1,1}$ each encode violations of the ``no proper source'' constraints of $\eopl$, and thus can be used to deduce weakenings of clauses in $F$ by querying the appropriate solution decision trees $g_{i,j}$.
	We omit the details.
\end{proof}

\section{Intersection Theorems}
\label{sec:intersection}

We can now finally prove \cref{thm:intersection-thm}, our intersection theorem for Reversible Resolution.
To prove the theorem we use the collapse theorems $\SOPL = \PLS \cap \PPADS$ and $\EOPL = \PLS \cap \PPAD$~\cite{GoosHJMPRT22-collapses}.
In particular, examining the proofs of the collapse theorems from \cite{GoosHJMPRT22-collapses}, we can extract the following black-box analogues.

\begin{theorem}
	Let $R \subseteq \B^n \times O$ be a total search problem, and suppose that there is a depth-$d_1$, $\sod_{s_1}$-formulation of $R$ and a depth-$d_2$, $\sol_{s_2}$-formulation of $R$.
	Then there is a depth $O(d)$ $\sopl_{s^3}$-formulation of $R$ where $d = \max \set{d_1, d_2}$ and $s = \max \set{s_1, s_2}$.
	\qed
\end{theorem}

\begin{theorem}
	Let $R \subseteq \B^n \times O$ be a total search problem, and suppose that there is a depth-$d_1$, $\sod_{s_1}$-formulation of $R$ and a depth-$d_2$, $\eol_{s_2}$-formulation of $R$.
	Then there is a depth $O(d)$ $\eopl_{s^3}$-formulation of $R$ where $d = \max \set{d_1, d_2}$ and $s = \max \set{s_1, s_2}$.
	\qed
\end{theorem}

\cref{thm:intersection-thm} is now an immediate corollary of the next theorem.
\begin{theorem}
	\label{thm:intersection-formal}
	Let $F$ be an unsatisfiable CNF formula.
	Let $d_1, d_2, s_1, s_2$ be positive integers and let $d = \max \set{d_1, d_2}$ and $s = \max \set{s_1, s_2}$.
	\begin{itemize}
		\item If there is a width-$d_1$, size-$s_1$ Resolution proof and a degree-$d_2$, size-$s_2$ unary Sherali--Adams proof of $F$ then there is width $O(d)$ and size $s^{O(1)}2^{O(d)}$ RevRes proof of $F$.
		\item If there is a width-$d_1$, size-$s_1$ Resolution proof and a degree-$d_2$, size-$s_2$ unary Nullstellensatz proof of $F$ then there is width $O(d)$ and size $s^{O(1)}2^{O(d)}$ RevResT proof of $F$.
	\end{itemize}
	In particular, $\textup{RevRes}(F) = \Theta(\textup{Res}(F) + \textup{uSA}(F))$ and $\textup{RevResT}(F) = \Theta(\textup{Res(F)} + \textup{uSA}(F))$.
\end{theorem}
\begin{proof}
	Since RevRes can be efficiently simulated by both Resolution and unary Sherali--Adams we have $\textup{Res}(F) = O(\textup{RevRes}(F))$ and $\textup{uSA}(F) = O(\textup{RevRes}(F))$.
	For the converse direction, suppose that we have a width-$d_1$, size-$s_1$ Resolution proof and a degree-$d_2$, size-$s_2$ unary Sherali--Adams proof.
	By \cite[Theorem 8.18]{Kamath2020} there is a depth-$O(d_1)$ $\sod_{O(s_1)}$-formulation of $S(F)$ and by \cref{thm:unarysa-ppads}, there is a depth-$O(d_2)$ $\eol_{O(s_2)}$-formulation for $S(F)$.
	Applying the above collapse theorem, this implies that there is a depth-$O(d)$ $\sopl_{s^3}$-formulation of $S(F)$, where $d = \max \set{d_1, d_2}$ and $s = \max \set{s_1, s_2}$.
	Finally, applying \cref{thm:revres-sopl}, we obtain a RevRes proof of $F$ with width $O(d)$ and size $s^{O(1)}2^{O(d)}$.
	We therefore have
	\[
		\textup{RevRes}(F)
		= O(d + \log s)
		= O(d_1 + d_2 + \log s_1 + \log s_2)
		= O(\textup{Res}(F) + \textup{uSA}(F)). 
	\]
	A similar proof using \cref{thm:unaryns-ppad} instead yields the characterisation of $\textup{RevResT}$.
\end{proof}

\section{Two Further Separations}

In this section we prove \cref{thm:pls-ppp,thm:eopl-ueopl}, restated below.
\PLSPPP*
\EOPLUEOPL*

The proofs of these theorems rely on a ``glueing'' technique that was implicitly used in~\cite{Beame1998} and which we make more explicit in this paper. We use the glueing technique as a tool to alleviate the lack of good proof systems characterizing $\PPP$ and $\UEOPL$. In particular, The glueing technique reduces the separation in \cref{thm:pls-ppp} to the easier separation $\PLS^{dt}\not\subseteq\PPADS^{dt}$, which we already proved in \cref{cor:pls-ppads} and \cref{thm:eopl-ueopl} uses the glueing technique together with a query lower bound for~$\eopl$ from \cite{Hubacek2020}. This glueing technique was also recently generalized by Jain, Li, Robere and Xun~\cite{jain2024pigeonhole} to prove lower bounds for classes above \PPP corresponding to the generalized pigeonhole principles.

\subsection{Glueability}

Let $\textsc{R} = (R_n)$, $R_n\subseteq\{0,1\}^n\times O_n$, be a $\TFNP^{dt}$ problem. We consider \emph{partial assignments} $x\in\T^n$ that define partial inputs to $R_n$. An index $i$ with $x_i=*$ is interpreted as a boolean variable whose value is not yet assigned. The \emph{size} of a partial assignment is its number of non-$*$ bits. We say that two partial assignments $x, y \in \T^n$ are \emph{consistent} if $x$ and $y$ agree on their non-$*$ bits. If $x$ and $y$ are consistent, we can form the partial assignment $x\cup y$ that assigns values to all variables assigned values in $x$ or $y$. We further say that $x$ is \emph{witnessing} if there exists some solution $o \in O_n$ such that for any $y \in \B^n$ consistent with $x$ we have $o \in R_n(y)$.

\begin{definition}[Glueable sets of assignments]
A set of partial assignments $P\subseteq\{0,1,*\}^n$ is \emph{$k$-glueable} if for each non-witnessing and consistent $p, p' \in P$, their union $p \cup p'$ is non-witnessing, and moreover, if we restrict $R_n$ by the assignment $p\cup p'$, the resulting search problem $(R_n\upharpoonright p\cup p')$ has decision tree complexity greater than $k$.
\end{definition}

This and following definitions are mostly motivated by their use in \cref{lem:PPP_and_glueable_implies_PPADS} and \cref{lem:ueopl-fp}. For instance, in \cref{lem:PPP_and_glueable_implies_PPADS} we consider $P$ to be the set of all partial assignments obtained by collecting leaves pointing to a particular hole in the $\PPP^{dt}$-reduction. The main idea is that the glueability property of $P$ then allows to disambiguate between pigeons to find which (if any) is mapping to the particular hole.

\begin{definition}[Completions]
Let $x \in \T^n$ be a partial assignment and $T$ a decision tree over $\B^n$. The completion $C(T, x)$ of $x$ by $T$ is the set obtained by collecting all the partial assignments corresponding to leaves of $T$ that are consistent with $x$ and taking their union with $x$. That is, $C(T, x) \coloneqq \{x \cup p: \text{$p$ is a leaf of $T$ consistent with $x$}\}$.
\end{definition}

\begin{definition}[Glueable problem]
Let $f\colon \mathbb{N} \to \mathbb{N}$ be a function. We say $\textsc{R}$ is \emph{$f(k)$-glueable} if any set $P\subseteq\T^n$ of partial assignments of size at most $k$, where $k\leq \text{poly}(\log n)$, can be completed by decision trees of depth at most $f(k)$ such that the union of the completions is $k$-glueable. That is, if there exists for each $x \in P$, some decision tree $T_x$ such that $\cup_{x\in P} C(T_x, x)$ is $k$-glueable. We further say that $\textsc{R}$ is \emph{glueable} if it is $\text{poly}(k)$-glueable.
\end{definition}

For example, it is implicit in \cite[\S3.1]{Beame1998} that the $\PPA^{dt}$-complete problem $\lonely$ (given a matching of an odd number of nodes, find an isolated node) is $O(k)$-glueable. In the case of $\lonely$, if $x \in \T^n$ asserts that node $u$ points to node $v$, then $T_x$ queries the pointing node for $v$ so that a solution is immediately witnessed if $u$ is isolated. We will shortly prove that $\sod$ and $\eopl$ are glueable, too. In what follows, we slightly depart from the above notation and also consider pointer-like partial assignments (as opposed to assignments over $\T$ only). Those are treated naturally; for instance, we can assume that reductions are constrained to query either all or no bits corresponding to a pointer.

\subsection{$\PLS^{dt} \not\subseteq \PPP^{dt}$}
\label{sec:pls-ppp}

We introduce for convenience the $\reversiblepigeonlong$ problem, which is a variant of $\pigeon$ where a reverse pointer is provided for each hole.

\begin{description}
\item[$\reversiblepigeonlong$ ($\reversiblepigeon_n$).]
This problem is the same as $\pigeon$ except that we are also given \emph{reverse} pointers $p_u \in [n] \cup \{\nul\}$ for each hole $u \in [n-1]$. The goal is to output any solution of $\pigeon$ or 
\begin{enumerate}
	\item[$2.$] \label{rpigeon2} $u \in [n]$ such that $p_{s_u} \neq u$. \hfill \emph{(successor/predecessor mismatch)}
\end{enumerate}
\end{description}
This problem is known to be $\PPADS^{dt}$-complete (see, e.g., \cite[Lemma~1]{GoosHJMPRT22-collapses}) so that $\reversiblepigeon^{dt} = \PPADS^{dt}$. The following key lemma is implicit in~\cite[\S3.1]{Beame1998}.

\begin{lemma} \label{lem:PPP_and_glueable_implies_PPADS}
If $\textsc{R} \in \PPP^{dt}$ and $\textsc{R}$ is glueable, then $\textsc{R} \in \PPADS^{dt}$.
\end{lemma}
\begin{proof}
Fix a $\pigeon_m$-formulation $(f_i, g_{i,i'})_{i,i' \in [m]}$ of $R_n$ that witnesses $\textsc{R} \in \PPP^{dt}$ and let $(T_i, S_{i,i'})$ be decision trees of depth $k = \text{poly}(\log n)$ implementing this reduction. Since $\textsc{R}$ is glueable, it is possible to complete the root-to-leaf paths of each $T_i$ to get a reduction $(T_i', S_{i,i'})$ of depth $d = \text{poly}(\log n)$ for which the set $P = \cup_{i \in [m]} \leaves(T'_i)$ is $k$-glueable. (Note that each $S_{i,i'}$ remains unchanged and has depth at most $k$.) We show how to construct decision trees $(H_j)_{j \in [m]}$ of depth $\leq d^2$ that compute reverse pointers for each hole of the $\pigeon_m$ instance. We start with the following claim.
\begin{claim}\label{clm:inconsistent} 
Suppose $p\in\leaves(T'_i)$ and $p'\in\leaves(T'_{i'})$ are distinct leaves that are both non-witnessing and labelled with the same hole. Then $p$ and $p'$ are inconsistent.
\end{claim}
\begin{proof}
If $i=i'$, then the claim is true since any two distinct leaves of the same tree are inconsistent. Suppose $i\neq i'$ and suppose for contradiction that $p$ and $p'$ are consistent. Then, $(i, i')$ is a valid solution to the $\pigeon_m$ instance $(T'_1(z), T'_2(z), \dots, T'_m(z))$ for any $z \in \B^n$ extending $p \cup p'$. By correctness of the reduction, this further implies that $S_{i, i'}$ can solve $(R_n\upharpoonright p \cup p')$ with at most $k$ queries---but this contradicts the fact that $P$ is $k$-glueable.
\end{proof}
Let us write $P_j \subseteq P$ for the set of all non-witnessing partial assignment corresponding to leaves labelled with hole $j$. The predecessor tree $H_j$ computes as follows. Pick an arbitrary leaf~$p \in P_j$ and query all the variables contained in $p$. At every leaf $x \in \T^n$ of the current version of $H_j$, the next step depends on the set of $x$-consistent assignments $P_j^x = \{p \in P_j: p \text{ consistent with } x\}$.
\begin{enumerate}
    \item If $|P_j^x| = 0$, then output label $\nul$.
    \item If $|P_j^x| = 1$, then output the \emph{unique} $i \in [m]$ (by \cref{clm:inconsistent}) such that $P^x_j \cap \text{leaves}(T'_i) \neq \emptyset$.
    \item If $|P_j^x| \geq 2$, pick an arbitrary $p \in P_j^x$ and recurse by  querying its variables, etc.
\end{enumerate}
Note that each predecessor tree $H_j$ has depth at most $d^2$: by pairwise inconsistency of $P_j$, at most $d$ paths are queried each of depth at most $d$. To complete the $\reversiblepigeon_m$-formulation of $R_n$, it remains to specify decision trees $(S_i)_{i \in [m]}$ that transform $\reversiblepigeon_m$-solutions of type~(\hyperref[rpigeon2]{2}) into $R_n$-solutions. Indeed, suppose $T'_i(z)=j$ but $H_j(z)\neq i$ for some input $z$ to $R_n$. Then, since $H_j$ decides unambiguously which non-witnessing assignment in $P_j$ is consistent with $z$ (if any), it must be the case that the leaf outputting $T'_i(z)=j$ is not in $P_j$, which means that it is witnessing. Thus, $S_i(z)$ simply runs $T'_i(z)$ and an $R_n$-solution must be witnessed during its execution.
\end{proof}

We note that the method used to disambiguate pigeons in \cref{lem:PPP_and_glueable_implies_PPADS} is common. For instance, it is key to prove the folklore certificate-to-query result $\queryD(f) \leq \queryNP^1(f) \cdot \queryNP^0(f)$ for boolean functions $f$. To show \cref{thm:pls-ppp}, the last missing piece is to show that $\sod$ is glueable. Indeed, if $\sod \in \PPP^{dt}$, then \cref{lem:PPP_and_glueable_implies_PPADS} would imply that $\PLS^{dt} \subseteq \PPADS^{dt}$, which contradicts \cref{cor:pls-ppads}. We show that $\sod$ is glueable in \cref{lem:sod_glueable} below.

For technical convenience, we consider here a minor variation of how we encode the successor pointers in the input to $\sod$. We let the input consist of successor pointers $s_u \in [n]$ for each grid node $u \in [n] \times [n]$ as well as an ``active'' bit $a_u \in \B$, where $a_u=0$ means that $u$ has a $\nul$ pointer. This is merely a different way to encode $\nul$ successors, and indeed, there is a trivial reduction to and from the original $\sod$ problem. The advantage of this new encoding is that it allows for querying the activity of a node \emph{without} querying its successor. This simplifies the completion process in the proof below.

\begin{lemma}
$\sod$ is glueable.
\label{lem:sod_glueable}
\end{lemma}
\begin{proof}
We show that $\sod_n$ is $O(k)$-glueable. Fix some partial $\sod_n$-assignment $x = (s, a)$ of size $k = \text{poly}(\log n)$, that is, $s_u \in [n] \cup \{*\}$ and $a_u \in \T$ for each grid node $u \in [n] \times [n]$. The decision tree $T$ completing $x$ starts by checking whether $x$ queries any active node below row $n-k-1$. If yes, $T$ picks any one such active node and follows the successor path until a sink is found, making the completion witnessing. Note that this step incurs at most $O(k)$ queries. Finally, $T$ ensures that any successor query in $x$ is followed by a query to the active bit of the successor. This costs at most $O(k)$ further queries.

Let $P$ be an arbitrary set of partial assignments each of size at most $k = \text{poly}(\log n)$ and let $P'$ be its completion with respect to the procedure defined above. We first show that $P'$ is $k$-gluable. Pick any two non-witnessing and consistent $p, p' \in P'$ and suppose toward contradiction that their union $p \cup p'$ is witnessing. If it reveals a $\sod$ solution $u$ of type (\hyperref[sod1]{1}) or (\hyperref[sod2]{2}), then it must be that one of $p$ and $p'$ checks for the active bit of $u$: a contradiction with the fact that $p$ and $p'$ are non-witnessing. On the other hand, if $p \cup p'$ reveals a solution $u$ of type (\hyperref[sod3]{3}), then it must be that one of $p$ and $p'$ checks for the successor $s_u$ of $u$, but the completion $T$ forces this check to be followed by a query to the active bit of $s_u$, making one of the initial partial assignments witnessing as well. Hence $p\cup p'$ is non-witnessing.

We finally argue that $(R_n\upharpoonright p\cup p')$ has query complexity greater than $k$ by describing an adversary that can fool any further $k$ queries to $p \cup p'$ without witnessing a solution. Recall that $p \cup p'$ makes no queries to nodes below row $n-k-1$. The adversary answers queries as follows. If the successor pointer of an active node is queried, then we answer with a pointer to any unqueried node on the next row and make it active (there always exists one as $k \ll n$). If a node $u$ is queried that is not the successor of any node, we make $u$ inactive ($a_u = 0$ and $s_u$ is arbitrary). This scheme ensures that a solution can only lie on the very last row $n$, which is not reachable in $k$ queries starting from row $n - k - 1$.
\end{proof}

\subsection{$\EOPL^{dt} \not\subseteq \UEOPL^{dt}$}

We prove \cref{thm:eopl-ueopl} using a similar plan as in \cref{sec:pls-ppp} above. Namely, we first show (\cref{lem:ueopl-fp}) that if we have a problem $\textsc{R} \in \UEOPL^{dt}$ that is glueable, then in fact $\textsc{R} \in \FP^{dt}$, that is, $R_n$ admits a shallow decision tree solving it. Second, we show (\cref{lem:eopl-glue}) that $\eopl_n$ is glueable. The combination of these two lemmas implies that if $\EOPL^{dt}\subseteq\UEOPL^{dt}$, then $\EOPL^{dt}=\FP^{dt}$. But it is known from prior work~\cite{Hubacek2020} (building on~\cite{Aldous1983, Zhang2009}) that $\EOPL^{dt}\neq \FP^{dt}$. This proves \cref{thm:eopl-ueopl}.

It remains to prove \cref{lem:ueopl-fp,lem:eopl-glue}.

\begin{lemma}
\label{lem:ueopl-fp}
If $\textsc{R} \in \UEOPL^{dt}$ and $\textsc{R}$ is glueable, then $\textsc{R} \in \FP^{dt}$.
\end{lemma}
\begin{proof}
Fix an $\ueopl_m$-formulation $(f_u, g_{u,u'})_{u,u' \in [m] \times [m]}$ of $R_n$ that witnesses $\textsc{R} \in \UEOPL^{dt}$ and let $(T_u, S_{u,u'})$ be decision trees of depth $k = \text{poly}(\log n)$ implementing this reduction. Note that the leaves of each $T_u$ are labelled by a successor and predecessor pointers in $[m]$. At the cost of doubling the depth of each $T_u$, we may assume that each leaf is additionally labelled with an ``activity'' bit, which can be computed by appending to each leaf labelled with successor $v$ the decision tree $T_v$. Since $\textsc{R}$ is glueable, it is possible to further complete the leaves of each $T_u$ to get a reduction $(T_u', S_{u,u'})$ of depth $d = \text{poly}(\log n)$ for which the set of leaves $P = \cup_{u \in [m] \times [m]} \leaves(T'_u)$ is $k$-glueable and each leaf label carries the aforementioned activity bit. Let us say that a node $u$ is \emph{good} for input $z$ if the leaf reached by $T_u'(z)$ is non-witnessing and $u$ is active.
\begin{claim}
For every input $z$, there is at most one good node on each row.
\end{claim}
\begin{proof}
Fix a row $j \in [m]$ and suppose for the sake of contradiction that the $j$-th row contains two good nodes $u'$ and $u$ on some input $z$. Let $p\in\leaves(T'_u)$ and $p'\in\leaves(T'_{u'})$ be the leaves reached on input $z$. Then $p$ and $p'$ are a pair of non-witnessing and consistent assignments. Thus, $(u,u')$ is a solution to $\ueopl_m$ on any input that extends $p\cup p'$. Hence the depth-$k$ decision tree $S_{u,u'}$ solves the search problem $(R_n\upharpoonright p\cup p')$. But this contradicts the fact that $P$ is $k$-glueable.
\end{proof}
Using this claim similarly as in the proof of \cref{lem:PPP_and_glueable_implies_PPADS}, we can construct, for each row $j \in [m]$, a decision tree $A_j$ of depth $\leq d^2$ that computes the column-index of a good node on row $j$ or outputs $\nul$ if the row contains no good node. The main argument is again the disambiguation trick.

We can now design an efficient decision tree for $R_n$: At the cost of running $O(\log m) = \text{poly}(\log n)$ of the $A_j$ trees, perform a binary search over the $m$ rows to find a good node $u$ on row $j$ such that the next row $j+1$ contains no good nodes. This means that either (i) the successor of $u$ is inactive, in which case we have found a solution to $\ueopl_m$ and we can use the $S$-trees to find a solution to~$R_n$, or (ii) the successor $u'$ of $u$ is active and $T'_{u'}(z)$ is witnessing, which solves $R_n$.
\end{proof}

We next show that an $\EOPL^{dt}$-complete problem is glueable. Instead of working with $\eopl_n$, it is convenient again to vary the input encoding. We define $\eopl^*$ as a version of $\eopl$ where in addition to successor/predecessor pointers, we are also given an ``activity'' bit.

\begin{description}
\item[$\eoplLong^*$ ($\eopl^*_n$).]
In addition to predecessor/successor pointers, each $u\in[n]\times[n]$ has an activity indicator bit $a_u \in \B$. We add the following solutions to $\eopl$:
	\begin{enumerate}
		\item[$5$.] \label{it:eopl5}
		$u$, if $u$'s activity does not match $a_u$. \hfill \emph{(active bit mismatch)}
		\item[$6$.] \label{it:eopl6}
		$u$, if $a_v=0$ and ($s_v\neq \nul$ or $p_v\neq \nul$).
		\hfill \emph{(inactive node with a pointer)}
		\item[$7$.] \label{it:eopl7}
		$u$, if $a_v=1$ and ($s_v= \nul$ or $p_v= \nul$).
		\hfill \emph{(active node with a $\nul$-pointer)}
	\end{enumerate}
\end{description}
Note that $\eopl^*$ is efficiently reducible to and from $\eopl$, so that $\eopl^*$ is $\EOPL^{dt}$-complete.

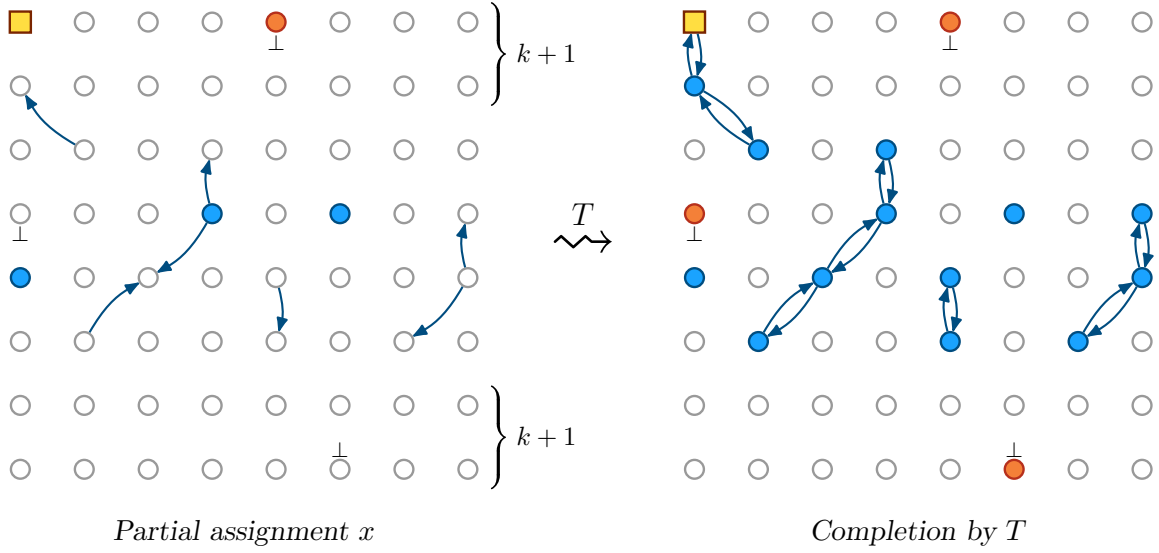
\begin{figure}[t!]
\centering
\begin{tikzpicture}[y=-1cm, scale=0.85]

% node drawing

\tikzstyle{node} = [circle, line width=0.9pt, draw = black!40!white, fill = white, inner sep = 0mm, minimum size = 2.5mm]
\tikzstyle{solution} = [fill=color_gadget_EOPL!90!,draw=black!50!color_gadget_EOPL]
\tikzstyle{node_a} = [node, side, rectangle, minimum size = 2.8mm]
\tikzstyle{node_solution} = [node, solution]

\tikzstyle{node_inact} = [node, fill=Orange,draw=BrickRed]

\tikzstyle{edge} = [-{Latex[round]}, line width=0.8pt, black!50!color_gadget_EOPL]

\tikzstyle{mybend} = [bend left=15]

\tikzstyle{b} = [label={[below,yshift=-5]\scriptsize$\bot$}]
\tikzstyle{bb} = [label={[below,yshift=10]\scriptsize$\bot$}]

{\Huge
\node at (8.8,3.5) {$\leadsto$};}
{\large
\node at (8.8,3) {$T$};}

\draw [very thick,decorate,
    decoration = {calligraphic brace,amplitude=5pt,raise=9pt}]
    (7,-.3) -- (7,1.3);
\node[right] at (7.6,0.5) {\small $k+1$};
\draw [very thick,decorate,
    decoration = {calligraphic brace,amplitude=5pt,raise=9pt}]
    (7,5.7) -- (7,7.3);
\node[right] at (7.6,6.5) {\small $k+1$};

\begin{scope}[xshift=300]

\node[node_a]       (P00) at (0,0) {};
\node[node] (P01) at (1,0) {};
\node[node] (P02) at (2,0) {};
\node[node] (P03) at (3,0) {};
\node[node_inact,b] (P04) at (4,0) {};
\node[node] (P05) at (5,0) {};
\node[node] (P06) at (6,0) {};
\node[node] (P07) at (7,0) {};

\node[node_solution]       (P10) at (0,1) {};
\node[node] (P11) at (1,1) {};
\node[node] (P12) at (2,1) {};
\node[node] (P13) at (3,1) {};
\node[node] (P14) at (4,1) {};
\node[node] (P15) at (5,1) {};
\node[node] (P16) at (6,1) {};
\node[node] (P17) at (7,1) {};

\node[node]       (P20) at (0,2) {};
\node[node_solution] (P21) at (1,2) {};
\node[node] (P22) at (2,2) {};
\node[node_solution] (P23) at (3,2) {};
\node[node] (P24) at (4,2) {};
\node[node] (P25) at (5,2) {};
\node[node] (P26) at (6,2) {};
\node[node] (P27) at (7,2) {};

\node[node_inact,b]       (P30) at (0,3) {};
\node[node] (P31) at (1,3) {};
\node[node] (P32) at (2,3) {};
\node[node_solution] (P33) at (3,3) {};
\node[node] (P34) at (4,3) {};
\node[node_solution] (P35) at (5,3) {};
\node[node] (P36) at (6,3) {};
\node[node_solution] (P37) at (7,3) {};

\node[node_solution] (P40) at (0,4) {};
\node[node] (P41) at (1,4) {};
\node[node_solution] (P42) at (2,4) {};
\node[node] (P43) at (3,4) {};
\node[node_solution] (P44) at (4,4) {};
\node[node] (P45) at (5,4) {};
\node[node] (P46) at (6,4) {};
\node[node_solution] (P47) at (7,4) {};

\node[node]       (P50) at (0,5) {};
\node[node_solution] (P51) at (1,5) {};
\node[node] (P52) at (2,5) {};
\node[node] (P53) at (3,5) {};
\node[node_solution] (P54) at (4,5) {};
\node[node] (P55) at (5,5) {};
\node[node_solution] (P56) at (6,5) {};
\node[node] (P57) at (7,5) {};

\node[node]       (P60) at (0,6) {};
\node[node] (P61) at (1,6) {};
\node[node] (P62) at (2,6) {};
\node[node] (P63) at (3,6) {};
\node[node] (P64) at (4,6) {};
\node[node] (P65) at (5,6) {};
\node[node] (P66) at (6,6) {};
\node[node] (P67) at (7,6) {};

\node[node] (P70) at (0,7) {};
\node[node] (P71) at (1,7) {};
\node[node] (P72) at (2,7) {};
\node[node] (P73) at (3,7) {};
\node[node] (P74) at (4,7) {};
\node[node_inact,bb] (P75) at (5,7) {};
\node[node] (P76) at (6,7) {};
\node[node] (P77) at (7,7) {};

\node at (3.5,8) {\slshape Completion by $T$};

% % edges drawing

\draw[edge] (P00) edge [mybend] (P10);
\draw[edge] (P10) edge [mybend] (P21);
\draw[edge] (P23) edge [mybend] (P33);
\draw[edge] (P33) edge [mybend] (P42);
\draw[edge] (P42) edge [mybend] (P51);
\draw[edge] (P44) edge [mybend] (P54);
\draw[edge] (P37) edge [mybend] (P47);
\draw[edge] (P47) edge [mybend] (P56);

\draw[edge] (P10) edge [mybend] (P00);
\draw[edge] (P21) edge [mybend] (P10);
\draw[edge] (P33) edge [mybend] (P23);
\draw[edge] (P42) edge [mybend] (P33);
\draw[edge] (P51) edge [mybend] (P42);
\draw[edge] (P54) edge [mybend] (P44);
\draw[edge] (P47) edge [mybend] (P37);
\draw[edge] (P56) edge [mybend] (P47);

\end{scope}

\begin{scope}

\node[node_a]       (P00) at (0,0) {};
\node[node] (P01) at (1,0) {};
\node[node] (P02) at (2,0) {};
\node[node] (P03) at (3,0) {};
\node[node_inact,b] (P04) at (4,0) {};
\node[node] (P05) at (5,0) {};
\node[node] (P06) at (6,0) {};
\node[node] (P07) at (7,0) {};

\node[node] (P10) at (0,1) {};
\node[node] (P11) at (1,1) {};
\node[node] (P12) at (2,1) {};
\node[node] (P13) at (3,1) {};
\node[node] (P14) at (4,1) {};
\node[node] (P15) at (5,1) {};
\node[node] (P16) at (6,1) {};
\node[node] (P17) at (7,1) {};

\node[node] (P20) at (0,2) {};
\node[node] (P21) at (1,2) {};
\node[node] (P22) at (2,2) {};
\node[node] (P23) at (3,2) {};
\node[node] (P24) at (4,2) {};
\node[node] (P25) at (5,2) {};
\node[node] (P26) at (6,2) {};
\node[node] (P27) at (7,2) {};

\node[node,b]       (P30) at (0,3) {};
\node[node] (P31) at (1,3) {};
\node[node] (P32) at (2,3) {};
\node[node_solution] (P33) at (3,3) {};
\node[node] (P34) at (4,3) {};
\node[node_solution] (P35) at (5,3) {};
\node[node] (P36) at (6,3) {};
\node[node] (P37) at (7,3) {};

\node[node_solution] (P40) at (0,4) {};
\node[node] (P41) at (1,4) {};
\node[node] (P42) at (2,4) {};
\node[node] (P43) at (3,4) {};
\node[node] (P44) at (4,4) {};
\node[node] (P45) at (5,4) {};
\node[node] (P46) at (6,4) {};
\node[node] (P47) at (7,4) {};

\node[node]       (P50) at (0,5) {};
\node[node] (P51) at (1,5) {};
\node[node] (P52) at (2,5) {};
\node[node] (P53) at (3,5) {};
\node[node] (P54) at (4,5) {};
\node[node] (P55) at (5,5) {};
\node[node] (P56) at (6,5) {};
\node[node] (P57) at (7,5) {};

\node[node]       (P60) at (0,6) {};
\node[node] (P61) at (1,6) {};
\node[node] (P62) at (2,6) {};
\node[node] (P63) at (3,6) {};
\node[node] (P64) at (4,6) {};
\node[node] (P65) at (5,6) {};
\node[node] (P66) at (6,6) {};
\node[node] (P67) at (7,6) {};

\node[node] (P70) at (0,7) {};
\node[node] (P71) at (1,7) {};
\node[node] (P72) at (2,7) {};
\node[node] (P73) at (3,7) {};
\node[node] (P74) at (4,7) {};
\node[node,bb] (P75) at (5,7) {};
\node[node] (P76) at (6,7) {};
\node[node] (P77) at (7,7) {};

\node at (3.5,8) {\slshape Partial assignment $x$};

% % edges drawing

\draw[edge] (P21) edge [mybend] (P10);

\draw[edge] (P33) edge [mybend] (P42);
\draw[edge] (P44) edge [mybend] (P54);
\draw[edge] (P47) edge [mybend] (P56);

\draw[edge] (P33) edge [mybend] (P23);
\draw[edge] (P51) edge [mybend] (P42);
\draw[edge] (P47) edge [mybend] (P37);

\end{scope}

\end{tikzpicture}
\vspace{2mm}
\caption{Example of a non-witnessing completion. A node is blue if $a_u=1$, orange if $a_u=0$, and white if~$a_u=*$ is not queried. The symbol $\bot$ indicates a $\nul$ pointer. The first and last $k+1$ rows contain no active nodes, besides those lying on the path starting at the distinguished node $(1,1)$.}
\label{fig:glueable}
\end{figure}

\begin{lemma} \label{lem:eopl-glue}
$\eopl^*$ is glueable.
\end{lemma}
\begin{proof}
We show that $\eopl^*$ is $O(k)$-glueable. Fix some partial $\eopl^*$-assignment $x = (p, s, a)$ of size $k = \text{poly}(\log n)$. The tree $T$ that completes $x$ proceeds as follows. We start by querying all variables assigned in $x$. Then we iterate each of the following steps until a solution is found or no further queries are made.
\begin{enumerate}[itemsep=0.5em]
    \item \emph{Always query activity bits and reverse pointers.}
    If we have queried a $\nul$-pointer $s_v=\nul$ or $p_v=\nul$, then we also query the activity bit $a_v$. This activity bit is $a_v=0$ unless we have found a solution of type (\hyperref[it:eopl6]{7}).
    
    Moreover, if we have queried a non-$\nul$ pointer $s_u = v$ (resp.~$p_u=v$), then we also query the bits $a_u$, $a_v$ and the pointer $p_v$ (resp.~$s_v$). Note that both activity bits must be~$1$ and the reverse pointer must point back, $p_v=u$ (resp.~$s_v=u$), as otherwise we can find a solution by making a couple more queries. Indeed, if $a_u=0$, then there is a solution of type~(\hyperref[it:eopl6]{6}). If~$a_u=1$ and $a_v=0$, then we can find a solution by determining the activity of $v$: either $v$ is active, which is a mismatch with $a_v=0$ (type (\hyperref[it:eopl5]{5})), or $v$ is inactive, which creates a sink. Finally, if~$a_u=1$ and $p_v\neq u$, then $u$ is inactive, which is a mismatch with $a_u=1$ (type (\hyperref[it:eopl5]{5})).
    
    \item \label{it:path} \emph{Follow the distinguished path.}
    Follow the successor path starting at the distinguished source node $(1, 1)$ until some node on row $k + 1$ is reached or a sink is found.
    \item \emph{Follow early paths.}
    If we have queried $a_u=1$ for some node $u$ in the first $k+1$ rows that does not lie on the path discovered in \cref{it:path}, then we follow $u$'s predecessor path until a solution is found.
    \item \emph{Follow late paths.}
    If we have queried $a_u=1$ for some node $u$ in the last $k+1$ rows, then we follow $u$'s successor path until a solution is found.
\end{enumerate}
This completion adds at most $O(k)$ queries to $x$. An example of a completion that is non-witnessing is given in \cref{fig:glueable}. It is straightforward to argue that the resulting set of completed assignments is $k$-glueable using an adversary strategy similar to the one described in the proof of \cref{lem:sod_glueable}.
\end{proof}

\appendix

\section{Appendix: Coefficient Size in Algebraic Proofs} \label{sec:coeff-ub}

In this appendix we show that if there are low-degree Nullstellensatz and Sherali--Adams refutations over $\mathbb{Z}$, then the coefficients in the refutations will also be not too large in magnitude.
In particular, if the degree of the proofs are $d$, the the magnitude of the coefficients can be assumed to be at most~$\exp(n^{O(d)})$ without loss of generality.
For Sherali--Adams this follows easily as any Sherali--Adams refutation over the reals can be converted into a Sherali--Adams refutation over $\mathbb{Z}$ without badly affecting the coefficient size.

\begin{theorem}
	Let $F$ be an unsatisfiable CNF formula on $n$ variables and $m$ clauses.
	If there is a degree-$d$ Sherali--Adams refutation of $F$ then there is a degree-$d$ Sherali--Adams refutation of $F$ over $\Z$ where every coefficient is bounded in magnitude by $\exp(n^{O(d)})$.
\end{theorem}
\begin{proof}
	This is essentially the usual proof of completeness for Sherali--Adams (see e.g.~\cite{Fleming2019}).
	Consider a degree-$d$ Sherali--Adams refutation of $F$ which, by \cref{lem:sa-normal-form}, we can write as \[ \sum_{i=1}^m -J_i\overline C_i + J = -1.\]
	We can express the existence of such a proof as a system of integer linear inequalities of the form $Ax = b, x \geq 0$ over $mn^{O(d)} = n^{O(d)}$ variables and over $n^{O(d)}$ constraints where all coefficients of the matrix $A$ and $b$ are in $\set{1,0,-1}$, and indeed $b$ has a single non-zero entry with value $-1$ (see \cite[Chapter 2]{Fleming2019} for an explicit description of the system).
	By known results on linear programming this implies that the coefficients of the above Sherali--Adams refutation can be assumed to be rational with description length $n^{O(d)}$.
	Let $L$ be the least common multiple of the denominators all rational numbers occurring in the refutation.
	By multiplying through by $L$ we obtain the identity \[ \sum_{i=1}^m - LJ_i \overline C_i + LJ = -L.\]
	We can then add the integer $L - 1$ to both sides (noting that $LJ + L-1$ is a conical junta) to obtain an integer-coefficient Sherali--Adams refutation with the desired coefficient bound.
\end{proof}

For Nullstellensatz the proof is slightly different as we need to recruit known bounds for integer solutions to systems of linear equations.
\begin{theorem}
	Let $F$ be an unsatisfiable CNF formula on $n$ variables and $m$ clauses.
	If there is a degree-$d$ Nullstellensatz refutation of $F$ over $\Z$, then there is a degree-$d$ Nullstellensatz refutation over $\Z$ where every coefficient has magnitude at most $\exp({n^{O(d)}})$.
\end{theorem}
\begin{proof}
	This follows the standard proof of completeness for Nullstellensatz proofs (see e.g.~\cite{Pitassi1996, Buss1998}).
	Write $F = C_1 \land \cdots \land C_m$ and suppose $F$ has $n$ variables.
	A degree-$d$ $\Z$-Nullstellensatz proof of $F$ can be written as \[ \sum_{i=1}^m q_i\overline C_i = 1\]
	for some integer-coefficient multilinear polynomials $q_i$.
	We can express the existence of such a proof as a system of $\mathbb{Z}$-linear equations $Ax = b$ over $mn^{O(d)}$ variables --- roughly one variable for each monomial $m$ of degree at most $d$ --- where each coefficient in $A$ and $b$ is small.
	The result then follows by the known strongly-polynomial time algorithms for finding integer solutions to systems of linear equations over $\Z$ (in particular, via the \emph{Hermite Normal Form} \cite{Kannan1979}).

	The system of linear equations is defined as follows.
	For each $i \in [m]$ and $S \subseteq [n]$ with $|S| \leq d$ we let $\hat q_i(S) \in \Z$ denote the coefficient of the monomial $x_S = \prod_{i \in S} x_i$ in the polynomial $q_i$.	
	Letting $C_{n, d}$ denote all subsets of $[n]$ of size $d$, we can write the Nullstellensatz refutation as \[ \sum_{i=1}^m \sum_{S \in C_{n,d}} \hat q_i(S) x_S \overline C_i = 1.\]
	From this, we get a system of $\mathbb{Z}$-linear equations over variables $\hat q_i(S)$ for each $i \in [m]$, $S \in C_{n,d}$ enforcing that all monomials in the proof of degree $d > 1$ must cancel out to $0$, and the monomials of degree $0$ must sum to $1$.
	The system of equations has one constraint for each monomial $x_S$ with $S \in C_{n,d}$ and at most $m |C_{n,d}| \leq mn^{O(d)}$ variables; each coefficient in the system of linear equations is $\pm 1$ from the expansion of $x_S \overline C_i$ into a sum of monomials.
	By reducing to Hermite Normal Form we can find an integer solution to this system with coefficients of size at most $\exp(n^{O(d)})$.
\end{proof}

Finally, we can consider RevRes proofs.
An obvious fact is that any Resolution proof with width $w$ has $n^{O(w)}$ distinct clauses without loss of generality.
However, this result fails for RevRes, since we can no longer reuse clauses an unlimited number of times.
By combining the previous results with the intersection theorem (\cref{thm:intersection-formal}), one can also immediately deduce the following result that gives a weak bound on the size of RevRes and RevResT proofs with bounded width.
We omit the proof.
\begin{corollary}
	Let $F$ be an unsatisfiable CNF formula on $n$ variables.
	If there is a width-$d$ RevRes refutation of $F$ (RevResT resp.) then there is a width-$O(d)$ and size $\exp(n^{O(d)})$ RevRes refutation of $F$ (RevResT resp.).
	\qed
\end{corollary}

\bigskip\bigskip
\subsection*{Acknowledgements}

We thank Albert Atserias, Ilario Bonacina, Pritish Kamath, and David Steurer for discussions, and the anonymous reviewers for their suggestions that helped us improve the presentation of the paper. M.G., A.H., S.J., and G.M.\ were supported by the Swiss State Secretariat for Education, Research and Innovation (SERI) under contract number MB22.00026. S.J.\ did part of the work while being supported by the Quantum Systems Accelerator through DOE. W.P., R.R., and R.T. were supported by NSERC.

\bigskip\bigskip
\pagebreak

% -----------------------------------------------
\DeclareUrlCommand{\Doi}{\urlstyle{sf}}
\renewcommand{\path}[1]{\small\Doi{#1}}
\renewcommand{\url}[1]{\href{#1}{\small\Doi{#1}}}
\bibliographystyle{alphaurl}
\bibliography{tfnp-refs}

\newcommand{\etalchar}[1]{$^{#1}$}
\begin{thebibliography}{dRNMR19}

\bibitem[AL19]{Atserias2019}
Albert Atserias and Massimo Lauria.
\newblock Circular (yet sound) proofs.
\newblock In {\em Proceedings of the 22nd Theory and Applications of
  Satisfiability Testing (SAT)}, pages 1--18. Springer, 2019.
\newblock \href {https://doi.org/10.1007/978-3-030-24258-9_1}
  {\path{doi:10.1007/978-3-030-24258-9_1}}.

\bibitem[Ald83]{Aldous1983}
David Aldous.
\newblock Minimization algorithms and random walk on the d-cube.
\newblock {\em The Annals of Probability}, 11(2):403--413, 1983.
\newblock URL: \url{http://www.jstor.org/stable/2243696}.

\bibitem[ALN16]{Atserias2016}
Albert Atserias, Massimo Lauria, and Jakob Nordstr\"{o}m.
\newblock Narrow proofs may be maximally long.
\newblock {\em ACM Transactions on Computational Logic}, 17(3):1--30, 2016.
\newblock \href {https://doi.org/10.1145/2898435} {\path{doi:10.1145/2898435}}.

\bibitem[BB22]{Bonacina2022a}
Ilario Bonacina and Maria~Luisa Bonet.
\newblock On the strength of {S}herali-{A}dams and {N}ullstellensatz as
  propositional proof systems.
\newblock In {\em Proceedings of the 37th Symposium on Logic in Computer
  Science (LICS)}. {ACM}, aug 2022.
\newblock \href {https://doi.org/10.1145/3531130.3533344}
  {\path{doi:10.1145/3531130.3533344}}.

\bibitem[BCE{\etalchar{+}}98]{Beame1998}
Paul Beame, Stephen Cook, Jeff Edmonds, Russell Impagliazzo, and Toniann
  Pitassi.
\newblock The relative complexity of {NP} search problems.
\newblock {\em Journal of Computer and System Sciences}, 57(1):3--19, 1998.
\newblock \href {https://doi.org/10.1006/jcss.1998.1575}
  {\path{doi:10.1006/jcss.1998.1575}}.

\bibitem[BCIP02]{BureshOppenheim2002}
Joshua Buresh{-}Oppenheim, Matthew Clegg, Russell Impagliazzo, and Toniann
  Pitassi.
\newblock Homogenization and the polynomial calculus.
\newblock {\em Computational Complexity}, 11(3-4):91--108, 2002.
\newblock \href {https://doi.org/10.1007/s00037-002-0171-6}
  {\path{doi:10.1007/s00037-002-0171-6}}.

\bibitem[Ben09]{BenSasson2009}
Eli Ben{-}Sasson.
\newblock Size-space tradeoffs for resolution.
\newblock {\em {SIAM} Journal on Computing}, 38(6):2511--2525, 2009.
\newblock \href {https://doi.org/10.1137/080723880}
  {\path{doi:10.1137/080723880}}.

\bibitem[BFI22]{BussFI2022}
Sam Buss, Noah Fleming, and Russell Impagliazzo.
\newblock {TFNP} characterizations of proof systems and monotone circuits,
  2022.
\newblock URL: \url{https://eccc.weizmann.ac.il/report/2022/141/}.

\bibitem[BIK{\etalchar{+}}94]{Beame1994}
Paul Beame, Russell Impagliazzo, Jan Kraj{\'{\i}}{\v{c}}ek, Toniann Pitassi,
  and Pavel Pudl{\'a}k.
\newblock Lower bounds on {H}ilbert's {N}ullstellensatz and propositional
  proofs.
\newblock In {\em Proceedings of the 35th Symposium on Foundations of Computer
  Science (FOCS)}, pages 794--806, 1994.
\newblock \href {https://doi.org/10.1109/SFCS.1994.365714}
  {\path{doi:10.1109/SFCS.1994.365714}}.

\bibitem[BKT14]{BussKT2014}
Samuel Buss, Leszek~Aleksander Kołodziejczyk, and Neil Thapen.
\newblock Fragments of approximate counting.
\newblock {\em The Journal of Symbolic Logic}, 79(2):496--525, 2014.
\newblock URL: \url{http://www.jstor.org/stable/43303745}.

\bibitem[BLM07]{Bonet2007}
Mar{\'{\i}}a~Luisa Bonet, Jordi Levy, and Felip Many{\`{a}}.
\newblock Resolution for {M}ax-{SAT}.
\newblock {\em Artificial Intelligence}, 171(8-9):606--618, 2007.
\newblock \href {https://doi.org/10.1016/j.artint.2007.03.001}
  {\path{doi:10.1016/j.artint.2007.03.001}}.

\bibitem[BM04]{Buresh2004}
Joshua Buresh{-}Oppenheim and Tsuyoshi Morioka.
\newblock Relativized {NP} search problems and propositional proof systems.
\newblock In {\em Proceedings of the 19th IEEE Conference on Computational
  Complexity (CCC)}, pages 54--67, 2004.
\newblock \href {https://doi.org/10.1109/CCC.2004.1313795}
  {\path{doi:10.1109/CCC.2004.1313795}}.

\bibitem[BR98]{Beame1998a}
Paul Beame and S{\o}ren Riis.
\newblock More on the relative strength of counting principles.
\newblock In {\em Proceedings of the DIMACS Workshop on Proof Complexity and
  Feasible Arithmetics}, volume~39, pages 13--35, 1998.

\bibitem[BT22]{Bonacina2022}
Ilario Bonacina and Neil Thapen.
\newblock A separation of {PLS} from {PPP}.
\newblock Technical report, Electronic Colloquium on Computational Complexity
  (ECCC), 2022.
\newblock URL: \url{https://eccc.weizmann.ac.il/report/2022/089/}.

\bibitem[Bus98]{Buss1998}
Samuel Buss.
\newblock Lower bounds on {N}ullstellensatz proofs via designs.
\newblock In {\em Proof Complexity and Feasible Arithmetics}, pages 59--71.
  AMS, 1998.

\bibitem[CDDT09]{ChenDDT09-Arrow-Debreu}
Xi~Chen, Decheng Dai, Ye~Du, and Shang{-}Hua Teng.
\newblock Settling the complexity of {A}rrow-{D}ebreu equilibria in markets
  with additively separable utilities.
\newblock In {\em Proceedings of the 50th Symposium on Foundations of Computer
  Science (FOCS)}, pages 273--282, 2009.
\newblock \href {https://doi.org/10.1109/FOCS.2009.29}
  {\path{doi:10.1109/FOCS.2009.29}}.

\bibitem[CDO15]{ChenDO15-anonymous-games}
Xi~Chen, David Durfee, and Anthi Orfanou.
\newblock On the complexity of {N}ash equilibria in anonymous games.
\newblock In {\em Proceedings of the 47th Symposium on Theory of Computing
  (STOC)}, pages 381--390, 2015.
\newblock \href {https://doi.org/10.1145/2746539.2746571}
  {\path{doi:10.1145/2746539.2746571}}.

\bibitem[CDT09]{ChenDT09-Nash}
Xi~Chen, Xiaotie Deng, and Shang-Hua Teng.
\newblock Settling the complexity of computing two-player {N}ash equilibria.
\newblock {\em Journal of the ACM}, 56(3):14:1--14:57, 2009.
\newblock \href {https://doi.org/10.1145/1516512.1516516}
  {\path{doi:10.1145/1516512.1516516}}.

\bibitem[CEI96]{Clegg1996}
Matthew Clegg, Jeff Edmonds, and Russell Impagliazzo.
\newblock Using the {G}roebner basis algorithm to find proofs of
  unsatisfiability.
\newblock In {\em Proceedings of the 28th Symposium on Theory of Computing
  (STOC)}, pages 174--183, 1996.
\newblock \href {https://doi.org/10.1145/237814.237860}
  {\path{doi:10.1145/237814.237860}}.

\bibitem[CPY17]{ChenPY17-non-monotone-markets}
Xi~Chen, Dimitris Paparas, and Mihalis Yannakakis.
\newblock The complexity of non-monotone markets.
\newblock {\em Journal of the ACM}, 64(3):20:1--20:56, 2017.
\newblock \href {https://doi.org/10.1145/3064810} {\path{doi:10.1145/3064810}}.

\bibitem[CR79]{Cook1979}
Stephen Cook and Robert Reckhow.
\newblock The relative efficiency of propositional proof systems.
\newblock {\em Journal of Symbolic Logic}, 44(1):36--50, 1979.
\newblock \href {https://doi.org/10.2307/2273702} {\path{doi:10.2307/2273702}}.

\bibitem[CSVY08]{CodenottiSVY08-economies-games}
Bruno Codenotti, Amin Saberi, Kasturi Varadarajan, and Yinyu Ye.
\newblock The complexity of equilibria: Hardness results for economies via a
  correspondence with games.
\newblock {\em Theoretical Computer Science}, 408(2--3):188--198, 2008.
\newblock \href {https://doi.org/10.1016/j.tcs.2008.08.007}
  {\path{doi:10.1016/j.tcs.2008.08.007}}.

\bibitem[Das19]{Daskalakis2019}
Constantinos Daskalakis.
\newblock Equilibria, fixed points, and computational complexity.
\newblock In {\em Proceedings of the International Congress of Mathematicians
  (ICM)}. World Scientific, 2019.
\newblock \href {https://doi.org/10.1142/9789813272880_0009}
  {\path{doi:10.1142/9789813272880_0009}}.

\bibitem[DGP09]{Daskalakis2009}
Constantinos Daskalakis, Paul Goldberg, and Christos Papadimitriou.
\newblock The complexity of computing a {N}ash equilibrium.
\newblock {\em SIAM Journal on Computing}, 39(1):195--259, 2009.
\newblock \href {https://doi.org/10.1137/070699652}
  {\path{doi:10.1137/070699652}}.

\bibitem[DM12]{Dantchev2012}
Stefan Dantchev and Barnaby Martin.
\newblock Rank complexity gap for {L}ov{\'{a}}sz-{S}chrijver and
  {S}herali-{A}dams proof systems.
\newblock {\em computational complexity}, 22(1):191--213, nov 2012.
\newblock \href {https://doi.org/10.1007/s00037-012-0049-1}
  {\path{doi:10.1007/s00037-012-0049-1}}.

\bibitem[DMR09]{Dantchev2009}
Stefan Dantchev, Barnaby Martin, and Mark Rhodes.
\newblock Tight rank lower bounds for the {S}herali--{A}dams proof system.
\newblock {\em Theoretical Computer Science}, 410(21-23):2054--2063, 2009.
\newblock \href {https://doi.org/10.1016/j.tcs.2009.01.002}
  {\path{doi:10.1016/j.tcs.2009.01.002}}.

\bibitem[DP11]{Daskalakis2011}
Constantinos Daskalakis and Christos Papadimitriou.
\newblock Continuous local search.
\newblock In {\em Proceedings of the 22nd Symposium on Discrete Algorithms
  (SODA)}, pages 790--804. SIAM, 2011.
\newblock \href {https://doi.org/10.1137/1.9781611973082.62}
  {\path{doi:10.1137/1.9781611973082.62}}.

\bibitem[DQS12]{DengQS12-cake}
Xiaotie Deng, Qi~Qi, and Amin Saberi.
\newblock Algorithmic solutions for envy-free cake cutting.
\newblock {\em Operations Research}, 60(6):1461--1476, 2012.
\newblock \href {https://doi.org/10.1287/opre.1120.1116}
  {\path{doi:10.1287/opre.1120.1116}}.

\bibitem[dRGR22]{Rezende2022}
Susanna de~Rezende, Mika G\"{o}\"{o}s, and Robert Robere.
\newblock Proofs, circuits, and communication.
\newblock {\em SIGACT News}, 53(1), 2022.
\newblock \href {https://doi.org/10.1145/3532737.3532745}
  {\path{doi:10.1145/3532737.3532745}}.

\bibitem[dRNMR19]{Rezende2019}
Susanna de~Rezende, Jakob Nordstr{\"o}m, Or~Meir, and Robert Robere.
\newblock Nullstellensatz size-degree trade-offs from reversible pebbling.
\newblock In Amir Shpilka, editor, {\em Proceedings of the 34th Computational
  Complexity Conference (CCC)}, volume 137, pages 18:1--18:16. Schloss
  Dagstuhl, 2019.
\newblock \href {https://doi.org/10.4230/LIPIcs.CCC.2019.18}
  {\path{doi:10.4230/LIPIcs.CCC.2019.18}}.

\bibitem[FG22]{FRG22-NS-CH-ham}
Aris Filos{-}Ratsikas and Paul Goldberg.
\newblock The complexity of necklace splitting, consensus-halving, and discrete
  ham sandwich.
\newblock {\em {SIAM} Journal on Computing}, 2022.
\newblock (to appear).
\newblock \href {https://doi.org/10.1137/20m1312678}
  {\path{doi:10.1137/20m1312678}}.

\bibitem[FGGR22]{Fleming2022}
Noah Fleming, Mika G\"{o}\"{o}s, Stefan Grosser, and Robert Robere.
\newblock On semi-algebraic proofs and algorithms.
\newblock In {\em Proceedings of the 13th Innovations in Theoretical Computer
  Science Conference (ITCS)}, volume 215 of {\em Leibniz International
  Proceedings in Informatics (LIPIcs)}, pages 69:1--69:25. Schloss Dagstuhl,
  2022.
\newblock \href {https://doi.org/10.4230/LIPIcs.ITCS.2022.69}
  {\path{doi:10.4230/LIPIcs.ITCS.2022.69}}.

\bibitem[FGHS21]{Fearnley2021}
John Fearnley, Paul~W. Goldberg, Alexandros Hollender, and Rahul Savani.
\newblock The complexity of gradient descent: {CLS} $=$ {PPAD} $\cap$ {PLS}.
\newblock In {\em Proceedings of the 53rd Symposium on Theory of Computing
  (STOC)}, pages 46--59, 2021.
\newblock \href {https://doi.org/10.1145/3406325.3451052}
  {\path{doi:10.1145/3406325.3451052}}.

\bibitem[FGMS20]{Fearnley2020}
John Fearnley, Spencer Gordon, Ruta Mehta, and Rahul Savani.
\newblock Unique end of potential line.
\newblock {\em Journal of Computer and System Sciences}, 114:1--35, 2020.
\newblock \href {https://doi.org/10.1016/j.jcss.2020.05.007}
  {\path{doi:10.1016/j.jcss.2020.05.007}}.

\bibitem[FKP19]{Fleming2019}
Noah Fleming, Pravesh Kothari, and Toniann Pitassi.
\newblock Semialgebraic proofs and efficient algorithm design.
\newblock {\em Foundations and Trends in Theoretical Computer Science},
  14(1-2):1--221, 2019.
\newblock \href {https://doi.org/10.1561/0400000086}
  {\path{doi:10.1561/0400000086}}.

\bibitem[FMSV23]{Filmus2023}
Yuval Filmus, Meena Mahajan, Gaurav Sood, and Marc Vinyals.
\newblock Maxsat resolution and subcube sums.
\newblock {\em {ACM} Trans. Comput. Log.}, 24(1):8:1--8:27, 2023.
\newblock \href {https://doi.org/10.1145/3565363} {\path{doi:10.1145/3565363}}.

\bibitem[FPT04]{FabrikantPT04-pure-Nash}
Alex Fabrikant, Christos Papadimitriou, and Kunal Talwar.
\newblock The complexity of pure {N}ash equilibria.
\newblock In {\em Proceedings of the 36th ACM Symposium on Theory of Computing
  (STOC)}, pages 604--612, 2004.
\newblock \href {https://doi.org/10.1145/1007352.1007445}
  {\path{doi:10.1145/1007352.1007445}}.

\bibitem[GHJ{\etalchar{+}}22]{GoosHJMPRT22-collapses}
Mika G{\"o}{\"o}s, Alexandros Hollender, Siddhartha Jain, Gilbert Maystre,
  William Pires, Robert Robere, and Ran Tao.
\newblock Further collapses in {TFNP}.
\newblock In {\em Proceedings of the 37th Computational Complexity Conference
  (CCC)}, pages 33:1--33:15, 2022.
\newblock \href {https://doi.org/10.4230/LIPICS.CCC.2022.33}
  {\path{doi:10.4230/LIPICS.CCC.2022.33}}.

\bibitem[GKRS18]{Goos2018}
Mika G{\"o}{\"o}s, Pritish Kamath, Robert Robere, and Dmitry Sokolov.
\newblock Adventures in monotone complexity and {TFNP}.
\newblock In {\em Proceedings of the 10th Innovations in Theoretical Computer
  Science Conference (ITCS)}, volume 124, pages 38:1--38:19, 2018.
\newblock \href {https://doi.org/10.4230/LIPIcs.ITCS.2019.38}
  {\path{doi:10.4230/LIPIcs.ITCS.2019.38}}.

\bibitem[GP18]{Goos2018cbs}
Mika G{\"o}{\"o}s and Toniann Pitassi.
\newblock Communication lower bounds via critical block sensitivity.
\newblock {\em SIAM Journal on Computing}, 47(5):1778--1806, 2018.
\newblock \href {https://doi.org/10.1137/16M1082007}
  {\path{doi:10.1137/16M1082007}}.

\bibitem[Hak21]{Hakoniemi2021}
Tuomas Hakoniemi.
\newblock Monomial size vs. bit-complexity in sums-of-squares and polynomial
  calculus.
\newblock In {\em Proceedings of the 36th Symposium on Logic in Computer
  Science (LICS)}. IEEE, 2021.
\newblock \href {https://doi.org/10.1109/lics52264.2021.9470545}
  {\path{doi:10.1109/lics52264.2021.9470545}}.

\bibitem[HKT24]{Hubacek2024}
Pavel Hub\'{a}\v{c}ek, Erfan Khaniki, and Neil Thapen.
\newblock {TFNP Intersections Through the Lens of Feasible Disjunction}.
\newblock In Venkatesan Guruswami, editor, {\em 15th Innovations in Theoretical
  Computer Science Conference (ITCS 2024)}, volume 287 of {\em Leibniz
  International Proceedings in Informatics (LIPIcs)}, pages 63:1--63:24,
  Dagstuhl, Germany, 2024. Schloss Dagstuhl -- Leibniz-Zentrum f{\"u}r
  Informatik.
\newblock URL:
  \url{https://drops.dagstuhl.de/entities/document/10.4230/LIPIcs.ITCS.2024.63},
  \href {https://doi.org/10.4230/LIPIcs.ITCS.2024.63}
  {\path{doi:10.4230/LIPIcs.ITCS.2024.63}}.

\bibitem[HN12]{Huynh2012}
Trinh Huynh and Jakob Nordstr{\"o}m.
\newblock On the virtue of succinct proofs: Amplifying communication complexity
  hardness to time--space trade-offs in proof complexity.
\newblock In {\em Proceedings of the 44th Symposium on Theory of Computing
  (STOC)}, pages 233--248. ACM, 2012.
\newblock \href {https://doi.org/10.1145/2213977.2214000}
  {\path{doi:10.1145/2213977.2214000}}.

\bibitem[HY20]{Hubacek2020}
Pavel Hub{\'{a}}{\v{c}}ek and Eylon Yogev.
\newblock Hardness of continuous local search: Query complexity and
  cryptographic lower bounds.
\newblock {\em {SIAM} Journal on Computing}, 49(6):1128--1172, 2020.
\newblock \href {https://doi.org/10.1137/17m1118014}
  {\path{doi:10.1137/17m1118014}}.

\bibitem[IR21]{Itsykson2021}
Dmitry Itsykson and Artur Riazanov.
\newblock Proof complexity of natural formulas via communication arguments.
\newblock In {\em Proceedings of the 36th Computational Complexity Conference
  (CCC)}, volume 200, pages 3:1--3:34. Schloss Dagstuhl, 2021.
\newblock \href {https://doi.org/10.4230/LIPIcs.CCC.2021.3}
  {\path{doi:10.4230/LIPIcs.CCC.2021.3}}.

\bibitem[JLRX24]{jain2024pigeonhole}
Siddhartha Jain, Jiawei Li, Robert Robere, and Zhiyang Xun.
\newblock On {P}igeonhole {P}rinciples and {R}amsey in {TFNP}, 2024.
\newblock \href {https://arxiv.org/abs/2401.12604} {\path{arXiv:2401.12604}}.

\bibitem[JPY88]{Johnson1988}
David Johnson, Christos Papadimitriou, and Mihalis Yannakakis.
\newblock How easy is local search?
\newblock {\em Journal of Computer and System Sciences}, 37(1):79--100, 1988.
\newblock \href {https://doi.org/10.1016/0022-0000(88)90046-3}
  {\path{doi:10.1016/0022-0000(88)90046-3}}.

\bibitem[Juk12]{Jukna2012}
Stasys Jukna.
\newblock {\em Boolean Function Complexity: Advances and Frontiers}, volume~27
  of {\em Algorithms and Combinatorics}.
\newblock Springer, 2012.

\bibitem[Kam20]{Kamath2020}
Pritish Kamath.
\newblock {\em Some hardness escalation results in computational complexity
  theory}.
\newblock PhD thesis, Massachusetts Institute of Technology, 2020.
\newblock URL: \url{https://dspace.mit.edu/handle/1721.1/128290}.

\bibitem[KB79]{Kannan1979}
Ravindran Kannan and Achim Bachem.
\newblock Polynomial algorithms for computing the smith and hermite normal
  forms of an integer matrix.
\newblock {\em SIAM Journal on Computing}, 8(4):499--507, 1979.
\newblock \href {https://doi.org/10.1137/0208040} {\path{doi:10.1137/0208040}}.

\bibitem[Kra19]{Krajicek2019}
Jan Kraj{\'{\i}}{\v{c}}ek.
\newblock {\em Proof Complexity}.
\newblock Cambridge University Press, 2019.

\bibitem[Kre89]{Krentel89-TSP}
Mark Krentel.
\newblock Structure in locally optimal solutions.
\newblock In {\em Proceedings of the 30th Symposium on Foundations of Computer
  Science (FOCS)}, pages 216--221, 1989.
\newblock \href {https://doi.org/10.1109/SFCS.1989.63481}
  {\path{doi:10.1109/SFCS.1989.63481}}.

\bibitem[Kre90]{Krentel90-weighted-CNF}
Mark Krentel.
\newblock On finding and verifying locally optimal solutions.
\newblock {\em SIAM Journal on Computing}, 19(4):742--749, 1990.
\newblock \href {https://doi.org/10.1137/0219052} {\path{doi:10.1137/0219052}}.

\bibitem[LHdG08]{Larrosa2008}
Javier Larrosa, Federico Heras, and Simon de~Givry.
\newblock A logical approach to efficient {M}ax-{SAT} solving.
\newblock {\em Artificial Intelligence}, 172(2-3):204--233, 2008.
\newblock \href {https://doi.org/10.1016/j.artint.2007.05.006}
  {\path{doi:10.1016/j.artint.2007.05.006}}.

\bibitem[LNNW95]{Lovasz1995}
L{\'a}szl{\'o} Lov{\'a}sz, Moni Naor, Ilan Newman, and Avi Wigderson.
\newblock Search problems in the decision tree model.
\newblock {\em SIAM Journal on Discrete Mathematics}, 8(1):119--132, 1995.
\newblock \href {https://doi.org/10.1137/S0895480192233867}
  {\path{doi:10.1137/S0895480192233867}}.

\bibitem[LPR24]{Li2024}
Yuhao Li, William Pires, and Robert Robere.
\newblock {Intersection Classes in TFNP and Proof Complexity}.
\newblock In Venkatesan Guruswami, editor, {\em 15th Innovations in Theoretical
  Computer Science Conference (ITCS 2024)}, volume 287 of {\em Leibniz
  International Proceedings in Informatics (LIPIcs)}, pages 74:1--74:22,
  Dagstuhl, Germany, 2024. Schloss Dagstuhl -- Leibniz-Zentrum f{\"u}r
  Informatik.
\newblock URL:
  \url{https://drops.dagstuhl.de/entities/document/10.4230/LIPIcs.ITCS.2024.74},
  \href {https://doi.org/10.4230/LIPIcs.ITCS.2024.74}
  {\path{doi:10.4230/LIPIcs.ITCS.2024.74}}.

\bibitem[Meh18]{Mehta2018-constant-rank}
Ruta Mehta.
\newblock Constant rank two-player games are {PPAD}-hard.
\newblock {\em SIAM Journal on Computing}, 47(5):1858--1887, January 2018.
\newblock \href {https://doi.org/10.1137/15m1032338}
  {\path{doi:10.1137/15m1032338}}.

\bibitem[Mor01]{Morioka2001}
Tsuyoshi Morioka.
\newblock Classification of search problems and their definability in bounded
  arithmetic.
\newblock Master's thesis, University of Toronto, 2001.
\newblock URL:
  \url{https://www.collectionscanada.ca/obj/s4/f2/dsk3/ftp04/MQ58775.pdf}.

\bibitem[MP91]{Megiddo1991}
Nimrod Megiddo and Christos Papadimitriou.
\newblock On total functions, existence theorems and computational complexity.
\newblock {\em Theoretical Computer Science}, 81(2):317--324, 1991.
\newblock \href {https://doi.org/10.1016/0304-3975(91)90200-L}
  {\path{doi:10.1016/0304-3975(91)90200-L}}.

\bibitem[NS94]{Nisan1994}
Noam Nisan and Mario Szegedy.
\newblock On the degree of boolean functions as real polynomials.
\newblock {\em Computational Complexity}, 4(4):301--313, dec 1994.
\newblock \href {https://doi.org/10.1007/bf01263419}
  {\path{doi:10.1007/bf01263419}}.

\bibitem[O'D17]{ODonnell2017}
Ryan O'Donnell.
\newblock {SOS} is not obviously automatizable, even approximately.
\newblock In {\em Proceedings of the 8th Innovations in Theoretical Computer
  Science Conference (ITCS)}, volume~67, pages 59:1--59:10. Schloss Dagstuhl,
  2017.
\newblock \href {https://doi.org/10.4230/LIPIcs.ITCS.2017.59}
  {\path{doi:10.4230/LIPIcs.ITCS.2017.59}}.

\bibitem[Pap94]{Papadimitriou1994}
Christos Papadimitriou.
\newblock On the complexity of the parity argument and other inefficient proofs
  of existence.
\newblock {\em Journal of Computer and System Sciences}, 48(3):498--532, 1994.
\newblock \href {https://doi.org/10.1016/s0022-0000(05)80063-7}
  {\path{doi:10.1016/s0022-0000(05)80063-7}}.

\bibitem[Pit96]{Pitassi1996}
Toniann Pitassi.
\newblock Algebraic propositional proof systems.
\newblock In {\em Descriptive Complexity and Finite Models, Proceedings of a
  {DIMACS} Workshop 1996}, volume~31 of {\em {DIMACS} Series in Discrete
  Mathematics and Theoretical Computer Science}, pages 215--244. {DIMACS/AMS},
  1996.
\newblock \href {https://doi.org/10.1090/dimacs/031/07}
  {\path{doi:10.1090/dimacs/031/07}}.

\bibitem[Pud15]{Pudlak15-Herbrand}
Pavel Pudl{\'a}k.
\newblock On the complexity of finding falsifying assignments for {H}erbrand
  disjunctions.
\newblock {\em Archive for Mathematical Logic}, 54(7-8):769--783, 2015.
\newblock \href {https://doi.org/10.1007/s00153-015-0439-6}
  {\path{doi:10.1007/s00153-015-0439-6}}.

\bibitem[RW92]{Raz1992}
Ran Raz and Avi Wigderson.
\newblock Monotone circuits for matching require linear depth.
\newblock {\em Journal of the ACM}, 39(3):736–744, jul 1992.
\newblock \href {https://doi.org/10.1145/146637.146684}
  {\path{doi:10.1145/146637.146684}}.

\bibitem[RW17]{Raghavendra2017}
Prasad Raghavendra and Benjamin Weitz.
\newblock On the bit complexity of sum-of-squares proofs.
\newblock In {\em Proceedings of the 44th International Colloquium on Automata,
  Languages, and Programming (ICALP)}, pages 80:1--80:13, 2017.
\newblock \href {https://doi.org/10.4230/LIPIcs.ICALP.2017.80}
  {\path{doi:10.4230/LIPIcs.ICALP.2017.80}}.

\bibitem[SA94]{Sherali1994}
Hanif Sherali and Warren Adams.
\newblock A hierarchy of relaxations and convex hull characterizations for
  mixed-integer zero--one programming problems.
\newblock {\em Discrete Applied Mathematics}, 52(1):83--106, jul 1994.
\newblock \href {https://doi.org/10.1016/0166-218x(92)00190-w}
  {\path{doi:10.1016/0166-218x(92)00190-w}}.

\bibitem[Sch91]{Schaeffer91-local-search}
Alejandro Sch{\"a}ffer.
\newblock Simple local search problems that are hard to solve.
\newblock {\em SIAM Journal on Computing}, 20(1):56--87, 1991.
\newblock \href {https://doi.org/10.1137/0220004} {\path{doi:10.1137/0220004}}.

\bibitem[SZZ18]{SotirakiZZ18-PPP}
Katerina Sotiraki, Manolis Zampetakis, and Giorgos Zirdelis.
\newblock {PPP}-completeness with connections to cryptography.
\newblock In {\em Proceedings of the 59th IEEE Symposium on Foundations of
  Computer Science (FOCS)}, pages 148--158, 2018.
\newblock \href {https://doi.org/10.1109/FOCS.2018.00023}
  {\path{doi:10.1109/FOCS.2018.00023}}.

\bibitem[Zha09]{Zhang2009}
Shengyu Zhang.
\newblock Tight bounds for randomized and quantum local search.
\newblock {\em SIAM Journal on Computing}, 39(3):948--977, 2009.
\newblock \href {https://doi.org/10.1137/06066775X}
  {\path{doi:10.1137/06066775X}}.

\end{thebibliography}

\end{document}